\documentclass[twoside,a4paper,11pt]{book}
\usepackage{amsmath}
\usepackage{amsthm}
\usepackage{graphicx}
\usepackage[usenames,dvipsnames]{xcolor}
\usepackage{bbm}
\usepackage{fancyhdr}
\usepackage{parskip} 
\usepackage[english]{babel} 
\usepackage[titletoc]{appendix} 
\usepackage{rotating} 
\usepackage[justification=centering,labelfont=bf,font=small]{caption} 
\usepackage{subcaption} 
\newtheorem{theorem}{Theorem}[chapter]
\usepackage{enumitem}
\newcommand{\plane}[2]{$#1#2$\nobreakdash-plane}
\theoremstyle{definition}

\newtheorem{definition}{Definition}[chapter]

\usepackage{braket}

\def\openone{\leavevmode\hbox{\small1\kern-3.8pt\normalsize1}}
\newcommand{\be}{\begin{equation}}
\newcommand{\ee}{\end{equation}}
\newcommand{\ba}{\begin{eqnarray}}
\newcommand{\ea}{\end{eqnarray}}
\newcommand{\ketbra}[2]{|#1\rangle \langle #2|}
\newcommand{\tr}{\operatorname{Tr}}

\newcommand{\one}{\bf{1}}

\newtheorem{thm}{Theorem}
\newcommand{\half}{\frac{1}{2}}
\newcommand{\qua}{\frac{1}{4}}
\newtheorem{lem}{Lemma}

\newtheorem{observation}{Observation}
\newtheorem{proposition}{Proposition}
\newtheorem{cor}{Corrolory}
\newcommand{\etal}{{\it{et al. }}}

\usepackage{xspace}

\newcommand{\emptypage}{\newpage\null\thispagestyle{plain}\newpage}

\newtheoremstyle{mystyle}
  {\baselineskip}
  {\topsep}
  {\itshape}
  {0pt}
  {\bfseries}
  {.}
  {5pt plus 1pt minus 1pt}
  {}
  \theoremstyle{mystyle}

\usepackage{indentfirst}

\usepackage[T1]{fontenc}
\newcommand{\tfont}{\raggedright\usefont{T1}{qhv}{b}{n}\selectfont} 
\newcommand{\tofont}{\usefont{T1}{qhv}{m}{n}\selectfont} 

\usepackage[charter,expert]{mathdesign}
\usepackage{charter}
\usepackage{dsfont} 

\newcommand{\tspace}{\hspace{0.6em}}

\usepackage{titlesec}
\titleformat{\chapter}[display]{\fontsize{22pt}{1em}\tfont}{{\chaptertitlename} \thechapter\vspace{0.4cm}}{0pt}{\fontsize{30pt}{1em}\tfont}
\titleformat{\section}[block]{\Large\tfont}{\thesection\tspace}{0pt}{\Large\tfont}
\titleformat{\subsection}[block]{\large\tfont}{\thesubsection\tspace}{0pt}{\large\tfont}
\titleformat{\subsubsection}[block]{\tfont}{\thesubsubsection\tspace}{0pt}{\tfont}

\usepackage{tocloft}
\tocloftpagestyle{plain}
\addto\captionsenglish{}

\usepackage[numbers]{natbib}

\definecolor{citec}{RGB}{190,40,5} 
\definecolor{linkc}{RGB}{5,160,40} 
\usepackage[colorlinks=true,linkcolor=linkc,citecolor=citec,urlcolor=blue,linktocpage,backref=page]{hyperref}

\usepackage{etoolbox}
\apptocmd{\sloppy}{\hbadness 10000\relax}{}{} 
\makeatletter
\patchcmd{\BR@backref}{\newblock}{\newblock[}{}{}
\patchcmd{\BR@backref}{\par}{]\par}{}{}
\makeatother

\begin{document}

\pagenumbering{Roman}
\thispagestyle{empty}
\baselineskip=18pt
\begin{center}
{\Large \bf Characterizing quantum correlations in the nonsignaling framework} \\
\vspace*{1cm}
{\large{\bf Thesis}} \\
\vspace*{0.5cm}
{For the award of the degree of}\\
\vspace{0.5cm}
{\large{\bf DOCTOR OF  PHILOSOPHY}} \\
\vspace{0.25cm}
\end{center}
\vspace*{3cm}
\begin{tabular}{lp{6cm}l}
{{\it Supervised by:}}  &&
{{\it Submitted by:}} \\
\\
{\bf Dr. Sudeshna Sinha} &&
{\bf C. Jebarathinam} \\
{\bf } &&\\

\end{tabular}
\begin{center}
\vspace*{1.5cm}
\hspace*{0cm}
\end{center}
\vspace*{-1cm}
\begin{center}
\includegraphics{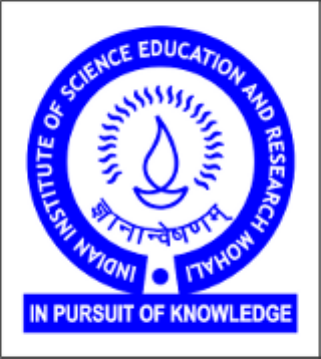}\\
{\bf Indian Institute of Science Education \& Research  Mohali\\
Mohali 140 306\\ India\\ October 2015
 }
\end{center}

\pagenumbering{roman}

\emptypage

\pagenumbering{Roman}
\section*{\centerline{Declaration}}
The work presented in this thesis has been carried
out by me  at the Indian Institute of
Science Education and Research Mohali.

This work has not been submitted in part or in
full for a degree, diploma or a fellowship to
any other University or Institute. Whenever
contributions of others are involved, every effort
has been made to indicate this clearly, with due
acknowledgement of collaborative research
and discussions. This thesis is a bonafide record
of original work done by me and all sources
listed within have been detailed in the
bibliography.
\vspace*{2cm}

\hspace*{-0.25in}
\parbox{2.5in}{
\noindent {\bf C Jebarathinam}
}

\vspace*{0.25in}

\hspace*{-0.25in}
\noindent
\parbox{2.5in}{
\noindent Place~:  \\
\noindent Date~:
}

\vspace*{0.5in}

\noindent
In my capacity as the supervisor of the
candidate's PhD thesis work, I certify that the
above statements by the candidate are true to the
best of my knowledge.

\vspace*{2cm}

\hspace*{-0.25in}
\parbox{2.5in}{
\noindent {\bf } \\
\noindent Professor Sudeshna Sinha\\
\noindent Department of Physical Sciences \\
\noindent IISER Mohali
}
\vspace*{0.25in}

\hspace*{-0.25in}
\noindent
\parbox{2.5in}{
\noindent Place~:  \\
\noindent Date~:
} 

\pagenumbering{roman}

\emptypage

\pagenumbering{Roman}
\thispagestyle{empty}
\centerline{\LARGE \bf List of Publications}
\vspace{1cm}
\vspace*{12pt}
\centerline{\bf \underline{Out of this Thesis}}
\begin{enumerate}
\addtolength{\itemsep}{12pt}
\item {\bf C Jebarathinam} Canonical decomposition of quantum correlations in the framework of generalized nonsignaling theories.
\href{https://arxiv.org/abs/1407.3170v4}{arXiv:1407.3170v4
[quant-ph]}
   \item {\bf C Jebarathinam} Isolating genuine nonclassicality in tripartite quantum correlations
\href{https://arxiv.org/abs/1407.5588v5}{arXiv:1407.5588v5
[quant-ph]} 
  \item {\bf C Jebarathinam} On total correlations in a bipartite quantum joint probability distribution.
   \href{http://arxiv.org/abs/1410.1472}{arXiv:1410.1472
[quant-ph]}

\end{enumerate}

\vspace{1cm}
\centerline{\bf \underline{Other Publications}}
\begin{enumerate}
\item {\bf C Jebaratnam}. Detecting genuine multipartite entanglement in steering scenarios.
\href{http://dx.doi.org/10.1103/PhysRevA.93.052311}{\emph{Phys.
Rev. A}, 93, 052311}.
\end{enumerate}
\thispagestyle{empty}

\pagenumbering{roman}

\emptypage


\thispagestyle{plain}
\section*{\centerline{Acknowledgements}}
I would like to thank Prof. Arvind for the opportunity to do quantum foundations and quantum information theory.
I am thankful to Professor N. Sathyamurthy, Director, IISER Mohali for his encouragement and support.
I am grateful to Professor C. G. Mahajan for his encouragement and support.
I thank IISER Mohali for financial support.

I am thankful to Dr. P. Rungta for suggesting me the problems in quantum foundations and for many inspiring discussions
while carrying out this work.
I thank Dr. R. Srikanth, PPISR Bangalore for discussions and useful suggestions.
I would like to thank Dr. K. P. Yogendran, IISER Mohali for discussions and useful comments.
I would like to thank my doctoral committee members: Prof. Sudeshna Sinha, Prof. Arvind and Dr. Mandip Singh.
I would like to thank Bhupesh Kumar for helping me to make the figures in this thesis.



\emptypage


\thispagestyle{plain}
\section*{\centerline{Abstract}}

Bell nonlocality of quantum theory refers to the nonclassical correlations
obtained by local measurements on spatially separated entangled
subsystems. Bell nonlocality is a resource for device-independent quantum
information processing. Quantum discord was introduced as a measure of
quantum correlations which captures nonclassical correlations in separable
states as well. Recently, it has been shown that non-null quantum discord
is a resource for quantum information processing.

Quantum correlations forms a subset of the set of nonsignaling boxes. This
allows us to characterize quantum correlations as a convex combination of
the extremal boxes of the nonsignaling polytope which are Popescu-Rohrlich
boxes (maximally nonlocal boxes) and local deterministic boxes. There
exists multiple decomposition of quantum correlations in the context of
the nonsignaling polytope. I find that the existence of Popescu-Rohrlich
box decomposition for local boxes associates two notions of discord which
capture nonclassicality of quantum correlations originating from Bell
nonlocality and EPR-steering.

I introduce, Bell and Mermin discord, and show that any bipartite
nonsignaling box admits a three-way decomposition. This decomposition
allows us to isolate the origin of nonclassicality into three disjoint
sources: a Popescu-Rohrlich box, a maximally EPR-steerable box, and a
classical correlation. Interestingly, I show that all non-null quantum
discord states which are neither classical-quantum states nor
quantum-classical states can give rise to nonclassical correlations which
have non-null Bell and/or Mermin discord for suitable incompatible
measurements. I introduce two notions of genuine discord, which are the
generalizations of Bell and Mermin discord to the multipartite scenario,
to characterize the presence of genuine nonclassicality in quantum
correlations.

\emptypage
\tableofcontents
\emptypage
\phantomsection
\addcontentsline{toc}{chapter}{Figures}
\listoffigures
\emptypage

\newpage

\pagenumbering{arabic}

\pagestyle{plain}


\chapter{Introduction}
Quantum theory successfully describes the nature within the domain of
the microscopic world in which classical physics fails to explain.
Quantum theory has many distinguishing features such as uncertainty
due to incompatibility of observables, no-cloning, intrinsic randomness
to name but a few. Unlike the special theory of relativity, the axioms
of the quantum theory are mathematical. There have been attempts to
give physical postulates for quantum theory \cite{JB,NLNS,CAP}.
Quantum theory is consistent with nonsignaling principle; however, it
predicts correlations that are nonlocal in the sense that it violates
a Bell inequality \cite{bell64,BNL}. In an attempt to conjecture that
nonsignaling and nonlocality as axioms for quantum theory, Popescu and
Rohrlich found that there are nonsignaling correlations that are more
nonlocal than quantum theory \cite{PR}. Thus, nonlocality which seems
distinguishing feature of quantum theory is a generic feature of
nonsignaling theories \cite{MAG06}.

In   generalized   nonsignaling   theory  (GNST),   correlations   are
constrained  only by  the nonsignaling  (NS) principle  and thus  GNST
allows  nonlocal  correlations  stronger than  that  allowed  by
quantum theory \cite{Barrett,MAG06}.   It is known that the  set of NS
correlations   forms  a   convex   polytope  known   as  NS   polytope
\cite{Barrett}.
Since  quantum correlations  are  contained in  the  NS polytope,  any
quantum  correlation can  be written  as a  convex combination  of the
extremal boxes of the polytope.  One  of the goals of studying GNST is
to find  out what singles  out quantum theory from  other nonsignaling
theories  \cite{Geb}.   GNST has  also  been  used to  study  nonlocal
correlations, for instance, measures of nonlocality and quantifier for
intrinsic  randomness have  been  proposed in  the  framework of  GNST
\cite{Forsteretal,EPR2B,Dhara}.  Bell-nonlocality,  i.e., the violation
of a Bell  inequality   is  a   resource  for   device-independent
quantum information processing \cite{DIQIP1,DIQIP2}.  Security  of device-independent
quantum key  distribution was  studied in the  context of  NS polytope
\cite{DQKD}.

All pure entangled states give rise to nonlocality \cite{GT}. In the case
mixed states, entanglement and nonlocality are inequivalent \cite{Werner}.
It is natural to consider that entangled states which do not violate a Bell inequality do not have nonclassicality.
However, it was shown that there are mixed entangled states which are useful for teleportation, but do not violate a Bell inequality \cite{TNL}.
Recently, it has been shown that there are mixed separable states
that give rise to advantage for certain quantum information tasks \cite{Dakicetal}; the key resource behind this advantage is believed to be quantum discord \cite{OZ}.
It would be interesting to study nonclassical correlations in nonzero quantum discord states, which include all entangled and separable states,
in the context of the NS polytope.

\section{Nonclassical correlations}
In 1935, Einstein, Podolsky, and Rosen (EPR) argued incompleteness of quantum theory using entangled states and
suggested that quantum theory could be complete if it is supplemented
with additional hidden variables that assume locality and reality \cite{EPR}.
Since then investigations into hidden variable theories were started to account for the predictions of quantum theory \cite{BellRev,MerminRev}.
In Ref. \cite{Bohm}, Bohm presented the EPR argument using two spin-1/2 particles (qubits) in a singlet state given as follows,
\be
\ket{\psi^-}=\frac{1}{\sqrt{2}}(\ket{01}-\ket{10}).
\ee
Consider an experiment in which two spatially separated parties, Alice and Bob, share the singlet state and measure the spin of their qubit along two perpendicular directions.
If Alice measures $\sigma_x$ ($\sigma_y$), she can predict the measurement result of $\sigma_x$ ($\sigma_y$) on Bob's side with certainty.
Thus, element physical reality exists for the measurements of $\sigma_x$ and $\sigma_y$ simultaneously according to the criterion of EPR.
Since the quantum theory does not simultaneously predict the results of any two incompatible measurements with certainty, EPR argued that quantum theory is incomplete.
In 1964, John Bell invalidated the assumptions of EPR (locality and realism) by showing that quantum theory is incompatible with local hidden variable (LHV)
theories \cite{bell64}.
The refutation of hidden variables by quantum theory was first demonstrated by Kochen and Specker \cite{Specker, KS},
they showed that measurement results of spin-1 systems predicted by quantum theory
is incompatible with noncontextual hidden variable (NCHV) theories.
\subsection{Nonlocality}
Bell experiments involve in testing whether the correlation between outcomes of space-like separated measurements exhibits nonlocality or not.
If the violation of a Bell inequality is observed,
then nonlocality of the correlation is demonstrated. In the bipartite Bell scenario, two spatially separated observers, Alice and Bob, receive subsystems
of a correlated composite system and they perform measurements $A$ and $B$ on their respective subsystems which produce outcomes $a$ and $b$. The correlation
between the outcomes is described by the conditional joint probability of getting the outcomes, $P(a,b|A,B)$. Since the measurements are happening at the space-like
separated regions, the correlation satisfies nonsignaling principle, i.e., Alice cannot signal to Bob by her choice of measurement and vice versa.

Bell inequalities are the bounds on the correlations under the constraint of LHV theories. In an LHV theory, there exist some hidden variables $\lambda$ which occur with
probability $p_\lambda$ such that the correlation satisfies the following locality condition,
\be
P(a,b|A,B)=\sum_\lambda p_\lambda P_\lambda(a|A)P_\lambda(b|B). \label{LHVm}
\ee
Suppose $\lambda$ corresponds to different run of the experiment, locality implies that
for each run of the experiment the joint probability for the outcome pair factorizes as the product of
marginals corresponding to Alice and Bob, i.e., $P_\lambda(a,b|A,B)=P_\lambda(a|A)P_\lambda(b|B)$.
Since the correlation that exhibits nonlocality cannot be written in the form given in Eq. (\ref{LHVm}), it violates a Bell inequality.

Suppose the parties generate the correlation by making measurements on a composite quantum system.
Quantum theory associates a quantum state described by the density operator $\rho$
in the Hilbert space $\mathcal{H}_A\otimes\mathcal{H}_B$ and local measurement operators $M^{A}_{a}$ and $M^{B}_{b}$
such that the correlation is predicted by Born's rule as follows,
\be
P(a,b|A,B)=\tr \left(\rho M^{A}_{a} \otimes M^{B}_{b}\right).
\ee
Quantum states come in two distinct types: entangled and separable. Since the separable states can be written as a convex combination of the product states,
\be
\rho=\sum_\lambda p_\lambda \rho^A_\lambda \otimes \rho^B_\lambda,
\ee
the correlations arising from these states satisfy the locality condition in Eq. (\ref{LHVm}). Thus, only entangled states can lead to the
violation of a Bell inequality.
\subsubsection{Bell-CHSH inequality}
The simplest physical situation that exhibits nonlocality is the scenario considered by Clauser etal \cite{chsh}. In Bell-CHSH scenario, Alice and Bob perform
two dichotomic measurements
$A_i$ and $B_j$ on their subsystems and generate outcomes $a_m$ and $b_n$, where $i,j,m,n\in\{0,1\}$.
Quantum correlations corresponding to this scenario can be generated by making spin projective measurements $A_i=\hat{a}_i\cdot \vec{\sigma}$
and $B_j=\hat{b}_j\cdot \vec{\sigma}$ on
an ensemble of two spin-$1/2$ particles (qubits) along the directions $\hat{a}_i$ and $\hat{b}_j$ which generate outcomes $a_m,b_n\in\{-1,+1\}$.

Clauser etal derived the following inequality,
\be
|\braket{A_0B_0}-\braket{A_0B_1}|\le2-|\braket{A_1B_0}+\braket{A_1B_1}|,
\ee
under the constraint that the correlations satisfy the locality condition in Eq. (\ref{LHVm}). This inequality is equivalent to,
\be
\mathcal{B}:=|\braket{A_0B_1}+\braket{A_1B_0}+\braket{A_0B_0}-\braket{A_1B_1}|\le2, \label{CHSH}
\ee
which is the famous CHSH inequality. Suppose Alice and Bob receive two spin-$1/2$ particles in the singlet state, they can generate correlation
which violates the Bell-CHSH inequality in Eq. (\ref{CHSH}). For the singlet state, quantum theory predicts
$\braket{A_iB_j}=\braket{\phi^-|\hat{a}_i\cdot\vec{\sigma}\otimes\hat{b}_j\cdot\vec{\sigma}|\phi^-}=-\hat{a}_i\cdot\hat{b}_j$.
For the following choice of measurement directions: $\hat{a}_0=\hat{x}$, $\hat{a}_1=\hat{y}$, $\hat{b}_0=-\frac{1}{\sqrt{2}}(\hat{x}+\hat{y})$
and $\hat{b}_1=\frac{1}{\sqrt{2}}(-\hat{x}+\hat{y})$, the singlet state gives rise to $\mathcal{B}=2\sqrt{2}>2$.
\subsubsection{Hardy's paradox}
Hardy's test doesn't involve inequalities and is based on logical contradiction with local realism \cite{Hardy1, Hardy2}.
Consider the correlations associated with the Bell-CHSH scenario that satisfy the following three constraints,
\ba
P(+1,+1|A_0,B_0)&=&0 \label{H1}\\
P(+1,-1|A_1,B_0)&=&0 \label{H2}\\
P(-1,+1|A_0,B_1)&=&0 \label{H3}.
\ea
If these correlations can be simulated by the LHV theory, they will satisfy the condition,
\be
P(+1,+1|A_1,B_1)=0. \label{H4}
\ee
We show that the violation of this condition with the constraints on the correlation given in Eqs. (\ref{H1})-(\ref{H3}) implies nonlocality.
Suppose Alice and Bob observe the outcome pair $+1$ and $+1$ for the measurement $A_1B_1$.
Under the assumption of locality, Eq. (\ref{H2}) and Eq. (\ref{H3}) imply that the outcome of Bob for the measurement of $B_0$ is $+1$ and the outcome
of Alice for the measurement of $A_0$ is $+1$. Since in a local realistic theory the measurement of one party should not
depend on the measurement choice of the other party, Alice and Bob must observe the outcome pair $+1$ and $+1$ for the measurement $A_0B_0$, however,
this contradicts Eq. (\ref{H1}).

Hardy showed that the correlations arising from the pure states except the extremal states (product and maximally entangled state)
satisfy the constraints in Eqs. (\ref{H1})-(\ref{H3}) while violating the constraint in Eq. (\ref{H4}) for suitable state dependent measurements.
Suppose Alice and Bob share the pure state $\ket{\psi}=b\ket{01}+c\ket{10}+d\ket{11}$ and make measurements $A_0=\sigma_z$, $A_1=\ketbra{a_+}{a_+}-\ketbra{a_-}{a_-}$,
$B_0=\sigma_z$ and $B_1=\ketbra{b_+}{b_+}-\ketbra{b_-}{b_-}$,
where $\ket{a_+}=\frac{d^*\ket{0}-b^*\ket{1}}{\sqrt{|b|^2+|d|^2}}$, $\ket{a_-}=\frac{b\ket{0}+d\ket{1}}{\sqrt{|b|^2+|d|^2}}$,
$\ket{b_+}=\frac{d^*\ket{0}-c^*\ket{1}}{\sqrt{|c|^2+|d|^2}}$ and $\ket{b_-}=\frac{c\ket{0}+d\ket{1}}{\sqrt{|c|^2+|d|^2}}$ \cite{Goldstein, HTvsCT}.
Then, the correlation satisfies the constraints in Eqs. (\ref{H1})-(\ref{H3}) and violates the condition in Eq. (\ref{H4}) as follows,
\be
P(+1,+1|A_1,B_1)=\frac{|bcd|^2}{(|b|^2+|d|^2)(|c|^2+|d|^2)},
\ee
which implies that the correlation is nonlocal if the state is neither a product state nor a maximally entangled state.
\subsection{Contextuality}
LHV theory is a special case of NCHV theory in that every LHV theory is an NCHV theory; however, the converse is not true. In NCHV theories, locality is replaced
by noncontextuality.
Noncontextuality can be illustrated by the following situation. Suppose
an observable $A$ is compatible with two observables $B$ and $C$, i.e., $[A,B]=[A,C]=0$ which implies that
the joint probabilities $p(ab|AB)$ and $p(ac|AC)$ can be defined.
Noncontextuality implies that outcome of the
measurement $A$ does not depend on whether it is measured with $B$ or $C$.
These observables exhibit contextuality if the joint probability $p(abc|ABC)$
cannot be defined.
The simplest physical system that exhibits contextuality is a qutrit system. Recently, KCBS derived a simplest noncontextual inequality
which is violated by a qutrit system with only five measurements \cite{KCBS}. It has been shown that in a qutrit-qubit system,
the violation of the KCBS inequality forbids the violation of the
CHSH inequality and vice versa which demonstrates monogamy between contextuality and nonlocality \cite{KCK}.
Similarly, we observe that if a maximally entangled state gives rise to KS paradox that demonstrates contextuality, the correlation does not exhibit nonlocality.
\subsubsection{Peres' version of Kochen-Specker (KS) paradox}
Peres \cite{Peres} showed that two-qubits in the singlet state exhibits KS paradox for the Pauli measurements $\sigma_{1x}$ and $\sigma_{1y}$
on the first qubit, and, $\sigma_{2x}$ and $\sigma_{2y}$ on the second qubit.
The outcomes exhibit anti-correlations for the measurements $\sigma_{1x}\sigma_{2x}$ and $\sigma_{1y}\sigma_{2y}$,
since the singlet state is a simultaneous eigenstate of these two measurement
operators as follows,
\ba
\sigma_{1x}\sigma_{2x}\ket{\phi^-}&=&-\ket{\phi^-} \label{es1} \\
\sigma_{1y}\sigma_{2y}\ket{\phi^-}&=&-\ket{\phi^-}. \label{es2}
\ea
This implies that the outcome pairs satisfy the following relation,
\be
v(\sigma_{1x}\sigma_{2x})=v(\sigma_{1y}\sigma_{2y})=-1. \label{v1}
\ee
For the other two choices of joint measurements $\sigma_{1x}\sigma_{2y}$
and $\sigma_{1y}\sigma_{2x}$, the outcomes are uncorrelated. However, the singlet state is eigenstate of the product of these two measurement operators as follows,
\be
(\sigma_{1x}\sigma_{2y})(\sigma_{1y}\sigma_{2x})\ket{\phi^-}=-\ket{\phi^-}, \label{es3}
\ee
which implies that the two outcome pairs satisfy the following relation,
\be
v(\sigma_{1x}\sigma_{2y})v(\sigma_{1y}\sigma_{2x})=-1. \label{v2}
\ee
If the outcomes can be predetermined noncontextually, Eqs. (\ref{v1}) and (\ref{v2}) imply that the following relation should be satisfied,
\ba
v(\sigma_{1x})v(\sigma_{2x})&=&-1\nonumber \\
v(\sigma_{1y})v(\sigma_{2y})&=&-1 \label{NCo}\\
v(\sigma_{1x})v(\sigma_{2y})v(\sigma_{1y})v(\sigma_{2x})&=&-1\nonumber
\ea
This relation is impossible to satisfy since the product of the left-hand side implies $+1$ which is not equal to the product of the right-hand side which is $-1$.

For the measurements that give rise to the Peres' paradox given in Eq. (\ref{NCo}), the correlation arising from the singlet state
violates the following EPR-steering inequality maximally \cite{CJWR},
\be
|\braket{\sigma_x\sigma_x}+\braket{\sigma_y\sigma_y}|\le\sqrt{2}.
\ee
Notice that the measurements that give rise to the maximal violation of the above EPR-steering inequality do not give rise to the violation of the Bell-CHSH inequality.
This suggests monogamy relation between the EPR-steering inequality and the Bell-CHSH inequality.
\subsubsection{Mermin's argument of GHZ paradox}
Greenberger, Horne and Zeilinger (GHZ) presented a paradox that illustrates nonlocality of quantum theory in the multipartite scenario
without using inequalities \cite{GHZ}.
Let us discuss the Mermin's version of the GHZ paradox \cite{GHZM} which is the tripartite generalization of the Peres' version of KS paradox \cite{UNLH}.
Consider the correlation arising from three-qubits in the following GHZ-state,
\be
\ket{\psi_{GHZ}}=\frac{1}{\sqrt{2}}\left[\ket{000}-\ket{111}\right].
\ee
for the two Pauli measurements $\sigma_{ix}$ and $\sigma_{iy}$ ($i=1,2,3$) performed on each qubit. Since the GHZ-state is the simultaneous
eigenstate of the three observables $\sigma_{1y}\sigma_{2y}\sigma_{3x}$, $\sigma_{1y}\sigma_{2x}\sigma_{3y}$,
and $\sigma_{1x}\sigma_{1y}\sigma_{1y}$ as follows,
\ba
\sigma_{1y}\sigma_{2y}\sigma_{3x}\ket{\psi_{GHZ}}&=&\ket{\psi_{GHZ}}\label{GHZ1}\\
\sigma_{1y}\sigma_{2x}\sigma_{3y}\ket{\psi_{GHZ}}&=&\ket{\psi_{GHZ}}\label{GHZ2}\\
\sigma_{1x}\sigma_{2y}\sigma_{3y}\ket{\psi_{GHZ}}&=&\ket{\psi_{GHZ}},\label{GHZ3}
\ea
the GHZ state gives rise to perfect correlations for these three measurements that is the product of the outcomes of
the three local Pauli measurements satisfy the following relation,
\be
v(\sigma_{1y}\sigma_{2y}\sigma_{3x})=v(\sigma_{1y}\sigma_{2x}\sigma_{3y})=v(\sigma_{1x}\sigma_{2y}\sigma_{3y})=1.
\ee
Since the three observables in Eqs. (\ref{GHZ1})-(\ref{GHZ3}) are mutually commuting, the GHZ-state is also an eigenstate of the product of these observables,
\be
(\sigma_{1y}\sigma_{2y}\sigma_{3x})(\sigma_{1y}\sigma_{2x}\sigma_{3y})v(\sigma_{1x}\sigma_{2y}\sigma_{3y})\ket{\psi_{GHZ}}
=(\sigma_{1x}\sigma_{2x}\sigma_{3x})\ket{\psi_{GHZ}}=-\ket{\psi_{GHZ}},
\ee
but this time with $-$ sign. The product of the local outcomes for the measurement of $\sigma_{1x}\sigma_{2x}\sigma_{3x}$ on the GHZ state implies,
\be
v(\sigma_{1x}\sigma_{2x}\sigma_{3x})=-1. \label{LCtd}
\ee
If local realistic value assignment is possible for the individual observables in Eqs. (\ref{GHZ1})-(\ref{GHZ3}), there exists hidden variables $\lambda$ such that
the following relation,
\be
v(\sigma_{1y})v(\sigma_{2y})v(\sigma_{3x})=v(\sigma_{1y})v(\sigma_{2x})v(\sigma_{3y})=v(\sigma_{1x})v(\sigma_{2y})v(\sigma_{3y})=1
\ee
should hold. The product of the left-hand side of this equation implies,
\be
v(\sigma_{1x})v(\sigma_{2x})v(\sigma_{3x})=1,
\ee
which, however, contradicts the condition in Eq. (\ref{LCtd}).

The GHZ paradox can be tested by the violation of the Mermin inequality \cite{mermin},
\be
|\braket{\sigma_{1x}\sigma_{2x}\sigma_{3x}}-\braket{\sigma_{1x}\sigma_{2x}\sigma_{3x}}
-\braket{\sigma_{1x}\sigma_{2x}\sigma_{3x}}-\braket{\sigma_{1x}\sigma_{2x}\sigma_{3x}}|\le2,
\ee
which is equivalent to a noncontextual inequality \cite{Canasetal}. Notice that the measurements that give rise to the violation of this inequality
does not violate a Svetlichny inequality \cite{SI}.
\subsection{Quantum discord}
In the seminal paper \cite{OZ}, quantum discord was defined as the difference between two inequivalent expressions for mutual information.
Nonzero quantum discord was proposed as a measure of quantum correlation which goes beyond entanglement.
Quantum discord of a bipartite state, $\rho$, equals to zero iff there exists a von-Neumann measurement $\{\Pi_k=\ket{\psi_k}\bra{\psi_k}\}$ such that \cite{Datta}
\be
\left(\Pi^{}_k \otimes \one\right) \rho \left(\Pi^{}_k \otimes \one\right)=\rho.
\ee
This implies that the zero-discord states can be written in the classical-quantum form \cite{CQ} $\rho=\sum_k \ket{\psi_k}\bra{\psi_k} \otimes \rho_k$
where $\ket{\psi_k}\bra{\psi_k}$ are the orthonormal states on Alice's side and $\rho_k$ are quantum states on Bob's side.
The set of classical-quantum states forms a nonconvex subset of the set of separable states \cite{Caves}.
A separable state which cannot be written in the classical-quantum form has nonclassical correlation.
It has been shown that almost all quantum states
have nonclassical correlation \cite{QCall}.
In Ref. \cite{QDRS}, it has been shown that the fidelity of remote quantum state preparation is related to geometric measure of quantum discord \cite{Dakicetal}.
The geometric measure of left discord is defined as,
\be
\mathcal{D}^\rightarrow(\rho)=2\min_{\chi\in{\Omega}_0}||\rho-\chi||^2,
\ee
where $\Omega_0$ denotes the set of classical-quantum states and $||X -Y||^2 = \tr[(X-Y)^2]$.
\section{Nonsignaling polytope}
Bipartite Bell scenario can be abstractly described in terms of input-output devices shared by two parties as follows.
Alice and Bob have access to a black box; when Alice and Bob input $A_i$ and $B_j$ into the box, the box yields
outputs $a_m$ and $b_n$. In the physical scenario, the inputs correspond to measurement choices and the outputs correspond to the outcomes of the measurements.
Let us denote the number of possible inputs on Alice's side and Bob's side by $d_i$ and $d_j$ and the number of possible outputs for a given choice
of input on Alice's side and Bob's side
by $d_m$ and $d_n$. A Bell scenario is characterized by the set of $N=d_i \times d_j \times d_m \times d_n$ joint probability distributions (JPD), $P(a_m,b_n|A_i.B_j)$,
which satisfy positivity,
\be
P(a_m,b_n|A_i,B_j)\ge 0, \label{postitvity}
\ee
normalization constraints,
\be
\sum_{m,n}P(a_m,b_n|A_i,B_j)=1 \quad \forall i,j, \label{normalizations}
\ee
and nonsignaling constraints,
\ba
\sum_{n}P(a_m,b_n|A_i,B_j)&=&P(a_m|A_i,B_j)=P(a_m|A_i) \quad \forall i,j, m, \label{nsA}\\
\sum_{m}P(a_m,b_n|A_i,B_j)&=&P(b_n|A_i,B_j)=P(b_n|B_j) \quad \forall i,j, n. \label{nsB}
\ea
We refer to the set $P(a_m,b_n|A_i.B_j)$ which satisfy the constraints in Eqs. (\ref{normalizations}-(\ref{nsB}) as correlation or box.

A box can be regarded
as the vector in an $N$-dimensional space whose coordinates are the joint probabilities.
Not all joint probabilities are independent in the set due to the normalization and the nonsignaling constraints.
For each input pair, one joint probability can be eliminated by using the normalization constraints in Eq. (\ref{normalizations});
the eliminated one is denoted as $P(a_{m'},b_{n'}|A_i,B_j)$. For a given input pair, the joint probabilities which have
the output that is contained in the eliminated joint probability can be written as,
\ba
P(a_m,b_{n'}|A_i,B_j)=P(a_m|A_i)-\sum_{n\ne n'}P(a_m,b_n|A_i,B_j)\\
P(a_{m'},b_n|A_i,B_j)=P(b_n|B_j)-\sum_{m\ne m'}P(a_m,b_n|A_i,B_j)
\ea
which follow from the nonsignaling constraints in Eqs. (\ref{nsA}) and (\ref{nsB}).
Notice that the marginal and the joint distributions which do not contain $a_{m'}$ or $b_{n'}$ are linearly independent.
Therefore, the set of linearly independent marginal and joint distributions form a basis of dimension,
\be
D(\mathcal{N})=d_i \times (d_m-1)+d_j \times (d_n-1)+d_i \times d_j \times (d_m-1) \times (d_n-1),
\ee
for the vector space that uniquely describes the set of nonsignaling correlations \cite{ppolytopes}. A basis set is not unique, i.e.,
there are the finite number of basis sets for the nonsignaling
space. The basis sets are related to each other by local reversible operations (LRO). LRO simply relabel the inputs and outputs:
Alice changing her input $i\rightarrow i\oplus 1$, and changing her output conditioned on the input: $m\rightarrow m\oplus\alpha i\oplus\beta$.
Bob can perform similar operations.
Local reversible operations (LRO) are analogous to local unitary operations in quantum theory.
It is known that Alice and Bob cannot decrease entanglement and cannot create entanglement from separability by local unitary operations
on the quantum states \cite{Hetal}, similarly, nonlocality and locality are invariant under LRO.
The set of nonsignaling correlations forms a polytope in $D(\mathcal{N})$-dimensional space since it is an intersection of the finite number of
hyperplanes given by Eqs. (\ref{normalizations})-(\ref{nsB}). This polytope is convex since the set of nonsignaling correlations is convex
i.e., convex combination of any two nonsignaling correlation is another nonsignaling
correlation. The nonsignaling polytope is given by the set of $D(\mathcal{N})$ linearly independent joint and marginal distributions which satisfy,
\ba
\sum_{n\ne n'}P(a_m,b_n|A_i,B_j) &\le& P(a_m|A_i)\nonumber \\
\sum_{m\ne m'}P(a_m,b_n|A_i,B_j) &\le& P(b_n|B_j) \quad \forall \quad i,j. \label{Hns}
\ea
These inequalities give $\mathcal{H}$-representation for the nonsignaling polytope.

Since a polytope can also be represented in the $\mathcal{V}$-representation in which it is a convex hull of the vertices of the polytope with positive weights.
The vertices of the nonsignaling polytope are the unique solutions of the constraints in Eqs. (\ref{postitvity})-(\ref{nsB}) with sufficient
number of times the inequalities in Eq. (\ref{postitvity}) are replaced by equalities.
The vertices of the nonsignaling polytope can be divided into two classes: deterministic and nondeterministic. A deterministic correlation can be written
as the product of the marginals
corresponding to Alice and Bob, $P_D(a_m,b_n|A_i,B_j)=P_D(a_m|A_i)P_D(b_n|B_j)$, here $P_D(a_m|A_i)$ and $P_D(a_m|A_i)$ can take either zero or one for all $m,n,i,j$.
A nondeterministic vertex is known as Popescu-Rohrlich box or maximally nonlocal box \cite{Barrett}.
\subsection*{Local polytope}
Any stochastic hidden variable model can be transformed into a deterministic hidden variable model \cite{Fine},
\be
P(a_m,b_n|A_i,B_j)=\sum_\lambda p_\lambda P_\lambda(a_m|A_i)P_\lambda(b_n|B_j), \label{LDM}
\ee
where $P_\lambda(a_m|A_i)$ and $P_\lambda(b_n|B_j)$ are deterministic. Therefore, the set of local correlations forms a convex polytope known as
Bell polytope or local polytope
whose vertices are the deterministic boxes. All the tight Bell inequalities \cite{WernerWolf},
\be
\sum_{m,n,i,j}C^{ij}_{mn}P(a_m,b_n|A_i, B_j)\le L,
\ee
which are the bounds on the certain linear combinations of the joint probabilities under the constraint in Eq. (\ref{LDM}), form the facets of the local polytope.
These facet inequalities
together with the inequalities in Eq. (\ref{Hns}) give $\mathcal{H}$-representation for the local polytope.
\subsection*{Quantum correlations}
Quantum correlations obtained by local measurements on bipartite quantum systems are given by,
\be
P(a_m,b_n|A_i,B_j)=\tr\left(\rho M^{A_i}_{a_m}\otimes M^{B_j}_{b_n}\right), \label{QTQC}
\ee
where $\rho$ is a bipartite quantum state in a Hilbert space $\mathcal{H}_A\otimes \mathcal{H}_B$, and, $M^{A_i}_{a_m}$ and $M^{B_j}_{b_n}$ are positive operator
valued measures
satisfying positivity, $M^{A_i}_{a_m}\ge0$ and $M^{B_j}_{b_n}\ge0$, and the normalizations, $\sum_{m}M^{A_i}_{a_m}=\openone$ and $\sum_n M^{B_j}_{b_n}=\openone$.
The correlation
predicted by quantum theory as given in Eq. (\ref{QTQC}) implies
that the marginal distributions of Alice and Bob satisfy the nonsignaling principle since
$\sum_nP(a_m,b_n|A_i,B_j)=\sum_n\tr\left(\rho M^{A_i}_{a_m}\otimes M^{B_j}_{b_n}\right)=\tr\left(\rho M^{A_i}_{a_m}\otimes \openone\right)$
and $\sum_mP(a_m,b_n|A_i,B_j)=\sum_m\tr\left(\rho M^{A_i}_{a_m}\otimes M^{B_j}_{b_n}\right)=\tr\left(\rho \openone\otimes M^{B_j}_{b_n}\right)$. Thus,
the set of
quantum correlations is contained in the nonsignaling polytope. Quantum correlations form a convex set; however, it is not a polytope \cite{Pitowski}
since it has infinitely many extremals.
Since there are quantum correlations that violate a Bell inequality and the violation is limited by the Tsirelson bound \cite{tsi1},
quantum correlations are sandwiched between the nonsignaling polytope and the local polytope.

\section{Motivation for the results}
Local correlations are considered as classical in the device-independent framework.
When the local Hilbert space dimensions
are constrained, there are local correlations which can have nonclassicality. There are two kinds of origin of nonclassicality which are
manifested in the type of measurements used for generating the local correlations. That is, nonclassicality of local correlations can originate from noncommuting measurements that demonstrate Bell nonlocality or EPR steering
without Bell nonlocality.
I observed that just like nonlocal correlations, the local correlations which can imply the presence of nonclassicality have a Popescu-Rohrlich box decomposition. This motivated me to obtain a canonical decomposition which can have
nonzero Popescu-Rohrlich box component even for the local correlations.

Moving to the multipartite scenario, the observation of genuine nonlocality implies the
presence of genuine quantum correlation in a device-independent way.
However, there are local correlations which can imply the presence of
genuine quantum correlation when the local Hilbert space dimensions are constrained.
In this thesis, we focus on those quantum correlations which correspond to Svetlichny-type and Mermin-type scenarios.
In the Svetlichny-type scenario, genuine nonlocality is observed using genuinely entangled states and noncommuting measurements which lead to
violation of a Svetlichny inequality \cite{SI,Ghoseetal}.
In this scenario, there are tripartite qubit correlations which are local, but nevertheless, have genuine nonclassicality originating from three-way nonlocality.
In the Mermin-type scenario, genuinely entangled states and noncommuting measurements
that do not demonstrate genuine nonlocality are used to demonstrate Mermin nonlocality \cite{mermin,vMI}.
In the Mermin-type scenario, there are tripartite qubit correlations which are local, but nevertheless, have genuine nonclassicality
originating from Mermin nonlocality.
I observed that just like three-way nonlocal correlations,
the local correlations which can imply the presence of genuine nonclassicality have a Svetlichny box decomposition. This motivated me to obtain a canonical decomposition which can have
nonzero Svetlichny box component even for the local correlations.

\section{Summary and results}
In this thesis, I characterize bipartite and multipartite quantum correlations using nonsignaling polytopes.
\subsection{Bipartite quantum correlations}
In Chapters \ref{Ch2} and \ref{Ch3}, we characterize bipartite nonsignaling boxes with two binary inputs and two binary outputs.
We introduce two notions of nonclassicality of quantum correlations originating from nonlocality and EPR-steering.
To quantify these two types of nonclassicality, we define the two measures, Bell discord and Mermin discord,
which are nonzero also for boxes admitting local hidden variable model.
We obtain canonical decomposition for nonsignaling boxes using the division of the full nonsignaling polytope with respect to these two measures.
We find that any qubit correlations can be decomposed into Popescu-Rohrlich box,
a maximally EPR-steerable box and a local box with Bell and Mermin discord equal to zero.
We characterize and quantify nonclassicality of bipartite quantum correlations using the canonical decomposition and
the two measures.
We show that all quantum states which have non-null quantum discord with respect to both the subsystems \cite{Dakicetal}
can have Bell discord or Mermin discord or both of them simultaneously.
We study nonclassicality of various two-qubit states to illustrate the relevance of Bell and Mermin discord to isolate the origin of nonclassicality.
In Chapter \ref{Ch4}, we introduce a third measure to study total correlations in nonclassical probability distributions arising from various
two-qubit states.
\subsection{Multipartite quantum correlations}
In Chapter \ref{Ch5}, we investigate tripartite quantum correlations using Svetlichny-box polytope
which is a generalization of the PR-box polytope to the multipartite scenario.
We define Svetlichny discord and Mermin discord which are the multipartite generalization of the two bipartite measures introduced in Chapters \ref{Ch2} and \ref{Ch3}.
We find that tripartite qubit correlations which are contained in the Svetlichny-box polytope can be written as a convex mixture of a Svetlichny-box
which exhibits three-way nonlocality,
a three-way contextual box that exhibits the GHZ paradox and a purely classical box that does not have Svetlichny and tripartite Mermin discord.
We illustrate that Svetlichny discord and Mermin discord quantify three-way nonlocality and three-way contextuality of all pure genuinely
entangled states with respect to this decomposition.
We find that separable and biseparable mixed three-qubit states that have an irreducible genuinely entangled state component can give rise to genuine three-way nonclassicality
with respect to the measures, Svetlichny and Mermin discord. We define a measure for total correlations to divide the total amount of correlations
in a given quantum joint probability distribution into three-way nonlocality, three-way contextuality and genuinely classical correlations. 

\chapter{Bell discord and Canonical decomposition of bipartite nonsignaling boxes}
\label{Ch2}


\section*{Abstract}
We  study nonclassicality  in bipartite quantum  correlations in
the   context   of   nonsignaling  polytopes,   that   goes   beyond
nonlocality.
We introduce the  measure, Bell discord, to quantify nonclassicality
of quantum correlations originating from Bell nonlocality. We find that any
nonsignaling box can be written as a convex mixture of an irreducible
Popescu-Rohrlich box and a local box with Bell discord equals to zero.
We illustrate that nonzero Bell discord of quantum correlations originate
from incompatible  measurements  that give  rise  to  Bell nonlocality.
\section{Introduction}
Nonlocality  of  quantum correlations  implies  the  presence of  both
incompatible  measurements and  entanglement \cite{IncomN}.   All pure
bipartite entangled  states violate a Bell  inequality for appropriate
incompatible measurements \cite{GT,PRQB}.  However, Werner showed that
nonlocality  and  entanglement  are   inequivalent;  there  are  mixed
entangled  states   which  have   LHV  models  for   all  measurements
\cite{Werner}.   Thus,  not  all  entangled states  can  lead  to  the
violation of a Bell inequality even when incompatible measurements are
performed on them.
Quantum discord  was introduced as  a measure of  quantum correlations
which  quantifies   nonclassicality  of   separable  states   as  well
\cite{OZ}. In Ref. \cite{Perinotti}, a notion of discord was introduced for states in causal probabilistic
theories \cite{CAP}, which demonstrated that non-null discord is generic nonclassical feature.
It  would be  interesting  to investigate  whether local  correlations
arising  from  incompatible  measurements  performed  on  the  quantum
discordant  states   can  have   nonclassicality.

In  this  work, we  introduce the measure, Bell discord,
to characterize  quantum  correlations in  the framework  of
GNST. Just like geometric measure of quantum discord \cite{Dakicetal},
nonzero Bell discord detects the presence of nonclassicality
in quantum correlations which do not violate a Bell  inequality.
We restrict to  the NS polytope in which the
black  boxes  have  two  binary-inputs  and  two  binary-outputs, i.e., we
characterize  only  those NS  boxes  with  two binary-inputs  and  two
binary-outputs. We show that any nonsignaling box can be decomposed into
Popescu-Rohrlich box and a local box with Bell discord equals to zero.
We find that a bipartite qubit correlation has nonzero Bell discord
if the measured state has nonzero left and right quantum discord \cite{Dakicetal}
and the measurements that give rise to them are incompatible.
\section{Preliminaries}\label{prl}
In GNST, bipartite systems are described by the black boxes shared between two parties. Suppose Alice and Bob input the random variables
$A_i$ and $B_j$ into a black box which they share and obtain the outputs $a_m$ and $b_n$,  the behavior of the given black box is described by
the set of conditional probability distributions, $P(a_m,b_n|A_i,B_j)$.
In the case of two binary-inputs and two binary-outputs, i.e., $m,n,i,j \in \{0,1\}$,
a black box is characterized by $16$ probability distributions which can
be represented in matrix notation as follows,

\begin{equation}
\left( \begin{array}{cccc}
P(a_0,b_0|A_0,B_0) & P(a_0,b_1|A_0,B_0) & P(a_1,b_0|A_0,B_0) & P(a_1,b_1|A_0,B_0) \\
P(a_0,b_0|A_0,B_1) & P(a_0,b_1|A_0,B_1) & P(a_1,b_0|A_0,B_1) & P(a_1,b_1|A_0,B_1) \\
P(a_0,b_0|A_1,B_0) & P(a_0,b_1|A_1,B_0) & P(a_1,b_0|A_1,B_0) & P(a_1,b_1|A_1,B_0) \\
P(a_0,b_0|A_1,B_1) & P(a_0,b_1|A_1,B_1) & P(a_1,b_0|A_1,B_1) & P(a_1,b_1|A_1,B_1) \\
\end{array} \right).
\end{equation}

Barrett \etal\cite{Barrett} showed that the set of bipartite nonsignaling boxes ($\mathcal{N}$) with two binary-inputs and two binary-outputs
forms an $8$ dimensional convex polytope with $24$
vertices. The vertices (or extremal boxes) of this polytope
are $8$ PR-boxes,
\begin{align}
P^{\alpha\beta\gamma}_{PR}(a_m,b_n|A_i,B_j)=\left\{
\begin{array}{lr}
\frac{1}{2}, & m\oplus n=i\cdot j \oplus \alpha i\oplus \beta j \oplus \gamma\\
0 , & \text{otherwise}\\
\end{array}
\right. \label{NLV}
\end{align}
and $16$ deterministic boxes:
\be
P^{\alpha\beta\gamma\epsilon}_D(a_m,b_n|A_i,B_j)=\left\{
\begin{array}{lr}
1, & m=\alpha i\oplus \beta\\
   & n=\gamma j\oplus \epsilon \\
0 , & \text{otherwise}.\\
\end{array}
\right.
\ee
Here $\alpha,\beta,\gamma,\epsilon\in \{0,1\}$  and $\oplus$ denotes addition modulo $2$. Any NS box can be written as a convex sum of the $24$ extremal boxes:
 \be
P(a_m, b_n|A_i,B_j)=\sum^7_{k=0}p_kP^k_{PR}+\sum^{15}_{l=0}q_lP^l_{D},
 \label{CHNS}
\ee
with $\sum_kp_k+\sum_lq_l=1$. Here $k=\alpha\beta\gamma$ and $l=\alpha\beta\gamma\epsilon$.
All the deterministic boxes can be written as the product of marginals corresponding to Alice and Bob, $P_D(a_m,b_n|A_i,B_j)=P_D(a_m|A_i)P_D(b_n|B_j)$,
whereas the $8$ PR-boxes
cannot be written in product form. Note that unlike the deterministic boxes, the marginals of the PR boxes are maximally mixed: {\it i.e.},
$P(a_{m}|A_i)=\frac{1}{2}=P(b_{n}|B_j)$ for all $i,j,m,n$.
The extremal boxes in a given class are equivalent under local reversible operations (LRO) which
include local relabelling of party's inputs and outputs.

Bell polytope ($\mathcal{L}$), which is a subpolytope of $\mathcal{N}$, is a convex hull of the $16$ deterministic boxes: if $P(a_m, b_n|A_i,B_j)\in \mathcal{L}$,
\ba
P(a_m, b_n|A_i,B_j)=\sum^{15}_{l=0}q_lP^l_{D}; \sum_lq_l=1. \label{LD}
\ea
Fine \cite{Fine} showed that a box can be simulated by the deterministic local hidden variable model given above iff the box
satisfies the complete set of Bell-CHSH inequalities \cite{chsh,WernerWolf}:
\ba
\mathcal{B}_{\alpha\beta\gamma} &:=& (-1)^\gamma\braket{A_0B_0}+(-1)^{\beta \oplus \gamma}\braket{A_0B_1}\nonumber\\
&&+(-1)^{\alpha \oplus \gamma}\braket{A_1B_0}+(-1)^{\alpha \oplus \beta \oplus \gamma \oplus 1} \braket{A_1B_1}\le2, \label{BCHSH1}
\ea
which are the nontrivial facets of the Bell polytope.
Here \begin{align*}\braket{A_iB_j}=\sum_{mn}(-1)^{m\oplus n}P(a_m,b_n|A_i,B_j).\end{align*}
All nonlocal boxes lie outside the Bell polytope and violate a Bell-CHSH inequality.

Quantum boxes which belong to the Bell-CHSH scenario \cite{chsh} are obtained by two dichotomic measurements on bipartite quantum states
described by the density matrix $\rho_{AB}$ in the Hilbert space $\mathcal{H}_A\otimes\mathcal{H}_B$.
The Born's rule predicts the behavior of the quantum boxes as follows,
\begin{equation}
P(a_m,b_n|A_i,B_j)=\mathrm{Tr}\left(\rho_{AB}\mathcal{M}_{A_i}^{a_m}\otimes\mathcal{M}_{B_j}^{b_n}\right),\label{QNS}
\end{equation}
where $\mathcal{M}_{A_i}^{a_m}$ and $\mathcal{M}_{B_j}^{b_n}$
are the measurement operators generating binary outcomes $a_{m},b_{n} \in \{-1,1\}$.
A nonlocal box given by decomposition in Eq. (\ref{CHNS}) is quantum if it can be written in the above form. Since the set of quantum boxes is convex \cite{WernerWolf},
any local box can be written
in the form given in Eq. (\ref{QNS}). In this work, we characterize quantum boxes arising from
spin projective measurements $A_i=\hat{a}_i \cdot \vec{\sigma}$ and $B_j=\hat{b}_j \cdot \vec{\sigma}$ along the directions $\hat{a}_i$ and $\hat{b}_j$
on two-qubit systems. Here $\vec{\sigma}$ is the vector of Pauli matrices.
\section{Bell discord}\label{BD}
Fine showed that a quantum box violates a Bell-CHSH inequality iff joint probability distributions for the triples of observables:
$A_0$, $B_0$, $B_1$ and $A_1$, $B_0$, $B_1$ cannot be defined \cite{Fine,Fine1}. This implies that the measurements that give rise
to the violation of a Bell-CHSH inequality are incompatible,
i.e., measurement observables on Alice's and Bob's sides are noncommuting: $[A_0,A_1]\ne0$ and $[B_0,B_1]\ne0$.
However, if a quantum box does not violate a Bell-CHSH inequality, it does not necessarily imply that it cannot arise from incompatible
measurements on an entangled state.

We consider isotropic PR-box \cite{MAG06} which is a mixture of a PR-box and white noise,
\be
P=pP_{PR}+(1-p)P_N. \label{PRiso}
\ee
Here $P_{PR}$ is the canonical PR-box,
\begin{equation}
P^{000}_{PR} = \left( \begin{array}{cccc}
\half & 0 & 0 & \half \\
\half & 0 & 0 & \half \\
\half & 0 & 0 & \half \\
0 & \half & \half & 0
\end{array} \right),
\label{eq:prbox}
\end{equation}
and $P_N$ is white noise defined as follows,
\begin{equation}
P_{N} = \left( \begin{array}{cccc}
\qua & \qua & \qua & \qua \\
\qua & \qua & \qua & \qua \\
\qua & \qua & \qua & \qua \\
\qua & \qua & \qua & \qua
\end{array} \right).
\label{eq:wn}
\end{equation}
The isotropic PR-box violates the Bell-CHSH inequality, i.e., $\mathcal{B}_{000}=4p>2$ if $p>\frac{1}{2}$.
Notice that even if the isotropic PR-box is local when $p\le \frac{1}{2}$, it admits a decomposition with
the single PR-box. We call such a single PR-box in the decomposition of any box (nonlocal, or not) irreducible PR-box.

The isotropic PR-box which is quantum physically realizable if $p\le\frac{1}{\sqrt{2}}$ \cite{MAG06} illustrates the following observation.
\begin{observation}
When local boxes arising from entangled two-qubit states have an irreducible PR-box component, the projective measurements that give rise to them are incompatible.
\end{observation}
For the noncommuting measurement observables $A_0=\sigma_x$, $A_1=\sigma_y$,
$B_0=\frac{1}{\sqrt{2}}(\sigma_x-\sigma_y)$ and $B_1=\frac{1}{\sqrt{2}}(\sigma_x+\sigma_y)$, the pure entangled states,
\be
\ket{\psi(\theta)}=\cos\theta\ket{00}+\sin\theta\ket{11}; \quad 0 \le \theta \le \pi/4,    \label{nmE}
\ee
give rise to the isotropic PR-box given in Eq. (\ref{PRiso})
with $p=\frac{\sin2\theta}{\sqrt{2}}$. For this choice of measurements, the box is nonlocal if $\sin2\theta>\frac{1}{\sqrt{2}}$. However, the
box has the irreducible PR-box component whenever the state is entangled.

The  observation that  a local  box  which has  an irreducible  PR-box
component  can arise  from incompatible  measurements on  an entangled
state motivates  to define a  notion of nonclassicality which  we call
Bell discord.
\begin{definition}\label{BDdef1}
A  box arising  from incompatible  measurements on  a given  two-qubit
state has \textit{Bell  discord} iff it admits a  decomposition with an
irreducible PR-box component.
\end{definition}
Bell discord is  not equivalent to Bell nonlocality  since local boxes
can  also have  an  irreducible PR-box  component;  for instance,  the
isotropic  PR-box in  Eq.  (\ref{PRiso}) has  Bell  discord if  $p>0$,
whereas it has Bell nonlocality if $p>\frac{1}{2}$.

Notice  that it  is not  necessary that  a given  local box  with Bell
discord can only  arise from incompatible measurements  on a two-qubit
state  since it  can  also  arise from  a  separable  state in  higher
dimensional space for  compatible measurements \cite{DQKD}. We will  show that any
local box with Bell discord  cannot arise from compatible measurements
on two-qubit systems.

Any isotropic PR-box,
\be
P=pP^{\alpha\beta\gamma}_{PR}+(1-p)P_N, \label{isoPR1}
\ee
has a special property that only one of the Bell functions,
\begin{align}
\mathcal{B}_{\alpha\beta}&= |(-1)^\gamma\braket{A_0B_0}+(-1)^{\beta}\braket{A_0B_1}\nonumber\\
&+(-1)^{\alpha}\braket{A_1B_0}+(-1)^{\alpha \oplus \beta  \oplus 1} \braket{A_1B_1}|, \label{bchshmod}
\end{align}
which are the modulus of the Bell-CHSH operators in Eq. (\ref{BCHSH1}), is nonzero.
This is due to the Bell function monogamy (see Appendix. \ref{mBF}) of the irreducible PR-box, $P^{\alpha\beta\gamma}_{PR}$, in the decomposition.
Thus, the above property quantifies Bell discord of the local isotropic PR-boxes.
Local boxes that have an irreducible PR-box component, in general, have more than one Bell functions nonzero.

Before defining a measure of Bell discord which quantifies irreducible PR-box in any box, we construct the following quantities,
\ba
\mathcal{G}_1&:=&\Big||\mathcal{B}_{00}-\mathcal{B}_{01}|-|\mathcal{B}_{10}-\mathcal{B}_{11}|\Big|\nonumber\\
\mathcal{G}_2&:=&\Big||\mathcal{B}_{00}-\mathcal{B}_{10}|-|\mathcal{B}_{01}-\mathcal{B}_{11}|\Big| \label{gi}\\
\mathcal{G}_3&:=&\Big||\mathcal{B}_{00}-\mathcal{B}_{11}|-|\mathcal{B}_{01}-\mathcal{B}_{10}|\Big| \nonumber .
\ea
Here $\mathcal{G}_i$ are constructed such that it satisfies the following properties: (i) positivity, i.e., $\mathcal{G}_i\ge0$, (ii) $\mathcal{G}_i=0$ for all the
deterministic boxes and (iii) the algebraic maximum of $\mathcal{G}_i$ is achieved by the PR-boxes, i.e., $\mathcal{G}_i=4$ for any PR-box.
\begin{definition}\label{BDdef}
Bell discord, $\mathcal{G}$, is defined as,
\begin{equation}
\mathcal{G} := \min_i \mathcal{G}_i, \label{defBD}
\end{equation}
where $\mathcal{G}_i$ are given in Eq. (\ref{gi}). Here $0\le\mathcal{G}\le4$.
\end{definition}
Bell discord is clearly invariant under LRO and
interchange of the subsystems since the set $\{\mathcal{G}_i, i=1,2,3\}$ is invariant under these two transformations. Therefore,
a $\mathcal{G}>0$ box cannot be transformed into a $\mathcal{G}=0$ box by LRO and vice versa.

\begin{observation}
The set of local boxes that have $\mathcal{G}=0$ forms a subset of the set of all local boxes and is nonconvex.
\end{observation}
\begin{proof}
The set of $\mathcal{G}=0$ boxes is nonconvex since certain convex combination of the deterministic boxes can have $\mathcal{G}>0$. For instance,
the boxes in Eq. (\ref{isoPR1}) can be written as a
convex combination of the deterministic boxes when $p\le\frac{1}{2}$, however, it has Bell discord $\mathcal{G}=4p>0$ if $p>0$.
As the deterministic boxes have $\mathcal{G}=0$ and the Bell polytope contains $\mathcal{G}>0$ boxes, the set of $\mathcal{G}=0$
boxes form a subset of the local boxes.
\end{proof}

The division of the Bell polytope with respect to $\mathcal{G}$ allows us to obtain the following canonical decomposition of the NS boxes (see
Appendix. \ref{pr1} for details).
\begin{theorem}
\label{thm1}
Any NS box can be decomposed into PR-box and a local box that does not have an irreducible PR-box component,
\be
P=\mu P^{\alpha\beta\gamma}_{PR}+\left(1-\mu\right)P_{L}^{\mathcal{G}=0}, \label{Gde}
\ee
where $\mu$ is the maximal irreducible PR-box component and $P_{L}^{\mathcal{G}=0}$ is the local box which has $\mathcal{G}=0$.
\end{theorem}

We say that the decomposition of the NS boxes given in Eq. (\ref{Gde}) is canonical in that it represents the classification of any NS box
according to whether it has Bell discord or not, which is more general than the classification of NS boxes into nonlocal and local boxes.
Notice that the irreducible PR-box component in Eq. (\ref{Gde}) should not be confused with the nonlocal cost which goes to zero for all the local boxes \cite{EPR2,EPR2B}.

\begin{figure}[h!]
\centering
\includegraphics[scale=0.40]{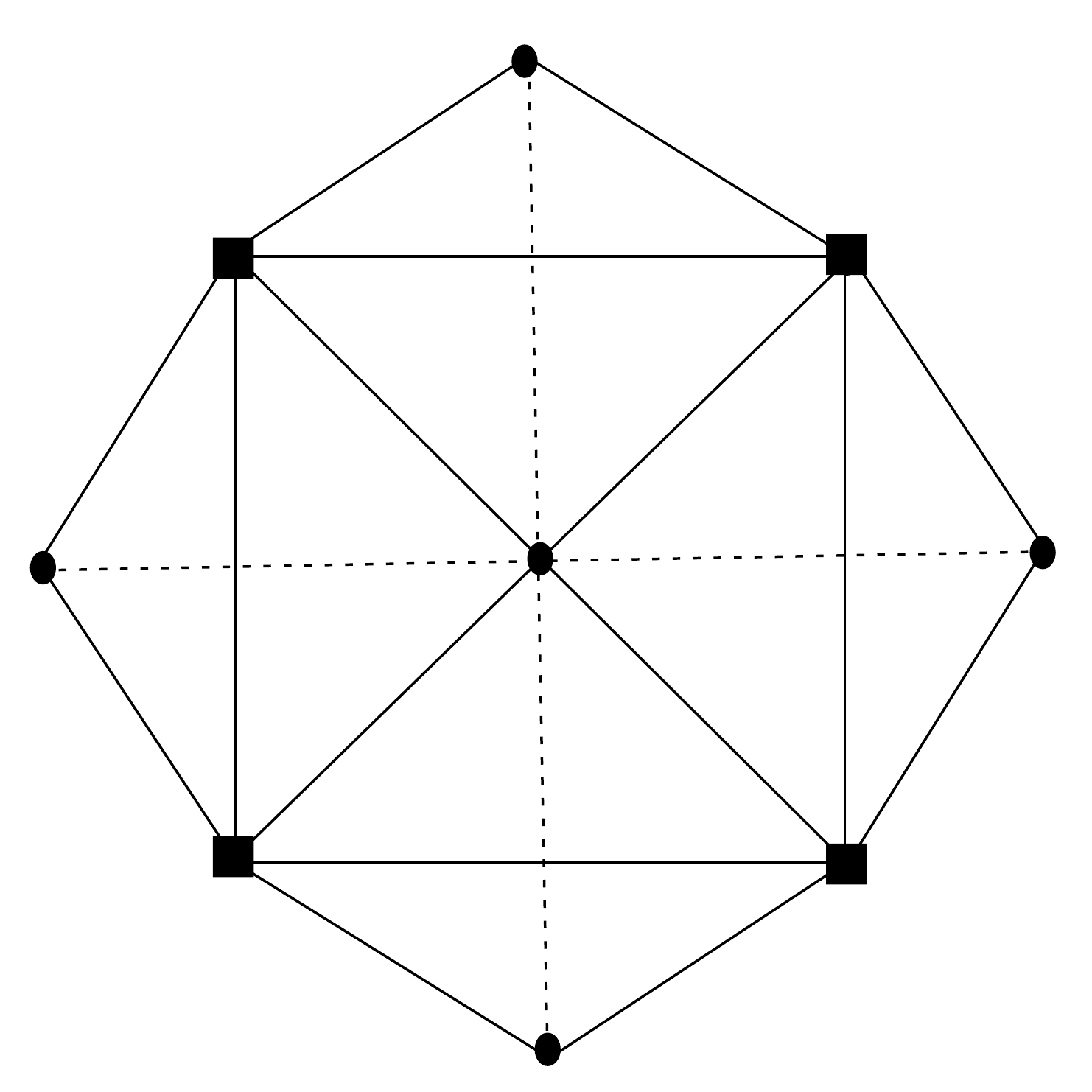}
\caption[Canonical decomposition]{A two dimensional representation of the NS polytope is shown here. Square represents the local polytope whose
vertices denoted by square points represent the deterministic boxes.
The circular points which lie above the local polytope represent the PR-boxes.
The points which lie on the lines connecting the center of the NS polytope (white noise) and the square points forms $\mathcal{G}=0$ nonconvex polytope.
Any point that goes outside the $\mathcal{G}=0$ region lies on a line joining a PR-box and a $\mathcal{G}=0$ box; for instance, any point that lies
on the dotted line can be written as a convex mixture of a PR-box and white noise.
}\label{Gpolytope}
\end{figure}

We now notice that a box has nonzero Bell discord iff it admits a decomposition that has an irreducible PR-box component.
For any box given by the decomposition in Eq. (\ref{Gde}),
$\mathcal{G}$ is linear (see Appendix \ref{lbdmd} for illustration), i.e., $\mathcal{G}(P)=\mu\mathcal{G}\left(P^{\alpha\beta\gamma}_{PR}\right)+\left(1-\mu
\right)\mathcal{G}\left(P^{\mathcal{G}=0}_L\right)$ which implies that $\mathcal{G}(P)=4\mu>0$ iff $\mu>0$.
Thus, if a box has nonzero Bell discord, it lies on a line joining a PR-box and a local box that does not have an irreducible PR-box component (see fig. \ref{Gpolytope} for
illustration).  The invariance of $\mathcal{G}$ under LRO implies that the irreducible PR-box component in the canonical decomposition given in Eq. (\ref{Gde}) is
invariant under LRO.

\section{Bell discord of two-qubit states}
We will apply Bell discord to the boxes arising from the pure entangled states and the Werner states.
Nonzero Bell discord of local boxes arising from these states originates from incompatible measurements which give rise to Bell nonlocality.
The incompatibility of measurement observables amounts to $\hat{a}_0\cdot\hat{a}_1 \ne 1$ and $\hat{b}_0\cdot\hat{b}_1 \ne 1$
for the measurement unit vectors. We will find that optimal Bell discord
is achieved by the orthogonal measurements on both the sides, i.e.,
$\hat{a}_0\cdot\hat{a}_1 = 0$ and $\hat{b}_0\cdot\hat{b}_1 = 0$. For a given state, a box has optimal Bell discord if only one of the Bell functions
$\mathcal{B}_{\alpha\beta}$ in Eq. (\ref{bchshmod}) is nonzero.

\subsection{Pure nonmaximally entangled states}
Any pure entangled state can be written in the Schmidt form \cite{Sch2} given in Eq. (\ref{nmE}).
Entanglement of these pure states can be quantified by the tangle, $\tau=\sin^22\theta$ \cite{CKW}.

(a) For the orthogonal measurement settings:
${\vec{a}_0}=\hat{x}$, ${\vec{a}_1}=\hat{y}$,
${\vec{b}_0} =\frac{1}{\sqrt{2}}(\hat{x}-\hat{y})$ and ${\vec{b}_1}=\frac{1}{\sqrt{2}}(\hat{x}+\hat{y})$,
the pure entangled states in Eq. (\ref{nmE}) give to
the isotropic PR-box as follows:
\ba
P=\frac{\sqrt{\tau}}{\sqrt{2}}P_{PR}+\left(1-\frac{\sqrt{\tau}}{\sqrt{2}}\right) P_N. \label{BSb}
\ea
The above box violates the Bell-CHSH inequality, i.e., $\mathcal{B}_{000}=2\sqrt{2\tau}>2$ if $\tau>\frac{1}{2}$
and has Bell discord $\mathcal{G}=2\sqrt{2\tau}>0$ if $\tau>0$. Notice that the irreducible PR-box component of the local box in Eq. (\ref{BSb})
is due to entanglement and the incompatible measurements that gives rise to Bell nonlocality.

(b) Popescu and Rohrlich showed that all the pure entangled states give rise to Bell nonlocality
for the state dependent settings \cite{PRQB}:
${\vec{a}_0}=\hat{z}$, ${\vec{a}_1}=\hat{x}$,
${\vec{b}_0}=\cos t\hat{z}+\sin t\hat{x}$ and ${\vec{b}_1}=\cos t\hat{z}-\sin t\hat{x}$,
where $\cos t=\frac{1}{\sqrt{1+\tau}}$. For this settings, the
box can be decomposed into PR-box and a local box which has nonmaximally mixed marginals and $\mathcal{G}=0$,
\be
P=\frac{\tau}{\sqrt{1+\tau}}P_{PR}+\left(1-\frac{\tau}{\sqrt{1+\tau}}\right)P^{\mathcal{G}=0}_L.\label{PRQ}
\ee
Here the $\mathcal{G}=0$ box, $P^{\mathcal{G}=0}_L$, becomes white noise for the maximally entangled state.
For the above box, the Bell-CHSH operator $\mathcal{B}_{000}=2\sqrt{1+\tau}>2$ if $\tau>0$ and Bell discord
$\mathcal{G}=\frac{4\tau}{\sqrt{1+\tau}}>0$ if $\tau>0$.

\begin{figure}[h!]
\centering
\includegraphics[scale=0.9]{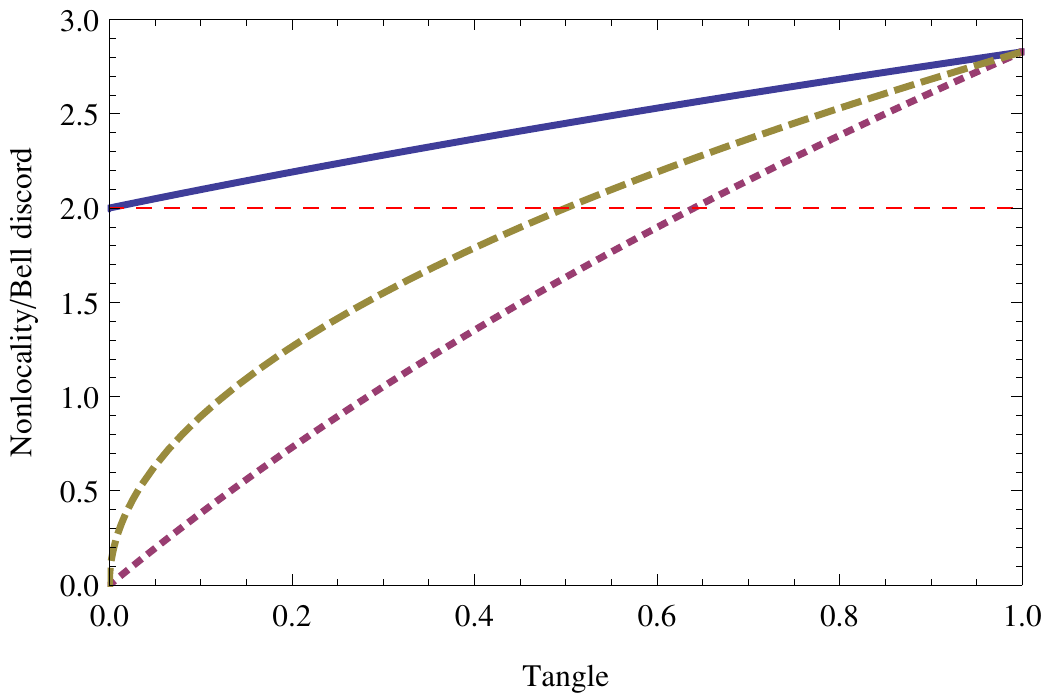}
\caption[Bell discord vs tangle]{Dashed line shows the plots of the Bell-CHSH inequality violation
and Bell discord for the box given in Eq. (\ref{BSb}). Solid and dotted lines show the plots of the Bell-CHSH inequality violation and
Bell discord respectively for the box given in Eq. (\ref{PRQ}).
We observe that the box in Eq. (\ref{PRQ}) gives optimal violation of the Bell-CHSH inequality, however, it does not give optimal Bell discord as this box
has less Bell discord than the box in Eq. (\ref{BSb}).}\label{plotine0}
\end{figure}

Notice that the box in Eq. (\ref{PRQ}) has less irreducible PR-box component than the box in Eq. (\ref{BSb}) for a given amount of entanglement
quantified by the tangle (see fig. \ref{plotine0}). Thus,
when the pure nonmaximally entangled states give rise to optimal violation of the Bell-CHSH inequality, the box does not have
optimal Bell discord and has nonmaximally mixed marginals.

\subsection{Werner states}
Consider the Werner states,
\be
\rho_W=p\ketbra{\psi^+}{\psi^+}+(1-p)\frac{\openone}{4}, \label{MQDs}
\ee
which are entangled iff $p>\frac{1}{3}$ \cite{Werner}. It is known that
the Werner states have nonzero quantum discord if $p>0$ \cite{OZ}. Similarly, we show that the Werner states can have Bell discord
if $p>0$. Notice that the separable Werner states admit a decomposition with an irreducible maximally entangled state component,
just like the local isotropic PR-box which admits a decomposition with an irreducible PR-box component.

For the orthogonal measurement settings that gives rise to the optimal Bell discord for the pure states given in Eq. (\ref{BSb}),
the Werner states give rise to the isotropic PR-box as follows,
\be
P=\frac{p}{\sqrt{2}} P_{PR}+\left(1-\frac{p}{\sqrt{2}}\right)P_N. \label{BW}
\ee
The above box violates the Bell-CHSH inequality if $p^2>\frac{1}{2}$ and has Bell discord $\mathcal{G}=2\sqrt{2p^2}>0$ if $p>0$.
Notice that Bell discord of the local box in Eq. (\ref{BW}) is due to the incompatible measurements performed on
the entangled states which cannot give rise to the violation of a
Bell-CHSH inequality or the separable nonzero quantum discord states.

It has been shown that quantum correlation in mixed states quantified by quantum discord plays the role of entanglement in pure states and
the Werner states are maximally quantum-correlated states \cite{GILU}.
Similarly, we observe that the boxes arising from the Werner states
in Eq. (\ref{BW}) have analogous behavior of the boxes arising from the pure states in
Eq. (\ref{BSb}):
\begin{observation}
When the pure entangled states and the Werner states give rise to optimal Bell discord,
the component of irreducible maximally entangled state, $p$, i.e., quantum discord of the mixed states
plays the same role as the concurrence \cite{WKW}, $\mathcal{C}=\sin2\theta$, i.e., entanglement of the pure states.
\end{observation}

\subsection{Mixed nonmaximally entangled states}
We consider the correlations arising from the mixed states that can be written as a mixture of the Bell state and the classically-correlated state,
\be
\rho= p \ketbra{\psi^+}{\psi^+}+(1-p) \rho_{CC}, \label{rCC}
\ee
where $\rho_{CC}=\frac{1}{2}(\ketbra{00}{00}+\ketbra{11}{11})$.
We illustrate that for the measurements that give rise to optimal Bell discord, these states have the same behavior as the Werner states,
and, for the measurements that give rise to optimal Bell nonlocality, these states and the pure states in Eq. (\ref{nmE})
have similar behavior:

For the settings that give rise to the noisy PR-box in Eq. (\ref{BSb}),
the correlations arising from the states in Eq. (\ref{rCC}) have the same decomposition as for the box arising from the Werner state
in Eq. (\ref{BW}) as the classically-correlated
state in Eq. (\ref{rCC}) gives rise to white noise for this settings. Therefore,
the correlations violate the Bell-CHSH inequality if $p>\frac{1}{\sqrt{2}}$ and have Bell discord $\mathcal{G}=2\sqrt{2}p>0$ if
$p>0$.

For the settings
${\vec{a}_0}=\hat{z}$, ${\vec{a}_1}=\hat{x}$,
${\vec{b}_0}=\cos t\hat{z}+\sin t\hat{x}$ and ${\vec{b}_1}=\cos t\hat{z}-\sin t\hat{x}$, where $\cos t=\frac{1}{\sqrt{1+p^2}}$, the correlations arising
from the mixed states in Eq. (\ref{rCC})
violate the Bell-CHSH inequality i.e., $\mathcal{B}_{000}=2\sqrt{1+p^2}>2$ if $p>0$ and have Bell discord
$\mathcal{G}=\frac{4p^2}{\sqrt{1+p^2}}$. Thus, these correlations have analogous properties of the box arising from the pure states
in Eq. (\ref{PRQ});
the parameter, $p$, in the mixed entangled states plays the role of the parameter, $\sin2\theta$, of the pure states.

\section{Appendix}

\subsection{Bell function monogamy}\label{mBF}
The observation that each Bell-CHSH inequality is violated to the algebraic maximum  by only one PR-box
and a nonlocal correlation cannot violate more than a Bell-CHSH inequality suggests
trade-off between the
Bell functions,
\ba
\mathcal{B}_{\alpha\beta} &:=&|\braket{A_0B_0}+(-1)^{\beta }\braket{A_0B_1}+(-1)^{\alpha}\braket{A_1B_0}\nonumber \\
&&+(-1)^{\alpha \oplus \beta  \oplus 1} \braket{A_1B_1}|. \label{MBF}
\ea
\begin{observation}
For any given nonsignaling box, $P(a_m,b_n|A_i,B_j)$, the Bell functions in Eq. (\ref{MBF}) satisfy the monogamy relationship,
\be
\mathcal{B}_{00}+\mathcal{B}_{j}\le4, \quad \forall j=01,10,11. \label{BFm}
\ee
\end{observation}
\begin{proof}
Since $\mathcal{B}_{\alpha\beta}\le2$ for all the local boxes, the trade-off relations in Eq. (\ref{BFm}) are satisfied by any correlation in the Bell polytope.
It is obvious that all the eight PR-boxes satisfy the trade-off since for any PR-box only one of the Bell functions attains the value $4$ and the rest of them are zero.
Geometrically, any correlation in the nonlocal region  lies on a line joining a PR-box and a Bell-local box which lies on the facet
of the local polytope i.e.,
any nonlocal correlation can be decomposed as follows,
\be
P_{NL}=pP^{\alpha\beta\gamma}_{PR}+(1-p)P_L, \label{GNLd}
\ee
where $P_L$ gives the local bound of a Bell-CHSH inequality.
Now we consider the nonlocal correlations which maximize the
left hand side of the trade-off in Eq. (\ref{BFm}); for instance, any convex mixture of the PR-box
and the deterministic box, $P=pP^{000}_{PR}+(1-p)P^{0000}_D$,
gives $\mathcal{B}_{00}+\mathcal{B}_{j}=4$, $\forall j=01,10,11$.
\end{proof}
The Bell function monogamy given in Eq. (\ref{BFm}) refers to the monogamy of a given correlation with respect to the different Bell-CHSH inequalities,
whereas the conventional monogamy refers to the
monogamy of a given Bell-type inequality with respect to the different marginal correlations of a given multipartite correlation \cite{Bellmono}.

\subsection{Proof of theorem \ref{thm1}}\label{pr1}
Before we show that any NS box can be written as a convex mixture of an irreducible PR-box and a local with $\mathcal{G}=0$, we make the following
observations.

\begin{observation}\label{umPR}
The unequal mixture of any two PR-boxes: $pP^i_{PR}+qP^j_{PR}$, here $p>q$, can be written as the mixture of an irreducible PR-box and a Bell-local box.
\end{observation}
\begin{proof}
$pP^i_{PR}+qP^j_{PR}=(p-q)P^i_{PR}+2qP^{ij}_l$. Here $P^{ij}_l=\frac{1}{2}(P^i_{PR}+P^j_{PR})$ is a Bell-local box since uniform mixture of any two PR-boxes
does not violate a Bell-CHSH inequality.
Notice that the second PR-box, $P^j_{PR}$,
in the unequal mixture is not irreducible as its presence vanishes by the uniform mixture in the other possible decomposition.
\end{proof}

\begin{observation}\label{Girre}
$\mathcal{G}$ calculates the irreducible PR-box component in the mixture of the $8$ PR-boxes: $\sum^7_{k=0} p_k P^k_{PR}$ given in Eq. (\ref{CHNS}).
\end{observation}
\begin{proof}
Notice that $P^{k+1}_{PR}$ is the anti-PR-box to $P^{k}_{PR}$ with $k=0,2,4,6$ since uniform mixture of these two PR-boxes gives white noise.
The evaluation of $\mathcal{G}_1$ for the mixture of the $8$ PR-boxes gives,
\begin{align}
\mathcal{G}_1\left(\sum_k p_k P^k_{PR}\right)&=4|\Big||p_0-p_1|-|p_2-p_3|\Big|\nonumber \\
&-\Big||p_4-p_5|-|p_6-p_7|\Big||.
\end{align}
The observation \ref{umPR} implies that the terms $|p_k-p_{k+1}|$ in this equation give the irreducible PR-box component
in the mixture of the two PR-boxes whose equal mixture gives white noise. Thus,
$\left(\min_i\mathcal{G}_i\left(\sum_k p_k P^k_{PR}\right)\right)/4$ gives the irreducible PR-box component in the mixture of
the $4$ reduced components of the PR-boxes that does not contain any anti-PR-box.
\end{proof}
\begin{observation}\label{nllg0}
Any NS box can be decomposed in a convex mixture of a nonlocal box and a local box with $\mathcal{G}=0$,
\be
P=\eta P_{NL}+ (1-\eta)P_L^{\mathcal{G}=0}. \label{cNlL}
\ee
\end{observation}
\begin{proof}
Since the set of NS boxes is convex and the Bell polytope is contained inside the full NS polytope,
any NS box lies on a line segment joining a nonlocal box and a local box.
Suppose the local box in the decomposition given in Eq. (\ref{cNlL}) has $\mathcal{G}>0$,
then it cannot represent all the $\mathcal{G}=0$ boxes. Thus, the division of the Bell polytope into a
$\mathcal{G}>0$ region and $\mathcal{G}=0$ region allows us to write any NS box as a convex mixture of a nonlocal box and a local box with $\mathcal{G}=0$.
\end{proof}

We now rewrite the decomposition of any NS box given in Eq. (\ref{CHNS}) as a convex combination of the $8$ PR-boxes and a restricted local box
that cannot be written as a convex sum of the PR-boxes and the deterministic boxes:
\be
P=\sum^7_{k=0} g_k P^k_{PR} +\left(1-\sum^7_{k=0} g_k\right)P_L; \quad k=\alpha\beta\gamma, \label{step1}
\ee
where $P_L\ne \sum_k r_k P^k_{PR}+\sum_l s_l P^l_D$, i.e., $P_L$ cannot have nonzero $r_k$ overall possible decompositions.
We wish to reduce the combination of the $8$ PR-boxes in Eq. (\ref{step1}) to the mixture of an irreducible PR-box and a local box
by using the procedure given in observation \ref{umPR}.
It follows from the observation \ref{Girre} that we should first reduce the mixture of the $8$ PR-boxes
to the mixture of the $4$ PR-boxes which does not contain any anti-PR-box, and white noise. Then,
we further reduce it to the mixture of an irreducible PR-box and the local boxes which are the uniform mixture of the two PR-boxes:
\be
\sum^7_{k=0}g_kP^{k}_{PR}=\mu P^{\alpha\beta\gamma}_{PR}+\sum^3_{l=1} p_lP^l_L+p_NP_N. \label{step2}
\ee
Here $\mu$ is obtained by minimizing the PR-box component over all possible decompositions,
i.e., $\mu>0$ iff $\sum^7_{k=0}g_kP^{k}_{PR}\ne\sum^3_{l=1} q_lP^l_L+p_NP_N$.
Now substituting Eq. (\ref{step2}) in Eq. (\ref{step1}),
we get the following decomposition of any NS box,
\be
P=\mathcal{\mu} P^{\alpha\beta\gamma}_{PR}+(1-\mu)P_L.  \label{proofgnz}
\ee
Here
\be
P_L=\frac{1}{1-\mathcal{\mu}}\left\{\sum^3_{l=1}p_lP^l_L+p_NP_N+\left(1-\sum_k g_k\right)P_L\right\}.\nonumber
\ee
This local box cannot have an irreducible PR-box component since $\mu$ is the maximal irreducible PR-box component.
Further, it follows from the observation \ref{nllg0} that the local box in Eq. (\ref{proofgnz}) must have $\mathcal{G}=0$. This ends the proof of the theorem \ref{thm1}.

\chapter{Mermin discord and $3$-decomposition of bipartite NS boxes}\label{Ch3}
\section*{Abstract}
We introduce the measure, Mermin discord, to characterize nonclassicality of bipartite 
quantum correlations originating from EPR-steering. We obtain a $3$-decomposition that 
any bipartite box with two binary inputs and two binary outputs can be decomposed into 
Popescu-Rohrlich (PR) box, a maximally local box, and a local box with Bell and Mermin 
discord equal to zero. Bell and Mermin discord quantify two types of nonclassicality of 
correlations arising from all quantum correlated states which are neither classical-quantum 
states nor quantum-classical states. We show that Bell and Mermin discord serve us the 
witnesses of nonclassicality of local boxes at the tomography level, i.e., nonzero value 
of these measures imply incompatible measurements and nonzero quantum discord by assuming 
the dimensionality and which measurements are performed. The $3$-decomposition serves us to 
isolate the origin of the two types of nonclassicality into a PR-box and a maximally local 
box which is related to EPR-steering, respectively. We study a quantum polytope that has an 
overlap with all the four regions of the full NS polytope to figure out the constraints of quantum correlations.
\section{Introduction}
EPR-steering is  a form  of quantum nonlocality  which is  weaker than
Bell   nonlocality   \cite{WJD}.     Quantum   correlations   exhibit
EPR-steering if they  cannot be described by the  hybrid LHV-Local Hidden
State  (LHS) model  \cite{EPRsi}.   EPR-steering is  witnessed by  the
violation  of steering  inequalities  \cite{CJWR,EPRsi,CFFW}.  Both  
incompatible measurements and  entanglement are necessary  for the violation  
of an EPR-steering  inequality.   EPR-steerablity,   i.e.,  violation  of  a
steering inequality is a  resource for semi-device-independent quantum
key distribution \cite{SDIQKD}.

In Chapter \ref{Ch2}, we have seen that local qubit correlations which
have Bell discord can arise from incompatible measurements. 
If a local box has zero Bell discord, it does not necessarily imply that 
it cannot arise from incompatible measurements on an entangled state. 
There are measurement correlations which have LHV model, nevertheless, 
violate an EPR-steering inequality when they arise from two-qubit systems. 
Therefore, both incompatible measurements and entanglement are necessary 
to produce these local boxes using two-qubit systems.

In this chapter, we  introduce the measure Mermin discord  to characterize  
quantum  correlations going beyond EPR-steering.
We observe that Bell and  Mermin discord divide the  full NS polytope
into  four regions  depending on  whether Bell  discord and/or  Mermin
discord is zero. This division of  the NS polytope allows us to obtain
a $3$-decomposition  of any  NS box. This  decomposition allows  us to
isolate the origin  of nonclassicality into three  disjoint sources: a
PR-box, a maximally  local box which exhibits  EPR-steerability, and a
classical box. We show that all quantum correlated states which have nonzero 
left and right quantum discord \cite{Dakicetal} can give rise to nonclassical correlations which 
have nonzero Bell and/or Mermin discord for suitable incompatible measurements.  
\section{Mermin discord}\label{MD}
A quantum box is EPR-steerable from Alice to Bob if it cannot be described by the hybrid LHV-LHS model,
\begin{align}
P(a_m,b_n|A_i,B_j)=\sum_\lambda P(\lambda) P(a_m|A_i,\lambda)P_Q(b_n|B_j,\lambda),
\end{align}
where $P_Q(b_n|B_j,\lambda)=\tr\rho_\lambda \mathcal{M}_{B_j}^{b_n}$ is the distribution arising from a quantum state $\rho_\lambda$. 
Consider the following EPR-steering inequality,
\be
\braket{A_0B_0}-\braket{A_1B_1}\le \sqrt{2}, \label{eprst}
\ee
where $B_0=\sigma_x$ and $B_1=\sigma_y$ \cite{EPRsi}. Those local boxes that violate this steering inequality cannot
have the LHV-LHS model in which Alice and Bob have access to black-box measurements and projective qubit measurements, respectively, to 
simulate the measurement correlations \cite{SDIQKD}.

For the incompatible measurements: 
$A_0=\sigma_x$, $A_1=\sigma_y$, $B_0=\sigma_x$ and $B_1=\sigma_y$,  the Bell state, $\ket{\psi^+}$, does not give rise to Bell nonlocality, 
however, it gives rise to the violation of the EPR-steering inequality in Eq. (\ref{eprst}).
For this choice of measurements, 
the Bell state gives rise to the following maximally local box,
\begin{equation}
P_M = \left( \begin{array}{cccc}
\half & 0 & 0 & \half \\
\qua & \qua & \qua & \qua \\
\qua & \qua & \qua & \qua \\
0 & \half & \half &  0
\end{array} \right).
\label{eq:merminbox0}
\end{equation}
We call a box that gives the local bound of a Bell-CHSH inequality in Eq. (\ref{BCHSH1}), i.e., $\mathcal{B}_{\alpha\beta\gamma}=2$, maximally local.
Further, the above box is maximally EPR-steerable in that it violates the EPR-steering inequality maximally.
Notice that the following maximally local and correlated box,
\begin{equation}
P_{CC} = \left( \begin{array}{cccc}
\half & 0 & 0 & \half \\
0 & \half & \half & 0 \\
\half & 0 & 0 & \half \\
0 & \half & \half &  0
\end{array} \right),
\label{eq:ccbox}
\end{equation}
is not EPR-steerable since it cannot arise from incompatible measurements on an entangled two-qubit state.
We refer to a maximally local and correlated box which is EPR-steerable as Mermin box.  

The Mermin box in Eq. (\ref{eq:merminbox0}) can also arise from a classically-correlated state in higher 
dimensional space for compatible measurements \cite{NLRan}. However, if one of the subsystem is restricted to be qubit, the Mermin box arises from
a maximally entangled two-qubit state as it can violate the EPR-steering inequality maximally. 
Thus, the violation of the steering inequality in Eq. (\ref{eprst}) implies the presence of 
entanglement in the local boxes in a semi-device-independent way \cite{SDIQKD}.

Consider isotropic Mermin box which is the convex mixture of the Mermin box in Eq. (\ref{eq:merminbox0}) and white noise,
\be
P=p P_M+(1-p)P_N. \label{Mmot}
\ee
For incompatible measurements that lead to the maximal violation of the EPR-steering inequality in Eq. (\ref{eprst}), the nonmaximally entangled states in Eq. (\ref{nmE})
give rise to the isotropic Mermin box with $p=\sin2\theta$. Analogous to the isotropic PR-box,
the isotropic Mermin box arising from the pure entangled states, $\ket{\psi(\theta)}$, violates the EPR-steering inequality  
if $\sin2\theta>\frac{1}{\sqrt{2}}$. However,
it has the irreducible Mermin box component whenever the state is entangled.
Thus, the isotropic Mermin box illustrates the following observation. 
\begin{observation}
When local boxes arising from  entangled two-qubit states have an irreducible Mermin box component, 
the measurements that give rise to them are incompatible.
\end{observation}
Notice that the isotropic Mermin-box has zero Bell discord, i.e., it has $\mathcal{G}=0$. 
The observation that the local boxes which have neither Bell discord nor EPR-steerablity can arise from incompatible measurements on entangled states motivates to define a
notion of nonclassicality which we call Mermin discord.
\begin{definition}\label{MDdef}
A box arising from incompatible measurements on a two-qubit state has \textit{Mermin discord} 
if it admits a decomposition with an irreducible Mermin box component.  
\end{definition}
We observe that the isotropic Mermin box can have EPR-steerablity only when the Mermin box component 
is larger than a certain amount. Thus, analogous to the statement that Bell discord and Bell nonlocality are inequivalent, 
we have the observation that Mermin discord is not equivalent to EPR-steering. 

\begin{figure}[h!]
\centering
\includegraphics[scale=0.30]{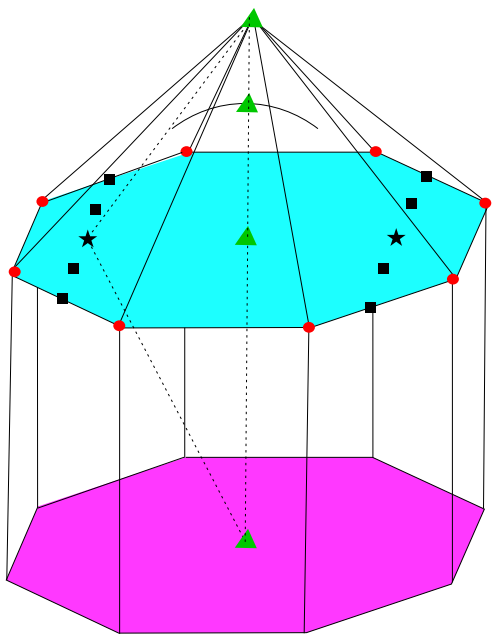}
\caption[PR-box polytope]{A three-dimensional representation of the NS polytope with two binary inputs and two binary outputs is shown here. The octagonal cylinder represents the
local polytope. The lines connecting the deterministic boxes represented by
red points define one of the facet for the local polytope;
the PR-box which violates the Bell-CHSH inequality corresponding
to this facet is represented by triangle point on the top of the NS polytope.
The region below the curved surface contains quantum correlations and the point on this curved surface is the Tsirelson box.
The star and square points on the facet of the local polytope represent quantum and nonquantum Mermin boxes respectively.
The triangular region (shown by dotted lines) which is a convex hull
of the PR-box, the Mermin box, and white noise represents the $3$-decomposition fact that
any point that lies inside the triangle can be decomposed into PR-box,
the Mermin-box and white noise. The line connecting the PR-box and white noise represents the isotropic PR-box and the line joining the Mermin
box and white noise represents the isotropic Mermin box.}\label{NS3dfig}
\end{figure}

We consider the following Mermin inequalities:
\begin{align}
\mathcal{M}_{\alpha\beta\gamma}&:=
(\alpha\oplus\beta\oplus1)\{(-1)^{\gamma}\braket{A_0B_0}\!+\!(-1)^{\alpha\oplus\beta\oplus\gamma\oplus1}\braket{A_1B_1}\}\nonumber\\ 
&+(\alpha\oplus\beta)\{(-1)^{\beta\oplus\gamma}\braket{A_0B_1}+(-1)^{\alpha\oplus\gamma}\braket{A_1B_0}\}\nonumber\\
  &\le2 \quad \text{for} \quad \alpha\beta\gamma=00\gamma,01\gamma;\nonumber \\
\mathcal{M}_{\alpha\beta\gamma}&:=(\alpha\oplus\beta)\{(-1)^{\gamma}\braket{A_0B_0}+(-1)^{\alpha\oplus\beta\oplus\gamma\oplus1}\braket{A_1B_1}\}\nonumber\\ 
&+(\alpha\oplus\beta\oplus1)\{(-1)^{\beta}\braket{A_0B_1}\!+\!(-1)^{\alpha}\braket{A_1B_0}\}\nonumber\\
  &\le2 \quad \text{for} \quad \alpha\beta\gamma=10\gamma,11\gamma. \label{bimi}
\end{align}
The left-hand side of the EPR-steering 
inequality in Eq. (\ref{eprst}) is one of the Mermin operators, $\mathcal{M}_{\alpha\beta\gamma}$, in the above inequalities. 
The multipartite generalization of $\mathcal{M}_{\alpha\beta\gamma}$ generate the Mermin inequalities \cite{mermin,WernerWolfmulti}, hence the name.
Just as the complete set of Bell-CHSH inequalities, the set of these Mermin inequalities 
is invariant under LRO and thus it forms a complete set \cite{WernerWolf}.

Consider the following $8$ maximally local boxes:
\begin{align}
P_M^{\alpha\beta\gamma}(a_m,b_n|A_i,B_j)=\left\{
\begin{array}{lr}
\frac{1}{4}, & i\oplus j =1 \\
\frac{1}{2}, & m\oplus n=i\cdot j \oplus  \alpha  i \oplus  \beta j \oplus \gamma\\ 
0 , & \text{otherwise},\nonumber\\
\end{array}
\right.   
\end{align}
here $\alpha\beta\gamma=00\gamma,10\gamma$, and,
for $\alpha\beta\gamma=01\gamma,11\gamma$,
\begin{align}
P_M^{\alpha\beta\gamma}(a_m,b_n|A_i,B_j)=\left\{
\begin{array}{lr}
\frac{1}{4}, & i\oplus j =0 \\
\frac{1}{2}, & m\oplus n=i\cdot j \oplus  \alpha  i\oplus \beta j \oplus  \gamma  \\ 
0 , & \text{otherwise},\\
\end{array} \label{Mmmm0}
\right.   
\end{align}
which are the equal mixture of the four deterministic boxes. These boxes can be obtained from the Mermin box in Eq. (\ref{eq:merminbox0}) by LRO.
Thus, there are $8$ Mermin-boxes which can have maximal EPR-steerability.
Just as there exists the correspondence between the $8$ PR-boxes and the $8$ Bell-CHSH inequalities, there exists the correspondence between the $8$ Mermin boxes
and the $8$ Mermin operators, $\mathcal{M}_{\alpha\beta\gamma}$, in Eq. (\ref{bimi}):
a Mermin box cannot take the algebraic maximum of $2$ for more than one Mermin operator. Notice that the Mermin operators can be written as the uniform mixture 
of two Bell-CHSH operators; for instance, $\mathcal{M}_{000}=\frac{1}{2}\left(\mathcal{B}_{000}+\mathcal{B}_{110}\right)$. Similarly, the Mermin boxes can also be 
decomposed into the uniform mixture of two PR-boxes; for instance, $P^{000}_M=\frac{1}{2}\left(P^{000}_{PR}+P^{110}_{PR}\right)$.

The complete set of bipartite Mermin inequalities in Eq. (\ref{bimi})
do not distinguish between EPR-steerable and non-steerable boxes since the algebraic maximum 
of any Mermin operator, $\mathcal{M}_{\alpha\beta\gamma}$, is $2$ which is equal to the right-hand side of Eq. (\ref{bimi}). 
However, magnitude of the modulus of the Mermin operators, 
$\mathcal{M}_{\alpha\beta}:=|\mathcal{M}_{\alpha\beta\gamma}|$,
serve to construct Mermin discord. Here $\mathcal{M}_{\alpha0}=|\braket{A_0B_0}+(-1)^{\alpha\oplus1}\braket{A_1B_1}|$ and 
$\mathcal{M}_{0\beta}=|\braket{A_0B_1}+(-1)^{\beta}\braket{A_1B_0}|$.

\begin{observation}
For any Mermin box, only one of the Mermin functions, $\mathcal{M}_{\alpha\beta}$, attains $2$ and the rest of them are zero, 
whereas for the deterministic boxes and the PR-boxes, two of the Mermin functions attain $2$ and the other two are zero. 
\end{observation}
This observation leads us to define a measure of Mermin discord similar to the measure of Bell discord.
\begin{definition}\label{defMD}
Mermin discord, $\mathcal{Q}$, is defined as,
\begin{equation}
\mathcal{Q} := \min_j \mathcal{Q}_j,
\end{equation}    
where, $\mathcal{Q}_1=\Big||\mathcal{M}_{00}-\mathcal{M}_{01}|-|\mathcal{M}_{10}-\mathcal{M}_{11}|\Big|$, and $\mathcal{Q}_2$ and $\mathcal{Q}_3$ 
are obtained by permuting $\mathcal{M}_{\alpha\beta}$ in $\mathcal{Q}_1$. Here $0\le\mathcal{Q}\le2$. 
\end{definition}
Mermin discord is constructed such that all the PR-boxes and the deterministic boxes have $\mathcal{Q}=0$, and, 
the algebraic maximum of $\mathcal{Q}$ is achieved by the Mermin boxes, 
i.e., $\mathcal{Q}=2$ for any Mermin box. Further, Mermin discord is invariant
under LRO and permutation of the parties as the set $\{\mathcal{Q}_j\}$ is invariant under these two transformations. 

We consider the following maximally-local box,
\begin{equation}
P^{nm}_M = \left( \begin{array}{cccc}
1 & 0 & 0 & 0 \\
\half  & \half  & 0 & 0 \\
\half & 0 & \half  & 0 \\
0 & \half & \half &  0
\end{array} \right).
\label{eq:nmerminbox}
\end{equation}
Notice that the Mermin box in Eq. (\ref{eq:merminbox0}) and the above box are equivalent with respect to $\braket{A_iB_j}$, i.e., both the  
boxes have $\braket{A_0B_0}=-\braket{A_1B_1}=1$
and $\braket{A_0B_1}=\braket{A_1B_0}=0$. These two maximally local boxes differ by their marginals; the Mermin box in Eq. (\ref{eq:merminbox0}) has 
maximally mixed marginals, whereas the one in Eq. (\ref{eq:nmerminbox}) has nonmaximally mixed marginals. 
\begin{observation}\label{mlq2}
A maximally-local box that has $\mathcal{Q}=2$ is, in general, a convex combination of a maximally mixed marginals Mermin box and the four 
nonmaximally mixed marginals Mermin boxes
which are equivalent with respect to $\braket{A_iB_j}$,
\be
P^{\alpha\beta\gamma}_{\mathcal{Q}=2}=\sum^4_{i=1}p_{M_i}P^{nm}_{M_i}+p_MP^{\alpha\beta\gamma}_M, \label{Q=2box}
\ee
where $P^{nm}_{M_i}$ are the four nonmaximally mixed marginals Mermin boxes which all have the same values for $\braket{A_iB_j}$ and 
$P^{\alpha\beta\gamma}_M=\frac{1}{4}\sum^4_{i=1}P^{nm}_{M_i}$ is one of the eight Mermin boxes in Eq. (\ref{Mmmm0}) which have maximally
mixed marginals.  
\end{observation}
\begin{proof}
Since the two Mermin boxes in Eqs. (\ref{eq:merminbox0}) and (\ref{eq:nmerminbox}) are equivalent with respect to $\braket{A_iB_j}$,
any convex mixture of these two boxes again have $\mathcal{Q}=2$. There are 
four nonmaximally mixed marginals Mermin boxes which are equivalent with respect to $\braket{A_iB_j}$ corresponding to a given maximally mixed marginals Mermin box. 
Thus, any convex mixture of these five Mermin boxes is again a $\mathcal{Q}=2$ box. 
It can be checked that the equal mixture of the four nonmaximally mixed marginals Mermin boxes 
which are equivalent with respect to $\braket{A_iB_j}$ gives the maximally mixed marginals Mermin box.
\end{proof}

\begin{observation}\label{qdg}
$\mathcal{Q}$ divides the $\mathcal{G}=0$ region into a $\mathcal{Q}>0$ region and $\mathcal{G}=\mathcal{Q}=0$ nonconvex region. 
\end{observation}
\begin{proof}
Since all the deterministic boxes have $\mathcal{G}=\mathcal{Q}=0$
and the Mermin boxes have $\mathcal{G}=0$,
the set of $\mathcal{G}=\mathcal{Q}=0$ boxes forms a nonconvex subregion of the $\mathcal{G}=0$ region.
\end{proof}

The division of the $\mathcal{G}=0$ region with respect to $\mathcal{Q}$ allows us to obtain the following canonical decomposition 
of the local boxes with $\mathcal{G}=0$ (see Appendix. \ref{pr2} for details).
\begin{theorem}\label{thm2}
Any local box, $P^{\mathcal{G}=0}_{L}$, which does not have Bell discord can be decomposed into maximally local box with $\mathcal{Q}=2$ 
and a local box with $\mathcal{G}=\mathcal{Q}=0$,
\be
P^{\mathcal{G}=0}_L=\zeta P^{\alpha\beta\gamma}_{\mathcal{Q}=2}+(1-\zeta)P^{\mathcal{G}=0}_{\mathcal{Q}=0}, \label{Qcanonical1} 
\ee
where, $P^{\alpha\beta\gamma}_{\mathcal{Q}=2}$, is the maximally local box with $\mathcal{Q}=2$, $\zeta$ is the maximal irreducible component of this box 
and $P^{\mathcal{G}=0}_{\mathcal{Q}=0}$ is the local box with $\mathcal{G}=\mathcal{Q}=0$.
\end{theorem}

From linearity of $\mathcal{Q}$ with respect to the decomposition given in Eq. (\ref{Qcanonical1}), it follows that 
$\mathcal{Q}(P^{\mathcal{G}=0}_L)=\zeta\mathcal{Q}\left(P^{\alpha\beta\gamma}_{\mathcal{Q}=2}\right)+(1-\zeta)\mathcal{Q}\left(P^{\mathcal{G}=0}_{\mathcal{Q}=0}\right)=2\zeta$.
This implies that the component, $\zeta$, in Eq. (\ref{Qcanonical1}) is invariant under LRO. 
\newline

\section{Mermin discord of two-qubit states}
The following inequalities,
\be
\mathcal{M}_{\alpha\beta\gamma}\le\sqrt{2}, \label{compEPR}
\ee
where $\mathcal{M}_{\alpha\beta\gamma}$ are the Mermin operators given in Eq. (\ref{bimi}), form the complete set of EPR-steering inequalities 
if the measurement operators on Alice's or Bob's side are anti-commuting qubit observables \cite{EPRsi}. Suppose $B_0=\sigma_x$
and $B_1=\sigma_y$, then these inequalities can be obtained from 
the EPR-steering inequality in Eq. (\ref{eprst}) by LRO. 
The local boxes which violate an EPR-steering inequality in Eq. (\ref{compEPR}) are the subset of the local boxes which have Mermin discord.

We will apply Mermin discord to the local boxes arising from the pure entangled states in Eq. (\ref{nmE}) and the Werner states
in Eq. (\ref{MQDs}). A nonzero Mermin discord of the non-steerable boxes originates from incompatible measurements that give rise to EPR-steering. 
We will find that optimal Mermin discord is achieved by the orthogonal measurements which do 
not give rise to Bell nonlocality.

\subsection{Pure entangled states}
(a) For the settings
${\vec{a}_0}=\hat{x}$, ${\vec{a}_1}=\hat{y}$,
${\vec{b}_0}=\hat{x}$ and ${\vec{b}_1}=\hat{y}$, the pure entangled states in Eq. (\ref{nmE}) 
give rise to the noisy Mermin-box which is a mixture of a Mermin box and white noise 
as follows:
\be
P=\sqrt{\tau}\left(\frac{P^{000}_{PR}+P^{110}_{PR}}{2}\right) +(1-\sqrt{\tau}) P_N, \label{MSb}
\ee 
where $\tau=\sin2\theta$. The above box violates 
the EPR-steering inequality, i.e., $\mathcal{M}_{000}=2\sqrt{\tau}>\sqrt{2}$ if $\tau>\frac{1}{2}$ and has Mermin discord $\mathcal{Q}=2\sqrt{\tau}>0$ if $\tau>0$.
Notice that the irreducible Mermin-box component in the non EPR-steerable box in Eq. (\ref{MSb}) is due to the incompatible measurements that gives rise to EPR-steering
and entanglement.

\begin{figure}[h!]
\centering
\includegraphics[scale=0.85]{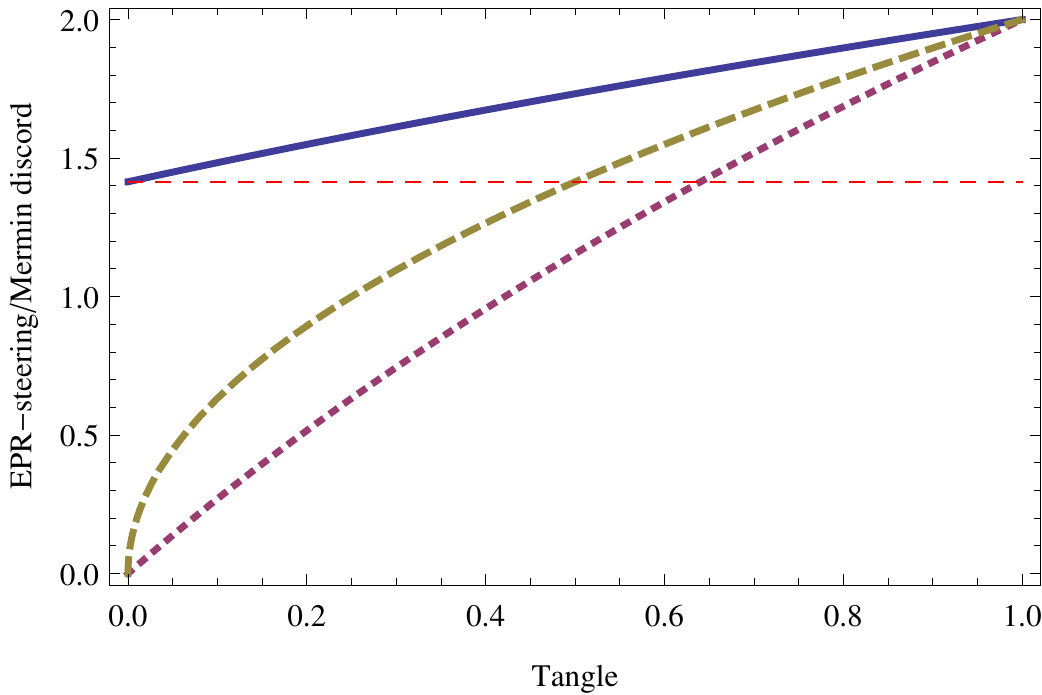} 
\caption[EPR steeering, Mermin discord vs tangle]{Dashed line shows the plots of the EPR-steering violation
and Mermin discord for the box given in Eq. (\ref{MSb}). Solid and dotted lines show the plots of the EPR-steering violation and 
Mermin discord respectively for the box given in Eq. (\ref{CSB}). We observe that the box in Eq. (\ref{CSB}) has less Mermin discord than
the box in Eq. (\ref{MSb}) despite the fact that the former gives rise to optimal violation of the EPR-steering inequality.}\label{figplot2}
\end{figure}

 (b) For the settings, ${\vec{a}_0}=\frac{1}{\sqrt{2}}(\hat{z}+\hat{x})$, ${\vec{a}_1}=\frac{1}{\sqrt{2}}(\hat{z}-\hat{x})$,
${\vec{b}_0}=\cos t\hat{z}+\sin t\hat{x}$, and ${\vec{b}_1}=\cos t\hat{z}-\sin t\hat{x}$, 
where $\cos t=\frac{1}{\sqrt{1+\tau} }$, all the pure entangled states violate the EPR-steering inequality, i.e., 
$\mathcal{M}_{000}=\sqrt{2}\sqrt{1+\tau}>\sqrt{2}$ if $\tau>0$.
For this settings, the box can be decomposed into Mermin box
and a nonmaximally mixed marginals box with $\mathcal{G}=\mathcal{Q}=0$,
\be
P=\nu\left(\frac{P^{000}_{PR}+P^{110}_{PR}}{2}\right)+\left(1-\nu\right)P^{\mathcal{G}=0}_{\mathcal{Q}=0}, \label{CSB}
\ee
where $\nu=\frac{\sqrt{2}\tau}{\sqrt{1+\tau}}$. The  $\mathcal{G}=\mathcal{Q}=0$ box, $P^{\mathcal{G}=0}_{\mathcal{Q}=0}$, in this decomposition 
becomes white noise, $P_N$, for the maximally entangled state. 
The above box has Mermin discord $\mathcal{Q}=\frac{2\sqrt{2}\tau}{\sqrt{1+\tau}}>0$ if $\tau>0$. 

Notice that the box in Eq. (\ref{MSb}) has more irreducible Mermin box component than
the box in Eq. (\ref{CSB}) for a given amount of entanglement (see fig. \ref{figplot2}).
Thus, when the pure nonmaximally entangled states give rise to optimal violation of an EPR-steering inequality, 
the box does not have optimal Mermin discord and has nonmaximally
mixed marginals.

\subsection{Werner states}
For the settings that gives rise to the optimal Mermin discord given in Eq. (\ref{MSb}), the box arising from the Werner states in Eq. (\ref{MQDs}) 
can be decomposed into Mermin box and white noise
as follows,
\be
P=(1-p) P_N+ p\left(\frac{P^{000}_{PR}+P^{110}_{PR}}{2}\right). \label{MW}
\ee 
The above box violates the EPR-steering inequality if $p>\frac{1}{\sqrt{2}}$ and 
has Mermin discord $\mathcal{Q}=2p>0$ if $p>0$. 
Thus, Mermin discord of the the local box in Eq. (\ref{MW}) also detects nonclassicality of the entangled states, which cannot give rise to the violation of an 
EPR-steering inequality, and the separable nonzero quantum discord states.

\section{Bell and Mermin discord vs nonzero quantum discord and incompatibility}\label{BMDiQD}
In the case of two-qubit states and projective measurements, we will show that both incompatible measurements and nonzero left and right quantum discord are necessary
for nonzero Bell/Mermin discord.
\newline

\begin{theorem}\label{thm3}
No compatible measurements on two-qubit states can give rise to nonzero Bell/Mermin discord.  
\end{theorem}
\begin{proof}
Any two-qubit state, up to local unitary equivalence, can be represented as, 
\ba
 \rho_{AB}&=&\frac{1}{4}(\openone_A\otimes\openone_B+\vec{r}\cdot\vec{\sigma}\otimes\openone_B+\openone_A\otimes\vec{s}\cdot\vec{\sigma} \nonumber \\
 &&+\sum\limits_{i=1}^{3}c_{i}\sigma_i\otimes\sigma_i), \label{a2q}
\ea
where the coefficients $c_i=\tr\rho_{AB}\sigma_i\otimes\sigma_i$, $i=x,y,z$, form a diagonal matrix denoted by $C$.
Here $|\vec{r}|^2+|\vec{s}|^2+||C||^2\le3$ with equality holds for the pure states. The expectation value of the above states is given by,
\begin{equation}
\braket{A_iB_j}=\hat{a}_i\cdot C\hat{b}_j.
\end{equation}
Let us calculate $\mathcal{G}$ and $\mathcal{Q}$ for the states given in Eq. (\ref{a2q}) for compatible measurements on Alice's side.
Suppose we choose measurement directions as $\hat{a}_0=\hat{a}_1=\hat{a}$, the measurement observables commute, i.e., $[A_0,A_1]=0$.
For this choice of compatible measurements on Alice's side, $\mathcal{B}_{00}=\mathcal{B}_{01}=2 \hat{a}_0\cdot C\hat{b}_0$, and, 
$\mathcal{B}_{10}=\mathcal{B}_{11}=2 \hat{a}_0\cdot C\hat{b}_1$. This implies that $\mathcal{G}=\mathcal{Q}=0$ for any choice of compatible measurements 
on one side and any choice of compatible/incompatible measurements on the other side. 
\end{proof}

Any separable state which has nonzero left and right quantum discord  cannot be decomposed in the classical-quantum (CQ)
or quantum-classical (QC) form \cite{Dakicetal}.
The CQ states can be written as,
\be 
\rho_{CQ}=\sum^{1}_{i=0}p_i\ketbra{i}{i}\otimes \chi_i \label{c-q},
\ee
whereas QC states can be written as, 
\be
\rho_{QC}=\sum^1_{j=0}p_j\phi_j\otimes \ketbra{j}{j}\label{q-c}.
\ee 
Here $\{\ket{i}\}$ and $\{\ket{j}\}$ are the orthonormal sets, and, $\chi_i$ and  $\phi_j$ are the arbitrary quantum states. 
Despite the CQ and QC states are not the product states in general, their joint expectation value can be written in the factorized form, 
$\braket{AB}= f(\hat{a})f(\hat{b})$, here $\hat{a}$ and $\hat{b}$ are the measurement directions chosen by Alice and Bob respectively.
This factorization of the expectation value for the CQ and QC states 
implies that they cannot have nonzero Bell/Mermin discord for all measurements.
\newline

\begin{theorem}\label{thm4}
All classical-quantum and quantum-classical states have zero Bell and Mermin discord, i.e., $\mathcal{G}=\mathcal{Q}=0$ for all measurements. 
\end{theorem}
\begin{proof}
In the Bloch sphere representation, the CQ
states in Eq. (\ref{c-q}) can be written as:
\begin{eqnarray}
\rho_{CQ}&=&\frac{p_0}{4}\left(\openone
+\hat{r}\cdot\vec{\sigma}\right)\otimes\left(\openone
+\vec{s}_0\cdot\vec{\sigma}\right)\nonumber \\&&+\frac{p_1}{4}\left(\openone
-\hat{r}\cdot\vec{\sigma}\right)\otimes\left(\openone
+\vec{s}_1\cdot\vec{\sigma}\right),
\end{eqnarray}
where $\hat{r}$ is the Bloch vector for the projectors $\ketbra{i}{i}$ and
$\vec{s}_i$ are the Bloch vector for the states $\chi_i$. Notice that $\hat{r}$ appears twice in the above decomposition because of the orthogonality of 
projectors on Alice's side; as a result of this, the expectation value factorizes as follows,
\be
\braket{A_iB_j}=\left(\hat{a}_i\cdot\hat{r}\right) \left(\hat{b}_j\cdot(p_0\vec{s}_0-p_1\vec{s}_1)\right), \label{product-like}
\ee
whose form is similar to that of a product state, \begin{align*}\rho=\rho_A\otimes\rho_B=\frac{1}{4}\left[\left(\openone
+\vec{r}\cdot\vec{\sigma}\right)\otimes\left(\openone
+\vec{s}\cdot\vec{\sigma}\right)\right].\end{align*} 
We have observed that the optimal settings have the following property: 
for the Bell discord one has, $\hat{a}_0\cdot\hat{a}_1=0$, $\hat{b}_0\cdot\hat{b}_1=0$ and $\hat{a}_i\cdot\hat{b}_j=\pm\frac{1}{\sqrt{2}}$, whereas
for the Mermin discord one has: $\hat{a}_0\cdot\hat{a}_1=0$, $\hat{b}_0\cdot\hat{b}_1=0$ and $\hat{a}_i=\pm\hat{b}_j$. 
Since the optimal settings that maximizes $\mathcal{G}$ and $\mathcal{Q}$ have the common property that measurements on
Alice's side or Bob's side are orthogonal, we choose orthogonal measurements on Alice' side
to maximize $\mathcal{G}$ and $\mathcal{Q}$ with respect to the correlation given in Eq. (\ref{product-like}).
Suppose we choose $\hat{a}_0 \cdot \hat{r}=1$, the orthogonality  
condition ($\hat{a}_0\cdot\hat{a}_1=0$) implies that $\hat{a}_1 \cdot \hat{r}=0$. For this choice of orthogonal measurements on Alice's side,   
$\mathcal{B}_{00}=|(\hat{b}_0+\hat{b}_1)\cdot(p_0\vec{s}_0-p_1\vec{s}_1)|$,
$\mathcal{B}_{01}=|(\hat{b}_0-\hat{b}_1)\cdot(p_0\vec{s}_0-p_1\vec{s}_1)|$,
$\mathcal{B}_{10}=|(\hat{b}_0+\hat{b}_1)\cdot(p_0\vec{s}_0-p_1\vec{s}_1)|$, and 
$\mathcal{B}_{11}=|(\hat{b}_0-\hat{b}_1)\cdot(p_0\vec{s}_0-p_1\vec{s}_1)|$ which
implies that $\mathcal{G}=\mathcal{Q}=0$ for all possible measurements on Bob's side. Similarly, we can prove that $\mathcal{G}=\mathcal{Q}=0$ 
for the QC states since $\mathcal{G}$ and $\mathcal{Q}$ are symmetric under the permutation of the parties. 
\end{proof}
Since the joint expectation value of any quantum-correlated state, which has nonzero left and right quantum discord, cannot be written in the factorized form, 
i.e., $\braket{AB}\ne f(\hat{a})f(\hat{b})$, 
all quantum correlated states can give rise to nonzero Bell/Mermin discord for suitable incompatible measurements.

\section{3-decomposition of NS boxes}\label{3-d}
The canonical decomposition given in Eq. (\ref{Gde}) is not the most general one for any given NS box.
Since the canonical decomposition for the boxes with $\mathcal{G}=0$ given in Eq. (\ref{Qcanonical1}) implies that the $\mathcal{G}=0$ box in Eq. (\ref{Gde}) 
can be decomposed into box with $\mathcal{Q}=2$ and a box $\mathcal{G}=\mathcal{Q}=0$,
we obtain the following $3$-decomposition fact of NS boxes.
\begin{theorem}\label{thm5}
Any NS box can be written as a convex mixture of a PR-box, a maximally-local box with $\mathcal{Q}=2$ and a local box with $\mathcal{G}=\mathcal{Q}=0$,
\be
P=\mu P^{\alpha\beta\gamma}_{PR}+\nu P^{\alpha\beta\gamma}_{\mathcal{Q}=2}+(1-\mu-\nu)P^{\mathcal{G}=0}_{\mathcal{Q}=0}. \label{pDecomp}
\ee
\end{theorem}
The $3$-decomposition given above serves as the most general canonical decomposition of the NS boxes as it classifies any given NS box according to whether it has 
Bell or/and Mermin discord.

\textbf{$3$-decomposition of two-qubit states.--} For any given quantum correlated state, there are three types of incompatible measurements which give rise to
(i) $\mathcal{G}>0$ and $\mathcal{Q}=0$ (ii) $\mathcal{G}=0$ and $\mathcal{Q}>0$ and (iii) $\mathcal{G}>0$ and $\mathcal{Q}>0$ ($3$-decomposition).
We will analyze $3$-decomposition of the pure entangled states and the Werner states in order to illustrate the new insights that may be obtained regarding the origin of 
nonclassicality.
\subsection{Maximally entangled state}
When the maximally entangled state gives rise to a nonlocal box which has a $3$-decomposition, 
the box also violates an EPR-steering inequality.
For the measurement settings: 
${\vec{a}_0}=\hat{x}$, ${\vec{a}_1}=\hat{y}$,
${\vec{b}_0} =\sqrt{p}\hat{x}-\sqrt{1-p}\hat{y}$ and ${\vec{b}_1}=\sqrt{1-p}\hat{x}+\sqrt{p}\hat{y}$, where $\frac{1}{2}\le p \le1$, the box
arising from the Bell state, $\ket{\psi^+}=\frac{1}{\sqrt{2}}(\ket{00}+\ket{11})$, 
can be decomposed into PR-box, a Mermin box which is a uniform mixture of two PR-boxes and white noise as follows,
\be
P=\mu P^{000}_{PR}+\nu \left(\frac{P^{000}_{PR}+P^{110}_{PR}}{2}\right)+(1-\mu-\nu)P_N, \label{0meb1}
\ee 
where $\mu=\sqrt{1-p}$ 
and $\nu=\sqrt{p}-\sqrt{1-p}$. 
The above box has Bell and Mermin discord simultaneously when $\frac{1}{2}<p<1$, i.e.,
$\mathcal{G}=4\sqrt{1-p}>0$ if $p\ne1$ and $\mathcal{Q}=2(\sqrt{p}-\sqrt{1-p})>0$ if $p\ne\frac{1}{2}$. 
The box in Eq. (\ref{0meb1}) violates the Bell-CHSH inequality, i.e., $\mathcal{B}_{000}=2\left(\sqrt{p}+\sqrt{1-p}\right)>2$ if
$p\ne1$ and the EPR-steering inequality, i.e., $\mathcal{M}_{000}=2\sqrt{p}>\sqrt{2}$ if $p\ne\frac{1}{2}$. 
Notice that when the settings becomes optimal for the violation of the EPR-steering inequality which happens at $p=1$, 
the PR-box and Mermin box components in the $3$-decomposition go to zero and maximal respectively. 
Thus, the Mermin-box component in the nonlocal box in Eq. (\ref{0meb1}) originates from incompatible measurements that give rise to maximal EPR-steerability.
\subsection{Pure nonmaximally entangled states}
(a) We define the settings:
${\vec{a}_0}=s\hat{x}+c\hat{y}$,
${\vec{a}_1}=c\hat{x}-s\hat{y}$,
${\vec{b}_0}=\frac{1}{\sqrt{2}}(\hat{x}+\hat{y})$ and
${\vec{b}_1}=\frac{1}{\sqrt{2}}(\hat{x}-\hat{y})$, where $s=\sin2\theta$ and $c=\cos2\theta$.
For this state dependent settings, the pure nonmaximally entangled states in Eq. (\ref{nmE}) give rise to a $3$-decomposition as follows,
\be
P=\left(1-\mu-\nu\right) P_N+ \nu\left(\frac{P^{000}_{PR}+P^{11\gamma}_{PR}}{2}\right)+\mu P^{000}_{PR}, \label{0BMSb}
\ee
where
$\nu=|c+s-|c-s||$ and $\mu=\frac{s}{\sqrt{2}}|s-c|$.
The box has nonzero Bell and Mermin discord as follows (see fig. \ref{pt3}), 
\ba
\mathcal{G}&=&2\sqrt{2\tau}|\sqrt{\tau}-\sqrt{1-\tau}|\nonumber \\
&&>0 \quad \text{except when} \quad s  \ne0, \frac{1}{\sqrt{2}} \nonumber
\ea
and 
\begin{align}
\mathcal{Q}&=\sqrt{2}s\Big||c+s|-|c-s|\Big|>0 \quad \text{except when} \quad s \ne0, \frac{1}{2}\nonumber\\
&=\left\{\begin{array}{lr}
2\sqrt{2}\tau \quad \text{when} \quad c>s\\ 
2\sqrt{2\tau(1-\tau^2)} \quad \text{when} \quad s>c.\\ 
\end{array}
\right.\nonumber
\end{align}
Notice that the box in Eq. (\ref{0BMSb}) has only Bell discord when $\theta=\pi/4$ 
since the settings becomes optimal for Bell discord. 
Similarly, it has only Mermin discord when $\theta=\pi/8$ since the settings becomes optimal for Mermin discord.

\begin{figure}[h!]
\centering
\includegraphics[scale=0.85]{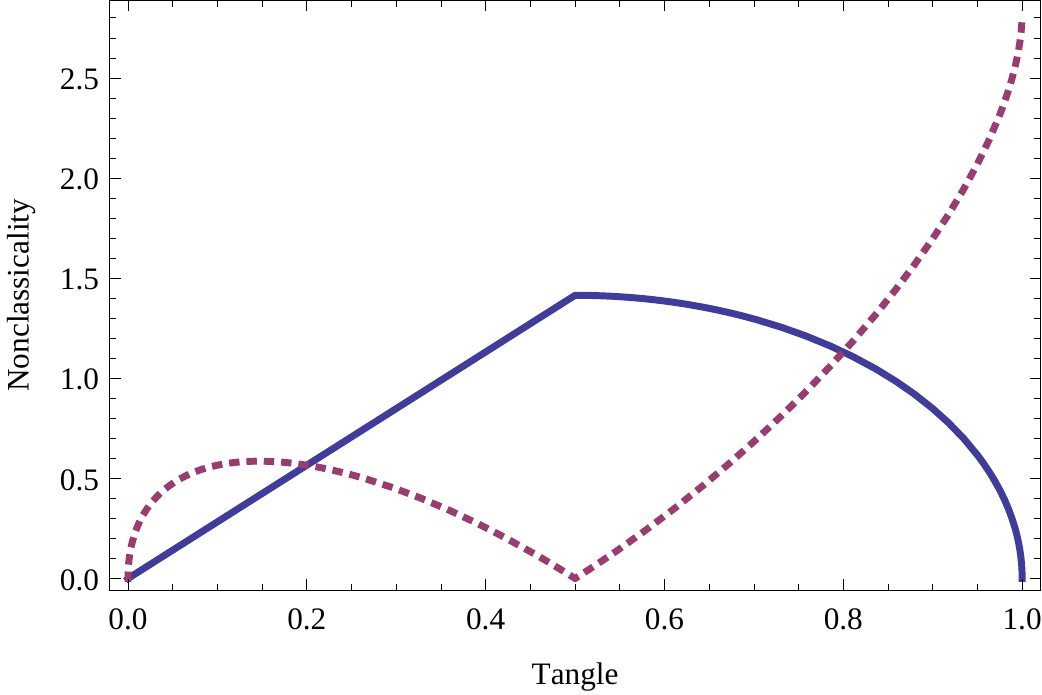} 
\caption[Bell and Mermin discord of the pure two-qubit states having 3-
decomposition]{Bell and Mermin discord of the box given in Eq. (\ref{0BMSb}) are shown by dotted and solid lines respectively.}\label{pt3}
\end{figure}
\begin{figure}
\centering
\includegraphics[scale=0.85]{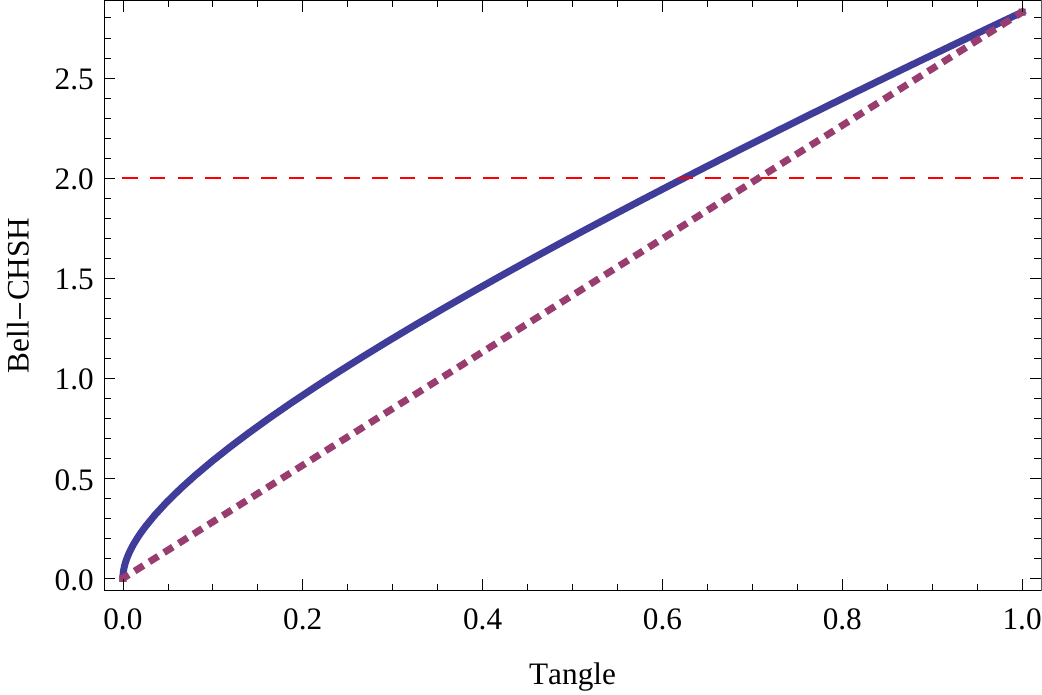} 
\caption[Two different quantum correlation arising from a given pure two-qubit state]{The violation of the Bell-CHSH inequality for the box in Eqs. (\ref{0BMSb}) and (\ref{0BMSb1}) are shown by dotted and solid lines respectively.}\label{pt4}
\end{figure}

(b) For the settings: ${\vec{a}_0}=c\hat{x}+s\hat{z}$,
${\vec{a}_1}=s\hat{x}-c\hat{z}$,
${\vec{b}_0}=\frac{1}{\sqrt{2}}(\hat{x}+\hat{z})$ and
${\vec{b}_1}=\frac{1}{\sqrt{2}}(-\hat{x}+\hat{z})$, the box arising from the pure entangled states has the following $3$-decomposition, 
\be
P=\left(1-\nu-\mu\right) P^{\mathcal{G}=0}_{\mathcal{Q}=0}+\nu\left(\frac{P^{000}_{PR}+P^{11\gamma}_{PR}}{2}\right)+\mu P^{000}_{PR}, \label{0BMSb1}
\ee
where the PR-box and Mermin box components, $\mu$ and $\nu$, 
are the same as for the box given in Eq. (\ref{0BMSb}).
The $\mathcal{G}=\mathcal{Q}=0$ box, $P^{\mathcal{G}=0}_{\mathcal{Q}=0}$, in Eq. (\ref{0BMSb1}) has nonmaximally mixed marginals, whereas the 
$\mathcal{G}=\mathcal{Q}=0$ box in Eq. (\ref{0BMSb})
has maximally mixed marginals. Thus, the boxes in Eqs. (\ref{0BMSb}) and (\ref{0BMSb1}) differ only by their marginals because of this reason
the violation of the Bell-CHSH inequality is larger for the latter box than the former box (see fig. \ref{pt4}).
\subsection{Mixed quantum discordant states} 
For the settings
${\vec{a}_0}=p\hat{x}+\sqrt{1-p^2}\hat{y}$,
${\vec{a}_1}=\sqrt{1-p^2}\hat{x}-p\hat{y}$,
${\vec{b}_0}=\frac{1}{\sqrt{2}}(\hat{x}+\hat{y})$ and
${\vec{b}_1}=\frac{1}{\sqrt{2}}(\hat{x}-\hat{y})$, the Werner states in Eq. (\ref{MQDs}) give rise to 
a $3$-decomposition as follows, 
\be
P=(1-\mu-\nu) P_N+ \nu\left(\frac{P^{000}_{PR}+P^{11\gamma}_{PR}}{2}\right)+\mu P^{000}_{PR}, \label{BMW}
\ee
where $\nu=\frac{p}{\sqrt{2}}|p+\sqrt{1-p^2}-\left|p-\sqrt{1-p^2}\right||$ and $\mu=\frac{p}{\sqrt{2}}\left|p-\sqrt{1-p^2}\right|$.
The box has nonzero Bell and Mermin discord as follows,
\ba
\mathcal{G}&=&2\sqrt{2}p\left|p-\sqrt{1-p^2}\right|\nonumber\\
&&>0 \quad \text{except when} \quad p\ne0, \frac{1}{\sqrt{2}} \nonumber
\ea
and 
\begin{align}
\mathcal{Q}&=\sqrt{2}p\left|p+\sqrt{1-p^2}-\left|p-\sqrt{1-p^2}\right|\right|\nonumber\\
&>0 \quad \text{except when} \quad p\ne0, 1 \nonumber\\
&=\left\{\begin{array}{lr}
2\sqrt{2}p^2 \quad \text{when} \quad 0\le p\le\frac{1}{2}\\ 
2\sqrt{2}\sqrt{p^2(1-p^2)} \quad \text{when} \quad \frac{1}{2}\le p\le1.\\ 
\end{array}
\right.\nonumber
\end{align}

\section{Tsirelson bound}
\begin{figure}[h!]\label{Tsirelson}
\centering
\includegraphics[scale=0.30]{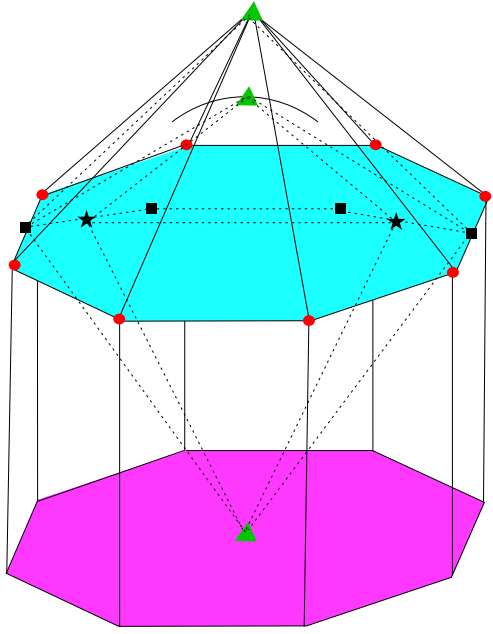} 
\caption[Quantum polytope]{The square and the star points on the facet of the local polytope represent the classically-correlated (CC) boxes  and the quantum Mermin boxes
respectively. 
The subpolytope, $\mathcal{N}_{mm}$, 
formed by the PR-boxes and the CC boxes is represented by the region connecting 
the triangle point on the top, the square points and the triangle point at the centre of the bottom (white noise). 
The subpolytope, $\mathcal{N}_{Tmm}$, whose vertices are the Tsirelson boxes and CC boxes is represented by the region connecting 
the triangle point on the curved surface,
the square points and white noise. The subpolytope, $\mathcal{N}_{Q}$, whose vertices are the Tsirelson boxes and Mermin boxes is represented
by the region connecting the triangle point on the curved surface,
the star points and white noise.
The region connecting the square points and white noise represents the subpolytope, $\mathcal{L}_{mm}$, formed by the CC boxes.
The subpolytope, $\mathcal{L}_{Q}$, formed by the Mermin boxes is represented by the region connecting the star points and white noise.}\label{finalfig}
\end{figure} 
Here we are interested in a restricted NS polytope, $\mathcal{N}_{Q}$, whose vertices are the $8$ Tsirelson boxes,
\be
P^{\alpha\beta\gamma}_{T}=\frac{1}{\sqrt{2}}P^{\alpha\beta\gamma}_{PR}+\left(1-\frac{1}{\sqrt{2}}\right)P_N, \label{tbox}
\ee
and the $8$ quantum Mermin-boxes, $P^{\alpha\beta\gamma}_M$, which are given in Eq. (\ref{Mmmm0}) to figure out the constraints of quantum correlations.
This polytope can be realized by quantum theory  which we illustrate by the correlations arising from the convex mixture of 
the $8$ maximally entangled states,
\be
\rho=\sum^1_{k=0}\sum^1_{j=0}p^j_k \ketbra{\psi^j_k}{\psi^j_k} +\sum^1_{k=0}\sum^1_{j=0} q^j_k\ketbra{\phi^j_k}{\phi^j_k}, \label{ch2bds}
\ee
where $\ket{\psi^j_k}=\frac{1}{\sqrt{2}}(\ket{00}+(-1)^{j} i^{k}\ket{11})$ and $\ket{\phi^j_k}=\frac{1}{\sqrt{2}}(\ket{01}+(-1)^{j} i^{k}\ket{10})$.
For the measurement settings, $\mathcal{M}_{T}$:
\be
{\vec{a}_0}=\hat{x}, \quad {\vec{a}_1}=\hat{y},\quad
{\vec{b}_0}=\frac{1}{\sqrt{2}}(\hat{x}-\hat{y}) \quad \text{and} \quad {\vec{b}_1}=\frac{1}{\sqrt{2}}(\hat{x}+\hat{y}), \label{M_N}
\ee 
the correlation arising from the states in Eq. (\ref{ch2bds}) can be decomposed into $8$ Tsirelson boxes,
\ba
P(\rho,\mathcal{M}_{T})\!&=&\!p^0_0P^{000}_{T}+p^1_0P^{001}_{T}+p^0_1P^{100}_{T}+p^1_1P^{101}_{T}\nonumber\\
&&+q^0_0P^{011}_{T}+q^1_0P^{010}_{T}+q^0_1P^{111}_{T}+q^1_1P^{110}_{T}.\label{8Tb}
\ea
For the measurement settings, $\mathcal{M}_{M}$:   
\be
{\vec{a}_0}=\hat{x}, \quad {\vec{a}_1}=\hat{y}, \quad 
{\vec{b}_0}=-\hat{y} \quad \text{and} \quad {\vec{b}_1}=\hat{x},\label{M_C}
\ee 
the correlation arising from the states in Eq. (\ref{ch2bds}) can be decomposed into $8$ Mermin boxes,
\ba
P(\rho,\mathcal{M}_{M})\!&=&\!p^0_0P^{000}_{M}+p^1_0P^{001}_{M}+p^0_1P^{100}_{M}+p^1_1P^{101}_{M}\nonumber\\
&&+q^0_0P^{011}_{M}+q^1_0P^{010}_{M}+q^0_1P^{111}_{M}+q^1_1P^{110}_{M}.\label{8Mb}
\ea
Since the set of quantum correlations is convex \cite{Pitowski,WernerWolfmulti},
any convex mixture of the two correlations given in Eqs. (\ref{8Tb}) and (\ref{8Mb}),
\be
P=\lambda P(\rho,\mathcal{M}_{T})+(1-\lambda)P(\rho,\mathcal{M}_{M}),
\ee
is also quantum realizable which implies that the polytope $\mathcal{N}_{Q}$
is quantum. 

We obtain the following relationship between the two quantum correlations given in Eqs. (\ref{8Tb}) and (\ref{8Mb}). 
\begin{observation}
For any state given in Eq. (\ref{ch2bds}), Bell discord of the correlation given in Eq. (\ref{8Tb}) is related to the Mermin discord of the correlation 
given in Eq. (\ref{8Mb}) as follows,
\be
\mathcal{G}(P(\rho, \mathcal{M}_T))=\sqrt{2}\mathcal{Q}(P(\rho, \mathcal{M}_{M})). \label{rBMD}
\ee
\end{observation}
\begin{proof}
The Bell functions for the settings given in Eq. (\ref{M_N}) reduce to the Mermin functions for the settings given in Eq. (\ref{M_C}) as follows: 
\ba
\mathcal{B}_{\alpha\beta}&\!=\!&\frac{1}{\sqrt{2}}|\braket{\sigma_x\otimes (\sigma_x+\sigma_y)}+(-1)^\beta\braket{\sigma_x\otimes (\sigma_x-\sigma_y)}
+(-1)^\alpha\braket{\sigma_y\otimes (\sigma_x+\sigma_y)}\nonumber\\
&&+(-1)^{\alpha\oplus\beta\oplus1}\braket{\sigma_y\otimes (\sigma_x-\sigma_y)}|\nonumber\\
&=&\left\{
\begin{array}{lr}
(\alpha\oplus\beta\oplus1)\sqrt{2}|(-1)^{\beta}\braket{\sigma_x\otimes \sigma_x}\!+\!(-1)^{\alpha}\braket{\sigma_y\otimes \sigma_y}| 
\\+(\alpha\oplus\beta)\sqrt{2}|(-1)^{\gamma}\braket{\sigma_x\otimes \sigma_y}+(-1)^{\alpha\oplus\beta\oplus\gamma\oplus1}\braket{\sigma_y\otimes \sigma_x}| \\
=\sqrt{2}\mathcal{M}_{\alpha\beta}     \quad \text{for} \quad \alpha\beta=00, 01,\\
(\alpha\oplus\beta)\sqrt{2}|(-1)^{\beta}\braket{\sigma_x\otimes \sigma_x}\!+\!(-1)^{\alpha}\braket{\sigma_y\otimes \sigma_y}|\\
+(\alpha\oplus\beta\oplus1)\sqrt{2}|(-1)^{\gamma
}\braket{\sigma_x\otimes \sigma_y}+(-1)^{\alpha\oplus\beta\oplus\gamma\oplus1}\braket{\sigma_y\otimes \sigma_x}|  \\
=\sqrt{2}\mathcal{M}_{\alpha\beta}    \quad\! \!\text{for} \quad \alpha\beta=10,11
\end{array}
\right.
 \label{B-M}
\ea
due to the linearity of quantum theory, $\braket{A+B}=\braket{A}+\braket{B}$.
The relationship between the Bell and Mermin functions given in Eq. (\ref{B-M}) implies that 
$\mathcal{G}(\rho, \mathcal{M}_T)=\sqrt{2}\mathcal{Q}(\rho, \mathcal{M}_{M})$.
\end{proof} 
The relationship between Bell and Mermin discord given in Eq. (\ref{rBMD}) implies that the Mermin boxes  
limit nonlocality of the most nonlocal quantum boxes
to the Tsirelson bound since $\mathcal{G}(\rho, \mathcal{M}_T)\le 2\sqrt{2}$ follows from the fact that $\mathcal{Q}(\rho, \mathcal{M}_M)\le2$.

We now discuss the constraints of the quantum region, $\mathcal{N}_{Q}$, inside the full NS polytope. 
Notice that correlations in the region $\mathcal{N}_{Q}$ have maximal local randomness i.e., $\braket{A}_i=\braket{B}_j=0$. If the full NS polytope is constrained by
maximal local randomness, it gives rise to a subpolytope, $\mathcal{N}_{mm}$, whose vertices are the $8$ PR-boxes and $8$ classically-correlated (CC) boxes,
\be
P_{CC}^{\alpha\beta\gamma}(a_m,b_n|A_i,B_j)=\left\{
\begin{array}{lr}
\frac{1}{2}, & m\oplus n=\alpha i \oplus \beta j \oplus \gamma \\ 
0 , & \text{otherwise}.\\
\end{array} \label{CCE}
\right.
\ee  
The polytope, $\mathcal{N}_{Tmm}$, whose vertices are the $8$ Tsirelson boxes and the $8$ CC boxes is obtained by constraining $\mathcal{N}_{mm}$ by the 
Tsirelson inequalities, $\mathcal{B}_{\alpha\beta\gamma}\le2\sqrt{2}$ \cite{tsi1}. The polytope $\mathcal{N}_{Tmm}$ is quantum since its vertices are quantum realizable \cite{Pitowski}. 
Notice that the polytope, $\mathcal{N}_Q$, is contained inside $\mathcal{N}_{Tmm}$ (see fig. \ref{finalfig}).  
Since the Mermin boxes with maximally mixed marginals limits nonlocality of quantum correlations, finding the physical constraints of
$\mathcal{N}_Q$ would help us to single out quantum theory.
The set of local boxes which have maximal local randomness forms a polytope, $\mathcal{L}_{mm}$,
whose vertices are the CC boxes.
Inside this polytope, there exists a polytope, $\mathcal{L}_{Q}$, whose vertices are the $8$ maximally mixed marginals Mermin boxes. 

\section{Conclusions}\label{conc}
We have introduced the measures, Bell discord ($\mathcal{G}$) and Mermin discord ($\mathcal{Q}$), to characterize quantum correlations 
arising from two-qubit states within the framework of GNST. 
We find that when local boxes have nonzero Bell/Mermin discord, they can arise from incompatible measurements on two-qubit states 
which have entanglement in the case of pure states and quantum correlation going beyond entanglement in the case of mixed states. 
Nonzero Bell discord of local boxes which have nonclassicality originates from incompatible measurements that give rise 
to Bell nonlocality. We have observed that there are local boxes which exhibits EPR-steerability. We have introduced Mermin boxes 
which are maximally local and have maximal EPR-steerability. Nonzero Mermin discord of non EPR-steerable boxes which have nonclassicality
originates from incompatible measurements that give rise to EPR-steering.
We have introduced  a $3$-decomposition which allows us to isolate the origin of nonclassicality into three disjoint sources: a PR-box, a Mermin box, 
and a classical box.

We find that all quantum-correlated states which are neither classical-quantum states nor quantum-classical states
can give rise to a $3$-decomposition, i.e., nonzero Bell discord or/and Mermin discord for suitable incompatible measurements.
We find that when pure entangled states and Werner states give rise optimal Bell or Mermin discord, 
quantum correlation quantified by quantum discord in the Werner states
plays a role analogous to entanglement in the pure states.
We have shown that Bell and Mermin discord in general serve as the 
witnesses of nonclassicality of local boxes at the tomography level \cite{Koon}, i.e., nonzero Bell/Mermin discord implies the presence of both nonzero quantum discord 
and incompatible measurements when the dimension of the measured systems is restricted to be $2 \times 2$ and  
measurements performed are restricted to be projective. 
However, we have considered only those boxes with two binary inputs and two binary outputs.
Similarly, it would be interesting to study quantum correlations arising from $d_A \times d_B$ states  
by using NS polytope in which the black boxes have more inputs and more outputs \cite{Barrett,minput}.
In Ref. \cite{JebaTri}, I have generalized Bell and Mermin discord to the multipartite scenario using Svetlichny inequalities 
and Mermin inequalities which detect genuine nonlocality and GHZ paradox \cite{BNL}.
\section{Appendix}
\subsection{Mermin boxes}
Bell polytope admits two types of Mermin boxes which can be distinguished by their marginals. We have found that there are $8$ Mermin boxes which have maximally mixed
marginals. The following $32$ maximally local boxes:
\ba
P_M^{\alpha\beta\gamma\epsilon}=\frac{1}{2}(\delta^i_{m\oplus
  i\oplus\alpha}\delta^j_{n\oplus   j\oplus\beta}   +\delta^i_{m\oplus
  \gamma}\delta^j_{n\oplus\epsilon}),\nonumber\\
   P_{M'}^{\alpha\beta\gamma\epsilon}=\frac{1}{2}(\delta^i_{m\oplus
  i\oplus\alpha}\delta^j_{n\oplus\beta}+\delta^i_{m\oplus
  \gamma}\delta^j_{n\oplus j \oplus \epsilon}), \label{nMmmm} 
\ea 
which are equal mixture of two deterministic boxes, can be obtained from the Mermin box in Eq. (\ref{eq:nmerminbox}) by LRO. Thus, there
are $32$ Mermin boxes with nonmaximally mixed marginals. 
As all the Mermin boxes are maximally-local, they lie on the facet of the Bell polytope (see fig. \ref{NS3dfig}).
\subsection{Proof of theorem \ref{thm2}}\label{pr2}
Since all Mermin boxes have $\mathcal{G}=0$, we obtain the following observation.
\begin{observation}
Any local box with $\mathcal{G}=0$ can be written as a convex mixture of the maximally local boxes with $\mathcal{Q}=2$ and the deterministic boxes,
\be
P^{\mathcal{G}=0}_L=\sum^7_{k=0} p_k P^{k}_{\mathcal{Q}=2}+\sum^{15}_{l=0}q_lP^l_{D}. \label{MNS}
\ee
Here $P^{k}_{\mathcal{Q}=2}$ is one of the maximally local box given in Eq. (\ref{Q=2box}).
\end{observation}
The following observations are useful to show the theorem \ref{thm2}.
\begin{observation}\label{imbc}
The unequal mixture of any two Mermin boxes which differ by $\braket{A_iB_j}$: $pP^1_M+qP^2_M$; $p>q$, 
can be written as a convex mixture of an irreducible Mermin box and a  box with $\mathcal{Q}=0$.  
\end{observation}
\begin{proof}
$pP^1_M+qP^2_M=(p-q)P^1_M+2qP_{\mathcal{Q}=0}$. Here $P_{\mathcal{Q}=0}=\frac{1}{2}(P^1_M+P^2_M)$
is a box with $\mathcal{Q}=0$ since it is a uniform mixture of the two Mermin boxes which differ by $\braket{A_iB_j}$.  
\end{proof}
\begin{observation}\label{Qirre}
$\mathcal{Q}$ calculates the component of irreducible maximally local box with $\mathcal{Q}=2$ 
in the mixture of the $8$ maximally local boxes: $\sum^7_{k=0} p_k P^k_{\mathcal{Q}=2}$ given in Eq. (\ref{MNS}). 
\end{observation}
\begin{proof}
Notice that the uniform mixture of $P^{k}_{\mathcal{Q}=2}$ and $P^{k+1}_{\mathcal{Q}=2}$ with $k=0,2,4,6$ 
gives a zero-expectation box, which has $\braket{A_iB_j}=0$ $\forall i,j$. 
We call $P^{k+1}_{\mathcal{Q}=2}$ anti-Mermin box.
The evaluation of $\mathcal{Q}_1$ for the mixture of the $8$ maximally local boxes gives,  
\begin{align}
\mathcal{Q}_1\left(\sum^7_{k=0} p_k P^k_{\mathcal{Q}=2}\right)&=2|\Big||p_0-p_1|-|p_2-p_3|\Big|\nonumber\\
&-\Big||p_4-p_5|-|p_6-p_7|\Big||.
\end{align}
The observation \ref{imbc} implies that the terms $|p_k-p_{k+1}|$ in this equation give the irreducible maximally local box component  
in the mixture of the two maximally local boxes boxes whose equal mixture gives a zero-expectation box. Thus, 
$\left(\min_i\mathcal{Q}_i\left(\sum^7_{k=0} p_k P^k_{\mathcal{Q}=2}\right)\right)/2$ gives the irreducible component of the box with $\mathcal{Q}=2$ in the mixture of 
the $4$ reduced components of the $\mathcal{Q}=2$ boxes that does not contain any anti-Mermin-box.  
\end{proof}

Let us now prove the theorem \ref{thm2} which goes similar to the proof of the theorem \ref{thm1}.
Any local box with $\mathcal{G}=0$ given by the decomposition in Eq. (\ref{MNS}) can be rewritten as a convex mixture of
the $8$ maximally local boxes which have $\mathcal{Q}=2$ and a local box which does not have the components of the $\mathcal{Q}=2$ boxes,
\be
P^{\mathcal{G}=0}_L=\sum^7_{k=0} q_k P^{k}_{\mathcal{Q}=2}+\left(1-\sum^7_{k=0} q_k\right)P_L, \label{step1m}
\ee
where $P_L\ne \sum_k r_k P^{k}_{\mathcal{Q}=2}+\sum_l s_l P^l_D$, i.e., $P_L$ cannot have nonzero $r_k$ overall possible decompositions. 
It follows from observations \ref{imbc} and  \ref{Qirre} that the mixture of the $8$ maximally local boxes in this decomposition
can be written as the mixture of an irreducible $\mathcal{Q}=2$ box, and the $7$ boxes which are the uniform 
mixture of two $\mathcal{Q}=2$ boxes:
\be
\sum^7_{k=0} q_k P^{k}_{\mathcal{Q}=2}=\zeta P^{\alpha\beta\gamma}_{\mathcal{Q}=2}+ \sum^4_{i=1} t_i P^i_{zc}+\sum^3_{i=1} v_i P^i_{L}.\label{step2m}
\ee
Here $\zeta$ is obtained by minimizing the component of the single maximally local box overall possible decomposition,  $P^i_{zc}$ are the zero-expectation boxes,
and $P^i_{L}$ are the uniform mixture two maximally local boxes which are not the zero-expectation boxes. 
Now substituting Eq. (\ref{step2m}) in Eq. (\ref{step1m}), 
we get the following decomposition of any box with $\mathcal{G}=0$,
\be
P^{\mathcal{G}=0}_L=\zeta P_{\mathcal{Q}=2}+(1-\zeta)P^{\mathcal{G}=0}_{\mathcal{Q}=0}.
\ee
Here 
\be
P^{\mathcal{G}=0}_{\mathcal{Q}=0}=\frac{1}{(1-\zeta)}\Big\{\sum^4_{i=1} t_i P^i_{zc}+\sum^3_{i=1} v_i P^i_{L}+\left(1-\sum^7_{k=0} q_k\right)P_L\Big\}. \nonumber
\ee
This box has $\mathcal{G}=\mathcal{Q}=0$ since it does not have the irreducible Mermin box and PR-box components, i.e., it belongs to the $\mathcal{G}=\mathcal{Q}=0$
region.

\subsection{Linearity of Bell and Mermin discord w.r.t the canonical decompositions}\label{lbdmd}
$\mathcal{G}$ is, in general, not linear for the decomposition of a given correlation into the convex mixture of two $\mathcal{G}>0$ boxes.
For instance, consider a correlation which is the convex mixture of two PR-boxes,
\be
P=pP^i_{PR}+qP^j_{PR}; \quad p>q,\label{nlg>}
\ee
which has $\mathcal{G}(P)=4(p-q)$. Suppose $\mathcal{G}$ is linear for this decomposition, $\mathcal{G}(P)=p\mathcal{G}(P^i_{PR})+q\mathcal{G}(P^j_{PR})=4\ne4(p-q)$.
However, $\mathcal{G}$ is linear for the decomposition of the correlation in Eq. (\ref{nlg>}) into a mixture of a single PR-box and a $\mathcal{G}=0$ box,
\be
P=(p-q)P^i_{PR}+2q\left(\frac{P^i_{PR}+P^j_{PR}}{2}\right). \label{ggnl}
\ee
$\mathcal{G}$ is, in general, also not linear for the decomposition of a correlation into the convex mixture of two $\mathcal{G}=0$ boxes.
For instance, consider the following uniform mixture of two Mermin boxes (the triangle point on the facet of the local polytope in fig. \ref{NS3dfig}),
\be
P=\frac{1}{2}P^{1}_M+\frac{1}{2}P^{2}_M,\label{nggnl}
\ee
where $P^{1}_M=\frac{1}{2}\left(P^{000}_{PR}+P^{111}_{PR}\right)$ and $P^{2}_M=\frac{1}{2}\left(P^{000}_{PR}+P^{110}_{PR}\right)$. Evaluation of $\mathcal{G}$
on the right hand side by using linearity gives $\frac{1}{2}\mathcal{G}(P^{1}_M)+\frac{1}{2}\mathcal{G}(P^{2}_M)=0$, however, $\mathcal{G}(P)=2$.
The correlation in Eq. (\ref{nggnl}) can also be written in the isotropic PR-box form as follows,
\be
P=\frac{1}{2}P^{000}_{PR}+\frac{1}{2}P_N.
\ee
It is obvious that $\mathcal{G}$ is linear for this decomposition. Similarly, we can observe that Mermin discord is, in general, not linear for the 
the decomposition of a given correlation into a mixture of two $\mathcal{Q}>0$ boxes or $\mathcal{Q}=0$ boxes and linear for the canonical decomposition. 

\chapter{On total correlations in bipartite quantum probability distributions}
\label{Ch4}
\section*{Abstract}
We discuss the problem of separating the total correlations in a given quantum probability distribution into nonlocality, contextuality, and classical correlations.
Bell discord and Mermin discord which quantify nonclassicality of quantum correlations going beyond Bell nonlocality and EPR-steering, respectively,
are interpreted as distance measures in the nonsignaling polytope.
A measure of total correlations is introduced to divide the total amount of correlations into a purely nonclassical and a
classical part. We show that quantum correlations arising from the two-qubit states satisfy additivity relations among these three measures.


\section{Introduction}When measurements on an ensemble of entangled particles give rise to the violations of a Bell inequality \cite{bell64,BNL},
one may ask the question of EPR2 \cite{EPR2} whether all the particle pairs in the ensemble behave nonlocally or only some pairs are nonlocally correlated and the other pairs are locally correlated. EPR2 approach to quantum correlation consists
in decomposing a given quantum joint probability distribution into a nonlocal and a local distribution
to find out whether the correlation is
fully nonlocal or it has local content.
EPR2 showed that if the particle pairs are in the singlet state, they all behave nonlocally. However, EPR2 showed that nonmaximally entangled states cannot have nonlocality purely. Thus, total correlations arising from measurements on composite quantum systems
can be divided into a purely nonlocal and a local part.

In Chapters \ref{Ch2} and \ref{Ch3}, Bell discord and Mermin discord have been proposed as measures of quantum correlations 
to quantify nonlocality and EPR-steering of correlations arising from the quantum correlated states \cite{OZ,GTC,Modietal}
and it has been  observed that any bipartite qubit correlation can be decomposed in a convex mixture of an irreducible nonlocal correlation,
an irreducible EPR-steerable correlation and a local correlations which has null Bell and Mermin discord.
This $3$-decomposition fact of quantum correlations suggests that when measurements on an ensemble
of the bipartite quantum system gives rise to Bell and Mermin discord
simultaneously, the ensemble can be divided into a purely nonlocal, an EPR-steerable and a local part which might have classical correlations.

In this work, we discuss the analogous problem of dividing the total correlations in a given quantum state into a purely nonclassical and a classical part
\cite{HV,GPW,Modietal} to quantum joint probability distributions.
We show that Bell discord and Mermin discord defined in Chapters \ref{Ch2} and \ref{Ch3}
can be interpreted as distance measures in the nonsignaling polytope and thus they are analogous to the geometric measure of
quantum discord \cite{Dakicetal}.
Inspired by this interpretation, we define a third distance measure to quantify the amount of total correlations in quantum joint probability distributions.
We study additivity relation for quantum correlations in two-qubit systems.

\section{The three distance measures}\label{measures}
The distance measures are useful tool in quantum information theory to quantify nonclassicality of quantum states and to divide the total correlations
in a given quantum state into a nonclassical and a purely a classical part \cite{Hetal,Modietal,ModietalUVQC}.
In Ref. \cite{ModietalUVQC}, measures of quantum correlations that go beyond entanglement were defined using the concept of distance measures
and it was shown that the distance from a given state to its closest product state
gives total correlations. Similarly, we will propose Bell discord and Mermin discord as distance measures for nonclassicality of quantum correlations
going beyond nonlocality. We will define a distance measure that is nonzero iff a given correlation described by the
joint probability distributions (JPD) is nonproduct to quantify
total correlations in quantum JPD.

Bell-CHSH scenario \cite{chsh} can be abstractly described in terms black boxes shared between two spatially separated observers; Alice and Bob input two variables $A_i$ and
$B_j$ into the box and obtain two distinct outputs $a_m$ and $b_n$ on their part of the box ($i,j,m,n\in\{0,1\}$).
The behavior of a given box is described by the set of $16$ joint
probability distributions (JPD),
\ba
P(a_m,b_n|A_i,B_j)&=&\frac{1}{4}[1+(-1)^m\braket{A_i}+(-1)^n\braket{B_j}\nonumber\\
&&+(-1)^{m\oplus n}\braket{A_iB_j}],
\ea
where $\braket{A_iB_j}=\sum_{m=n}P(a_m,b_n|A_i,B_j)-\sum_{m\ne n}P(a_m,b_n|A_i,B_j)$ are joint expectation values,
and, $\braket{A_i}=P(a_0|A_i)-P(a_1|A_i)$ and $\braket{B_j}=P(b_0|B_j)-P(b_1|B_j)$ are marginal expectation values.
Here $\oplus$ denotes addition modulus $2$.
The set of nonsignaling boxes ($\mathcal{N}$) corresponding to this scenario forms an $8$ dimensional convex polytope which has $24$ extremal boxes \cite{Barrett}:
they are $8$ PR-boxes,
\be
P^{\alpha\beta\gamma}_{PR}(a_m,b_n|A_i,B_j)=\left\{
\begin{array}{lr}
\frac{1}{2}, & m\oplus n=ij \oplus \alpha i\oplus \beta j \oplus \gamma\\
0 , & \text{otherwise},\\
\end{array}
\right. \label{NLV}
\ee
and $16$ deterministic boxes:
\be
P^{\alpha\beta\gamma\epsilon}_D=\left\{
\begin{array}{lr}
1, & m=\alpha i\oplus \beta,\\
& n=\gamma j\oplus \epsilon \\
0 , & \text{otherwise}.\\
\end{array}
\right.
\ee

\subsection{Bell discord}
All the Bell-CHSH inequalities \cite{WernerWolf},
\ba
\mathcal{B}_{\alpha\beta\gamma} &:= &(-1)^\gamma\braket{A_0B_0}+(-1)^{\beta \oplus \gamma}\braket{A_0B_1}\nonumber\\
&+&(-1)^{\alpha \oplus \gamma}\braket{A_1B_0}+(-1)^{\alpha \oplus \beta \oplus \gamma \oplus 1} \braket{A_1B_1}\le2, \label{BCHSH}
\ea
form eight facets for the Bell polytope.
We may consider the eight Bell functions, $\mathcal{B}_{\alpha\beta\gamma}$,
to form the eight orthogonal coordinates for the metric space in which
distance is measured by the modulus of these Bell functions, $\mathcal{B}_{\alpha\beta}:=|\mathcal{B}_{\alpha\beta\gamma}|$.
\begin{observation}
The Bell functions, $\mathcal{B}_{\alpha\beta}$, satisfy the triangle inequality,
\begin{equation}
\mathcal{B}_{\alpha\beta}(P_1,P_2)\le \mathcal{B}_{\alpha\beta}(P_1)+\mathcal{B}_{\alpha\beta}(P_2). \label{Tr}
\end{equation}
\end{observation}
\begin{proof}
Consider the following convex mixture of the two PR-boxes,
\be
P=pP^{000}_{PR}+qP^{001}_{PR},
\ee
which has $\mathcal{B}_{00}(P)=4|p-q|$. Here $\mathcal{B}_{00}(P)$ can be regarded as measuring the distance between the boxes $P_1=pP^{000}_{PR}+(1-p)P_N$ and
$P_2=qP^{001}_{PR}+(1-q)P_N$ which have $\mathcal{B}_{00}(P_1)=4p$ and $\mathcal{B}_{00}(P_2)=4q$. The triangle inequality in Eq. (\ref{Tr})
follows since $\mathcal{B}_{00}(P_1, P_2)=4|p-q|\le \mathcal{B}_{\alpha\beta}(P_1)+\mathcal{B}_{\alpha\beta}(P_2)= 4p+4q=4$.
\end{proof}
The isotropic PR-boxes,
\be
P_{iPR}^{\alpha\beta\gamma}=p_{nl}P^{\alpha\beta\gamma}_{PR}+(1-p_{nl})P_N, \label{isoPR}
\ee
define the eight orthogonal coordinates in which each coordinate is a line joining a PR-box and white noise.
Geometrically for a given box, each $\mathcal{B}_{\alpha\beta}$ measures the distance of a box which is, in general, different than the given box from the origin.
The white noise, $P_N$, which has
$\mathcal{B}_{\alpha\beta\gamma}=0$ is at the origin.
Since a PR-box can lie on top of only one of the facets, the distance of a PR-box from the origin is measured by only one of the Bell functions. For instance, the PR-box, $P^{00\gamma}_{PR}$, gives $\mathcal{B}_{00}=4$ and the other $\mathcal{B}_{\alpha\beta}$ are zero; it is at the largest
distance from the origin. Since the isotropic PR-boxes in Eq. (\ref{isoPR})
lie along only one of the coordinates, they have only one of the Bell function nonzero, i.e.,
$\mathcal{B}_{\alpha\beta}=4p_{nl}$ and the rest of the three Bell functions take zero.
All the four Bell functions measure the distance of any deterministic box simultaneously
since the deterministic boxes have $\mathcal{B}_{\alpha\beta}=2$ for all $\alpha\beta$, i.e., they lie on the hyperplane.

Bell discord, $\mathcal{G}$, is constructed using the Bell functions as follows,
\be
\mathcal{G}=\min_i\mathcal{G}_i,
\ee
where $\mathcal{G}_1=\Big||\mathcal{B}_{00}-\mathcal{B}_{01}|-|\mathcal{B}_{10}-\mathcal{B}_{11}|\Big|$ and $\mathcal{G}_2$ and $\mathcal{G}_3$ are obtained by permuting
$\mathcal{B}_{\alpha\beta}$ in $\mathcal{G}_1$. Here $0\le\mathcal{G}\le4$. The deterministic boxes have $\mathcal{G}=0$, whereas the PR-boxes have $\mathcal{G}=4$.
As Bell discord is made up of $\mathcal{B}_{\alpha\beta}$, it also satisfies the triangle inequality.
\begin{proposition}
If a given nonextremal correlation has an irreducible PR-box component, $\mathcal{G}$ measures how far the given correlation from a local box that does not have an
irreducible PR-box component in the metric space defined by the Bell functions.
\end{proposition}
\begin{proof}
Any NS correlation can be written as a convex combination of an irreducible PR-box and a local box which has $\mathcal{G}=0$ \cite{Jeba},
\be
P=\mathcal{G'}P^{\alpha\beta\gamma}_{PR}+(1-\mathcal{G'})P^{\mathcal{G}=0}_L. \label{gcano}
\ee
This canonical decomposition implies that the correlation that has an irreducible PR-box component lies on the line segment joining the PR-box and the local box
with $\mathcal{G}=0$.
Thus, Bell discord of the correlation in Eq. (\ref{gcano}) given by $\mathcal{G}(P)=4\mathcal{G}'$ gives the
distance of the given correlation from the $\mathcal{G}=0$ box in the canonical decomposition.
\end{proof}

Consider the following convex mixture of the PR-box and the deterministic box,
\be
P=pP^{000}_{PR}+qP^{0000}_D.
\ee
For these correlations, $\mathcal{B}_{000}=p\mathcal{B}_{000}(P^{000}_{PR})+q\mathcal{B}_{000}(P_D)=4p+2q=2(p+1)$ and $\mathcal{G}=4p$.
Notice that, $\mathcal{B}_{00}\ge \mathcal{G}$; $\mathcal{B}_{00}$ measures the distance of the correlation from the origin and is equal to the sum of the distance of the noisy deterministic box, $qP_D+(1-q)P_N$, and the noisy PR-box, $pP_{PR}+(1-p)P_N$, whereas $\mathcal{G}$ measures the distance of the correlation from the
deterministic box and is equal to the distance of the correlation from the origin minus the distance of the noisy deterministic box.

\textit{Bell-CHSH inequality violation versus nonzero Bell discord:-}
For any NS box given by the canonical decomposition in Eq. (\ref{gcano}),
the Bell-CHSH operator $\mathcal{B}_{\alpha\beta\gamma}(P)=4\mathcal{G'}+l(1-\mathcal{G'})$, where $l=\mathcal{B}_{\alpha\beta\gamma}\left(P^{\mathcal{G}=0}_L\right)$.
Consider the case when $l\ge0$. If $\mathcal{G'}>\frac{1}{2}$, it is for sure that the correlation gives
the violation of the Bell-CHSH inequality. Now consider the following two cases.

(i) Suppose $\mathcal{B}_{\alpha\beta\gamma}\left(P^{\mathcal{G}=0}_L\right)=0$, the correlations cannot give 
rise to the violation of the Bell-CHSH inequality when $0\le p\le\frac{1}{2}$.
Therefore, for the violation of the Bell-CHSH inequality upon increasing the PR-box content, 
first the box has to be lifted to the face of the Bell polytope by the PR-box content which happens at $\mathcal{G'}=\frac{1}{2}$.

(ii) Suppose $\mathcal{B}_{\alpha\beta\gamma}\left(P^{\mathcal{G}=0}_L\right)=2$. Then any small amount of the PR-box content will give rise 
to the violation of the Bell-CHSH
inequality because the box lies on the face of the Bell polytope when $\mathcal{G'}=0$.

Thus, the violation of a Bell inequality depends on the amount of
irreducible PR-box content as well as the local box in the
canonical decomposition, whereas nonzero Bell discord depends only on the amount of irreducible PR-box content.
Popescu and Rohrlich showed that all pure entangled states violate a Bell-CHSH inequality \cite{PRQB}. However, there are mixed entangled states that do not violate
a Bell-CHSH inequality \cite{Hetal}.
The reason for the nonviolation of any Bell inequality by some entangled states is that the local box in the canonical decomposition does not have sufficient amount of magnitude for
the Bell operator to lift the correlation to go outside the Bell polytope.
\subsection{Mermin discord}
We may as well consider the eight Mermin functions,
\ba
\mathcal{M}_{\alpha\beta\gamma}&:=
&(\alpha\oplus\beta\oplus1)\{(-1)^{\beta}\braket{A_0B_1}\!+\!(-1)^{\alpha}\braket{A_1B_0}\}\nonumber\\ &&+(\alpha\oplus\beta)\{(-1)^{\gamma}\braket{A_0B_0}+(-1)^{\alpha\oplus\beta\oplus\gamma\oplus 1}\braket{A_1B_1}\}\nonumber\\
&& \text{for} \quad \alpha\beta\gamma=00\gamma,01\gamma;\nonumber \\
\mathcal{M}_{\alpha\beta\gamma}&:=&(\alpha\oplus\beta)\{(-1)^{\beta}\braket{A_0B_1}\!+\!(-1)^{\alpha}\braket{A_1B_0}\}\nonumber\\ &&+(\alpha\oplus\beta\oplus1)\{(-1)^{\gamma}\braket{A_0B_0}+(-1)^{\alpha\oplus\beta\oplus\gamma\oplus1}\braket{A_1B_1}\}\nonumber\\
&& \text{for} \quad \alpha\beta\gamma=10\gamma,11\gamma,
\ea
to form eight orthogonal coordinates for the metric space in which $\mathcal{M}_{\alpha\beta}:=|\mathcal{M}_{\alpha\beta\gamma}|$ serve as the distance function.
The following eight Mermin boxes which have maximally mixed marginals,
\begin{align}
P_M^{\alpha\beta\gamma}(a_m,b_n|A_i,B_j)&=\left\{
\begin{array}{lr}
\frac{1}{4}, & i\oplus j =0 \\
\frac{1}{2}, & m\oplus n=i\cdot j \oplus \alpha i \oplus \beta j \oplus \gamma\\
0 , & \text{otherwise},\nonumber\\
\end{array}
\right. \text{for} \quad \alpha\beta\gamma=00\gamma,10\gamma \\
&=\left\{
\begin{array}{lr}
\frac{1}{4}, & i\oplus j =1 \\
\frac{1}{2}, & m\oplus n=i\cdot j \oplus \alpha i\oplus \beta j \oplus \gamma \\
0 , & \text{otherwise},\\
\end{array} \label{Mmmm}
\right. \text{for} \quad \alpha\beta\gamma=01\gamma,11\gamma
\end{align}
lie along extremum of only one of the coordinates. Therefore, the distance of the isotropic Mermin-boxes,
\be
P_{iM}^{\alpha\beta\gamma}=p_{c}P^{\alpha\beta\gamma}_M+(1-p_{c})P_N, \label{isoM}
\ee
are measured by only one of the Mermin functions.

Mermin discord, $\mathcal{Q}$, is constructed using the Mermin functions as follows,
\be
\mathcal{Q}=\min_i\mathcal{Q}_i.
\ee
Here $\mathcal{Q}_1=\Big||\mathcal{M}_{00}-\mathcal{M}_{01}|-|\mathcal{M}_{10}-\mathcal{M}_{11}|\Big|$ and $\mathcal{Q}_2$ and $\mathcal{Q}_3$ are obtained by permuting
$\mathcal{M}_{\alpha\beta}$ in $\mathcal{Q}_1$. Since the distance of the PR-boxes and the deterministic boxes are simultaneously measured by two Mermin functions (i.e., they
lie on the hyperplane), they have $\mathcal{Q}=0$. The isotropic Mermin boxes in Eq. (\ref{isoM}) have $\mathcal{Q}=2p_c$.
\begin{proposition}
If a given correlation has nonzero Mermin discord, $\mathcal{Q}$ measures the distance of the given correlation from a correlation with $\mathcal{Q}=0$
in the metric space of Mermin functions.
\end{proposition}
\begin{proof}
Any NS correlation can be written as a convex mixture of a maximally local box with $\mathcal{Q}=2$
which lies on extremum of one of the coordinates, $\mathcal{M}_{\alpha\beta}$, and a $\mathcal{Q}=0$ box \cite{Jeba},
\be
P=\mathcal{Q'}P^{\alpha\beta\gamma}_{\mathcal{Q}=2}+(1-\mathcal{Q'})P_{\mathcal{Q}=0}. \label{canoQ}
\ee
This canonical decomposition implies that the correlation that has an irreducible Mermin box component lies on a line segment joining the
$\mathcal{Q}=2$ box and the $\mathcal{Q}=0$ box.
Thus, Mermin discord of the correlation in Eq. (\ref{canoQ}) given by $\mathcal{Q}(P)=2\mathcal{Q}'$ measures the distance of the given correlation from the $\mathcal{Q}=0$ box,
$P_{\mathcal{Q}=0}$, in the canonical decomposition.
\end{proof}
\subsection{$\mathcal{T}$ measure}
The analysis of quantum correlations arising from the two-qubit states done in the last chapter implies that
up to local reversible operations any quantum correlation
can be decomposed into a convex mixture of a PR-box, a Mermin-box, and a restricted local box,
\be
P=\mathcal{G}'P^{000}_{PR}+\mathcal{Q}'\left(\frac{P^{000}_{PR}+P^{11\gamma}_{PR}}{2}\right)+(1-\mathcal{G}'-\mathcal{Q}')P^{\mathcal{G}=0}_{\mathcal{Q}=0}, \label{CQd}
\ee
where $\frac{1}{2}\left(P^{000}_{PR}+P^{11\gamma}_{PR}\right)$ are the two Mermin boxes canonical to the PR-box, $P^{000}_{PR}$,
and $P^{\mathcal{G}=0}_{\mathcal{Q}=0}$ is the local box which has $\mathcal{G}=\mathcal{Q}=0$. The local box in this decomposition is, in general, a nonproduct
box and, therefore, possesses classical correlations. The $3$-decomposition given in Eq. (\ref{CQd}) implies that total nonclassical correlation in a given qubit box
is a sum of Bell discord and Mermin discord.

The observation that $\mathcal{G}$ and $\mathcal{Q}$ measure the distance of a given box from the corresponding
$\mathcal{G}=0$ box and $\mathcal{Q}=0$ box, respectively, in the $3$-decomposition invites us to define the quantity $\mathcal{T}$ that gives the
distance of a given quantum box from the corresponding uncorrelated box that is a product of the marginals of the given box.
\begin{definition}
$\mathcal{T}$ is defined as,
\be
\mathcal{T}=\max_{\alpha\beta}\mathcal{T}_{\alpha\beta}. \label{tot}
\ee
Here,
\be
\mathcal{T}_{\alpha\beta}=|\mathcal{B}_{\alpha\beta}-\mathcal{B}^{prod}_{\alpha\beta}|, \nonumber
\ee
where,
\ba
\mathcal{B}^{prod}_{\alpha\beta}&=&|\braket{A_0}\braket{B_0}+(-1)^\beta\braket{A_0}\braket{B_1}\nonumber
\\&&+(-1)^\alpha\braket{A_1}\braket{B_0}+(-1)^{\alpha\oplus\beta\oplus1}\braket{A_1}\braket{B_1}|.\nonumber
\ea
\end{definition}
This measure has the following properties:
\begin{enumerate}
\item $\mathcal{T}\ge0$.
\item $\mathcal{T}=0$ iff the box is product i.e., $P(a_m,b_n|A_i,B_j)=P_A(a_m|A_i)P_B(b_n|B_j)$.
\begin{proof}
Since $\mathcal{B}_{\alpha\beta}=\mathcal{B}^{prod}_{\alpha\beta}$ for the product box, $\mathcal{T}_{\alpha\beta}=0$ $\forall$ $\alpha\beta$.
For any box that can not written in the product form, $\mathcal{B}_{\alpha\beta}\ne\mathcal{B}^{prod}_{\alpha\beta}$ which, in turn, implies that $\mathcal{T}_{\alpha\beta}>0$
for any nonproduct box.
\end{proof}
\item Maximization in Eq. (\ref{tot}) makes $\mathcal{T}$ invariant under LRO and permutation of the parties.
As the canonical decomposition for quantum correlations in Eq. (\ref{CQd}) implies that $\max\mathcal{B}_{\alpha\beta}$
contains the total amount of nonclassicality in the given JPD, maximization is used in Eq. (\ref{tot}) rather than minimization.
\begin{proof}Under local reversible operations and the permutation of the parties $\mathcal{T}_{\alpha\beta}$ in Eq. (\ref{tot}) transform into each
other.
\end{proof}
\end{enumerate}
As a consequence of the three properties of $\mathcal{T}$ given above, we obtain the following additivity relation for quantum correlations.
\begin{thm}
When a given two-qubit state gives rise to Bell and/or Mermin discord, the correlation satisfy,
\be
\mathcal{T}=\mathcal{G}+\mathcal{Q}\pm\mathcal{C}.
\ee
Here $\mathcal{C}$ quantifies classical correlations.
\end{thm}
\begin{proof}
Consider the correlation given by the canonical decomposition given in Eq. (\ref{CQd}). Since this correlation maximizes $\mathcal{B}_{00}$,
\ba
\mathcal{T}(P)&=&|\mathcal{B}_{00}(P)-\mathcal{B}^{prod}_{00}(P)|\nonumber\\
\!&=&\!\left|4\mathcal{G}'\!+\!2\mathcal{Q}'\!+\!\left(1-\mathcal{G}'-\mathcal{Q}'\right)
\left(\mathcal{B}_{00}\left(P^{\mathcal{Q}=0}_{\mathcal{G}=0}\right)-\mathcal{B}^{prod}_{00}\left(P^{\mathcal{Q}=0}_{\mathcal{G}=0}\right)\right)\right|\nonumber\\
&=&\mathcal{G}+\mathcal{Q}\pm\mathcal{C},
\ea
where
\be
\mathcal{C}=\left(1-\mathcal{G}'-\mathcal{Q}'\right)\left|\mathcal{B}_{00}\left(P^{\mathcal{Q}=0}_{\mathcal{G}=0}\right)-\mathcal{B}^{prod}_{00}\left(P^{\mathcal{Q}=0}_{\mathcal{G}=0}\right)\right|.
\ee
\end{proof}
\section{Quantum correlations}\label{QC2}
Here we study total correlations in the quantum boxes obtained by spin projective measurements on the two-qubit systems: Alice performs measurements $A_i=\hat{a}_i\cdot\vec{\sigma}$
on her qubit along the two directions $\hat{a}_i$ and Bob performs measurements $B_j=\hat{b}_j\cdot\vec{\sigma}$
on her qubit along the two directions $\hat{b}_j$.
Any quantum-correlated state which is neither a classical-quantum state nor a quantum-classical state can give rise to (1) a Bell discordant box which
has $\mathcal{G}>0$ and $\mathcal{Q}=0$, (2) a Mermin discordant box
which has $\mathcal{G}=0$ and $\mathcal{Q}>0$, and (3) a Bell-Mermin discordant box which has $\mathcal{G}>0$ and $\mathcal{Q}>0$, for three different incompatible measurements \cite{Jeba}.
Just like the set of zero quantum discord is non-convex \cite{QCall,Caves}, 
the set of $\mathcal{G}=\mathcal{Q}=0$ correlations forms a nonconvex subset of all local correlations. The set of quantum correlations that
violate a Bell-CHSH inequality is a subset of $\mathcal{G}>0$ correlations.
The set of quantum correlations that violate an EPR-steering inequality \cite{CJWR},
\be
\mathcal{M}_{\alpha\beta\gamma}\le\sqrt{2},
\ee
with $[A_0, A_1]=-1$ or $[B_0, B_1]=-1$, is a subset of $\mathcal{Q}>0$ correlations.

For the incompatible measurements: 
$A_0=\sigma_x$, $A_1=\sigma_y$, $B_0=\sigma_x$ and $B_1=\sigma_y$,  the Bell state, 
\be
\ket{\psi^+}=\frac{1}{\sqrt{2}}(\ket{00}+\ket{11}),
\ee
does not give rise to Bell nonlocality, 
however, it gives rise to Peres' version of KS paradox \cite{Peres}. For this choice of measurements, 
the Bell state gives rise to the following Mermin box,
\begin{equation}
P_M = \left( \begin{array}{cccc}
\half & 0 & 0 & \half \\
\qua & \qua & \qua & \qua \\
\qua & \qua & \qua & \qua \\
0 & \half & \half &  0
\end{array} \right).
\label{eq:merminbox}
\end{equation}
Yet, this correlation is contextual in the sense that it exhibits logical contradiction with noncontextual-realism, i.e., the outcomes 
does not admit a non-contextual-realist value
assignment as follows: The first and fourth rows in Eq. (\ref{eq:merminbox}) imply that the outcomes of $A_0B_0=1$ and $A_1B_1=-1$; if the outcomes are predetermined
noncontextually, it should satisfy, $A_0B_1A_1B_0=-1$, but this contradicts the  rows $2$ and $3$ because there is  a nonzero  probability  
for $A_0B_1=A_1B_0=1$ or  $A_0B_1=A_1B_0=-1$. We shall refer Mermin box as a contextual box when it violates an EPR-steering inequality.  
The measurements that gives rise to maximal violation of a Bell-CHSH inequality
(the Tsirelson bound) does not give rise to the violation of an EPR-steering inequality and vice versa due to the monogamy between nonlocality and contextuality,
\be
\mathcal{G}+2\mathcal{Q}\le4. \label{mGQ}
\ee
For general incompatible measurements, quantum correlations arising from the entangled states violate a Bell-CHSH inequality and an EPR-steering inequality simultaneously,
however, the trade-off exists between the amount of nonlocality and the amount of contextuality as given by the above relation.
This trade-off relation is analogous to the trade-off relationship between KCBS inequality and
Bell-CHSH inequality derived in Ref. \cite{KCK} in the sense that both reveals monogamy between contextuality and nonlocality.

Since the correlations arising from the product states, $\rho_{AB}=\rho_A \otimes \rho_B$, factorize as the product of marginals corresponding to Alice and Bob,
they have $\mathcal{T}=0$. The set of $\mathcal{T}=0$ boxes is a subset of the set of boxes with $\mathcal{G}=\mathcal{Q}=0$, $\left\{P^{\mathcal{G}=0}_{\mathcal{Q}=0}\right\}$.
Any nonproduct state can give rise to nonzero $\mathcal{T}$.
The set of $\mathcal{G}>0$ boxes and $\mathcal{Q}>0$ boxes are the subset of $\mathcal{T}>0$ boxes.
\subsection{Maximally entangled state}
Define the measurement settings:
${\vec{a}_0}=\hat{x}$, ${\vec{a}_1}=\hat{y}$,
${\vec{b}_0} =\sqrt{p}\hat{x}-\sqrt{1-p}\hat{y}$ and ${\vec{b}_1}=\sqrt{1-p}\hat{x}+\sqrt{p}\hat{y}$, where $\frac{1}{2}\le p \le1$.
For this settings, the correlations arising from the Bell state, $\ket{\psi^+}$, 
can be decomposed in a convex mixture of a PR-box, a contextual box, and white noise as,
\be
P=\mathcal{G}'P^{000}_{PR}+\mathcal{Q}'\left(\frac{P^{000}_{PR}+P^{110}_{PR}}{2}\right)+(1-\mathcal{G}'-\mathcal{Q}')P_N, \label{meb1}
\ee
where $\mathcal{G}'=\sqrt{1-p}$
and $\mathcal{Q}'=\sqrt{p}-\sqrt{1-p}$. These correlations violate the Bell-CHSH inequality i.e., $\mathcal{B}_{00}=2\left(\sqrt{p}+\sqrt{1-p}\right)>2$ if
$p\ne1$ and violate the EPR-steering inequality i.e., $\mathcal{M}_{11}=2\sqrt{p}>\sqrt{2}$ if $p\ne\frac{1}{2}$.
Since the correlation maximally violates the Bell-CHSH inequality when $p=\frac{1}{2}$, each pair in the ensemble of two-qubits exhibits
nonlocality for the chosen measurements \cite{EPR2}. When $p$ is increased from $\frac{1}{2}$ to $1$, the number of pairs exhibiting nonlocality decreases and goes to zero
when $p=1$. However, the correlation maximally violates the EPR-steering inequality when $p=1$ which implies that each pair in the ensemble of two-qubits
exhibits local contextuality as the measurements gives rise to the bipartite version of the GHZ paradox \cite{UNLH,GHZ}.
If $p$ is decreased from $1$ to $\frac{1}{2}$, the number of pairs exhibiting local contextuality decreases and the number of pairs
exhibiting nonlocality increases as the violation EPR-steering inequality decreases and the violation of Bell-CHSH inequality increases.
The total amount of correlations in the JPD given in Eq. (\ref{meb1}) is quantified by,
\be
\mathcal{T}=2\left(\sqrt{p}+\sqrt{1-p}\right)=\mathcal{G}+\mathcal{Q}=\left\{\begin{array}{lr}
\mathcal{G} \quad \text{when} \quad p=\frac{1}{2}\\
\mathcal{Q} \quad \text{when} \quad p=1\\
\end{array},
\right.
\ee
which implies that the JPD does not have the component of a classically correlated box.
When the chosen measurements are performed on the ensemble of two-qubits,
each pair in a fraction of the ensemble quantified by $\mathcal{Q}'$ behaves contextually, each pair in a fraction of the ensemble quantified by $\sqrt{2}\mathcal{G}'$
behaves nonlocally and the remaining fraction behaves as noise.
\subsection{Schmidt states}
Consider the correlations arising from the Schmidt states:
\be
\rho_{S}\!=\!\frac{1}{4}\!\left(\!\openone\!\otimes\!\openone\!+\!c(\!\sigma_z\!\otimes\!\openone\!
+\!\openone\!\otimes\!\sigma_z\!)\!+\!s(\!\sigma_x\!\otimes\!\sigma_x\!-\!\sigma_y\!\otimes\!\sigma_y\!)
\!+\!\sigma_z\!\otimes\!\sigma_z\!\right)\!, \label{Schmidt1}
\ee
where $c=\cos2\theta$, $s=\sin2\theta$ and $0\le\theta\le\frac{\pi}{4}$.
The correlation can be decomposed into a convex mixture of a correlation arising from the maximally entangled state and a correlation arising from
a classically correlated state,
\be
P=sP\left(\ket{\psi^+}\right)+(1-s)P\left(\rho_{CC}\right),
\ee
where $P\left(\ket{\psi^+}\right)$ is a correlation arising from the maximally entangled state and
$P\left(\rho_{CC}\right)$ is a correlation arising from the classically correlated state,
\be
\rho=\frac{1}{2}\left(1+\frac{c}{1-s}\right)\ket{00}\bra{00}+\frac{1}{2}\left(1-\frac{c}{1-s}\right)\ket{11}\bra{11},\nonumber
\ee
which is not a physical state.
\subsubsection{Bell-Schmidt box}
\textit{(i) Maximally mixed marginals correlations:-}
The Schmidt states give to the noisy PR-box:
\ba
P=s\left[\frac{1}{\sqrt{2}} P^{000}_{PR}+\left(1-\frac{1}{\sqrt{2}}\right)P_N\right]+(1-s)P_N, \label{BSb1}
\ea
for the measurement settings:
${\vec{a}_0}=\hat{x}$, ${\vec{a}_1}=\hat{y}$,
${\vec{b}_0} =\frac{1}{\sqrt{2}}(\hat{x}-\hat{y})$ and ${\vec{b}_1}=\frac{1}{\sqrt{2}}(\hat{x}+\hat{y})$.
These correlations violate the Bell-CHSH inequality i.e., $\mathcal{B}_{00}=2\sqrt{2}s>2$ if $s>\frac{1}{\sqrt{2}}$.
Since the local box in Eq. (\ref{BSb1}) gives $\mathcal{B}_{00}=0$,
violation of a Bell-CHSH inequality is not achieved by entanglement when $0<p\le\frac{1}{\sqrt{2}}$.
The correlations have,
\be
\mathcal{T}=\mathcal{G}=2\sqrt{2}s,
\ee
which implies that both $\mathcal{T}$ and $\mathcal{G}$ measure the distance of the box from white noise. For this measurement settings, a fraction of the ensemble quantified
by $s$ exhibits nonlocality purely and the remaining fraction behaves as white noise.

\textit{(i) Nonmaximally mixed marginals correlations:-}
For the Popescu-Rohrlich measurement settings \cite{PRQB}: ${\vec{a}_0}=\hat{z}$, ${\vec{a}_1}=\hat{x}$,
${\vec{b}_0}=\cos t\hat{z}+\sin t\hat{x}$ and ${\vec{b}_1}=\cos t\hat{z}-\sin t\hat{x}$,
where $\cos t=\frac{1}{\sqrt{1+s^2}}$, the correlations can be decomposed into PR-box and a local box with nonmaximally mixed marginals and $\mathcal{G}=0$,
\ba
P&=&s^2\left[\frac{1}{\sqrt{1+s^2}}P_{PR}+\left(1-\frac{1}{\sqrt{1+s^2}}\right)P_N\right]\nonumber\\
&&+\left(1-s^2\right)P^{\mathcal{G}=0}_L(\rho).\label{PRQ1}
\ea
Here $P^{\mathcal{G}=0}_L(\rho)$ is a distribution arising from the product state,
\be
\rho=\rho_A \otimes \rho_B, \label{ScPr1}
\ee
where
\be
\rho_A=\rho_B=\frac{1}{2}\left[1+\frac{c}{1-s^2}\right]\ket{0}\bra{0}+\frac{1}{2}\left[1-\frac{c}{1-s^2}\right]\ketbra{1}{1}\nonumber.
\ee
The $\mathcal{G}=0$ box in this decomposition is responsible for the violation of the Bell inequality when $0< s\le\frac{1}{\sqrt{2}}$;
as the box is already lifted to the face of the Bell polytope when $s=0$, any tiny amount of entanglement can give rise to the violation of the Bell-CHSH inequality i.e., $\mathcal{B}_{00}=2\sqrt{1+s^2}>2$ if $s>0$.
The correlations have,
\be
\mathcal{T}=\mathcal{G}=\frac{4s^2}{\sqrt{1+s^2}}.
\ee
That is both $\mathcal{G}$ and $\mathcal{T}$ measure the distance of the box from the local box in the canonical decomposition as
$P^{\mathcal{G}=0}_L$ in Eq. (\ref{PRQ1}) is a product box.
Despite the correlations in Eq. (\ref{BSb1}) do not violate the Bell-CHSH inequality when $0< s\le\frac{1}{\sqrt{2}}$, they have more nonlocality
than the correlations in Eq. (\ref{PRQ1}) as the former correlations have more irreducible PR-box component than the latter correlations.
When the Popescu-Rohrlich measurements
are performed on the Schmidt state, a fraction of the ensemble quantified by $\frac{\sqrt{2}s^2}{\sqrt{1+s^2}}$ exhibits nonlocality purely and the pairs in the
remaining fraction are
uncorrelated.

For the settings ${\vec{a}_0}=\hat{z}$, ${\vec{a}_1}=\hat{x}$,
${\vec{b}_0}=\frac{1}{\sqrt{2}}(\hat{z}+\hat{x})$ and ${\vec{b}_1}=\frac{1}{\sqrt{2}}(\hat{z}-\hat{x})$, the correlations can be decomposed as follows,
\be
P=s\left[\frac{1}{\sqrt{2}}P_{PR}+\left(1-\frac{1}{\sqrt{2}}\right)P_N\right]+(1-s)P^{\mathcal{G}=0}_L(\rho),\label{ZSb1}
\ee
where $P^{\mathcal{G}=0}_L(\rho)$ arises from the correlated state,
\be
\rho=\frac{1}{2}\left(1+\frac{c}{1-s}\right)\ket{00}\bra{00}+\frac{1}{2}\left(1-\frac{c}{1-s}\right)\ket{11}\bra{11}. \nonumber
\ee
The difference between this box and the box in Eq. (\ref{BSb1}) is that the $\mathcal{G}=0$ box in Eq. (\ref{ZSb1}) is not a product box.
The correlations violate the Bell inequality, i.e.,
$\mathcal{B}_{00}=\sqrt{2}(1+s)>2$ if $s>\sqrt{2}-1$; since the $\mathcal{G}=0$ box in Eq. (\ref{ZSb1}) is nonproduct, more entangled states violate the Bell inequality compared
to the correlations in Eq. (\ref{BSb1}).
The correlations have $\mathcal{G}=2\sqrt{2}s$ and
$\mathcal{T}=\sqrt{2}s(1+s)$. Since the JPD has the component of a classically correlated box, it has $\mathcal{T}\ne\mathcal{G}$.
The classical correlations are quantified by,
\be
\mathcal{C}=\mathcal{G}-\mathcal{T}=\sqrt{2}s(1-s)>0 \quad \text{when} \quad s\ne0,1.
\ee
Thus, a fraction of the ensemble given by $s$ exhibits nonlocality purely, and the pairs in the remaining fraction exhibit classical correlations.
\subsubsection{Mermin-Schmidt box}
(i) For the settings
${\vec{a}_0}=\hat{x}$, ${\vec{a}_1}=-\hat{y}$,
${\vec{b}_0}=\hat{y}$ and ${\vec{b}_1}=\hat{x}$, the Schmidt states give rise to the noisy Mermin-box:
\be
P=s \left(\frac{P^{000}_{PR}+P^{111}_{PR}}{2}\right)+(1-s) P_N, \label{MSb1}
\ee
which violates the EPR-steering inequality i.e., $\mathcal{M}_{00}=2s>\sqrt{2}$ if $s>\frac{1}{\sqrt{2}}$.
Grudka \etal \cite{Grudkaetal} have quantified contextuality of isotropic XOR-boxes and it has been 
found that an isotropic XOR box is contextual only when the component of the XOR box 
is larger than a certain amount; similarly, we observe that the isotropic Mermin box in Eq. (\ref{MSb1}) 
can exhibit EPR-steering only when the Mermin box component 
is larger than a certain amount. Thus, analogous to the statement that Bell discord and nonlocality are inequivalent, 
we have the observation that Mermin discord is not equivalent to contextuality of quantum correlations. 
The local correlations in Eq. (\ref{MSb1}) have,
\be
\mathcal{T}=\mathcal{Q}=2s,
\ee
which implies that a fraction of the ensemble quantified by $s$ behaves contextually, and the remaining fraction behaves as white noise.

(ii) For the settings
${\vec{a}_0}=\frac{1}{\sqrt{2}}(\hat{z}+\hat{x})$, ${\vec{a}_1}=\frac{1}{\sqrt{2}}(\hat{z}-\hat{x})$,
${\vec{b}_0}=\cos t\hat{z}-\sin t\hat{x}$, and ${\vec{b}_1}=\cos t\hat{z}+\sin t\hat{x}$,
where $\cos t=\frac{1}{\sqrt{1+s^2} }$, the correlations can be decomposed in a convex mixture of a Mermin box and a local box with $\mathcal{Q}=0$
and nonmaximally mixed marginals,
\ba
P&=&s^2\left[\frac{\sqrt{2}}{\sqrt{1+s^2}}\left(\frac{P^{000}_{PR}+P^{111}_{PR}}{2}\right)+\left(1-\frac{\sqrt{2}}{\sqrt{1+s^2}}\right)P_N\right]\nonumber\\
&&+\left(1-s^2\right)P_{\mathcal{Q}=0}(\rho), \label{CSB1}
\ea
where $P_{\mathcal{Q}=0}(\rho)$ is a local box arising from the product state in Eq. (\ref{ScPr1}).
Since the $\mathcal{Q}=0$ box in this decomposition gives the local bound when $s=0$, the box violates the EPR-steering inequality, i.e.,
$\mathcal{M}_{00}=\sqrt{2}\sqrt{1+s^2}>\sqrt{2}$ if $s>0$. The box has,
\be
\mathcal{T}=\mathcal{Q}=\frac{2\sqrt{2}s^2}{\sqrt{1+s^2}}.
\ee
Since the $\mathcal{Q}=0$ box in Eq. (\ref{CSB1}) is a product box, the amount of total correlations equals to Mermin discord. Notice that for a
given amount of entanglement, the correlations in Eq. (\ref{MSb1})
have more Mermin discord than that for the correlations
in Eq. (\ref{CSB1}) which implies that the latter correlations have less amount of contextuality than the former correlations.

For the settings
${\vec{a}_0}=\frac{1}{\sqrt{2}}(\hat{z}+\hat{x})$, ${\vec{a}_1}=\frac{1}{\sqrt{2}}(\hat{z}-\hat{x})$,
${\vec{b}_0}=\frac{1}{\sqrt{2}}(\hat{z}-\hat{x})$, and ${\vec{b}_1}=\frac{1}{\sqrt{2}}(\hat{z}+\hat{x})$,
the Schmidt states give rise to the following correlation,
\be
P=s\left(\frac{P^{000}_{PR}+P^{111}_{PR}}{2}\right)+(1-s)P^{\mathcal{G}=0}_L(\rho),\label{CSB2}
\ee
where $P^{\mathcal{G}=0}_L(\rho)$ arises from the correlated state,
\be
\rho=\frac{1}{2}\left(1+\frac{c}{1-s}\right)\ket{00}\bra{00}+\frac{1}{2}\left(1-\frac{c}{1-s}\right)\ket{11}\bra{11}.\nonumber
\ee
This box violates the EPR-steering inequality i.e., $\mathcal{M}_{00}=(1+s)>\sqrt{2}$ if $s>\sqrt{2}-1$ which is larger violation than that for the box in Eq. (\ref{MSb1}).
The correlations have $\mathcal{T}=s(1+s)$ and $\mathcal{Q}=2s$ which implies
that the classical correlations in the JPD is quantified as follows,
\be
\mathcal{C}=\mathcal{Q}-\mathcal{T}=s(1-s)>0 \quad \text{when} \quad s\ne0, 1.
\ee
\subsubsection{Bell-Mermin-Schmidt box}
(i) The correlations can be decomposed into a convex mixture of a PR-box, a Mermin-box, and white noise:
\ba
P&=&\left(1-q-g\right) P_N+ \frac{q}{2}\left(P^{000}_{PR}+P^{11\gamma}_{PR}\right)\nonumber\\
&+&g\left[\frac{1}{\sqrt{2}}P^{000}_{PR}+\left(1-\frac{1}{\sqrt{2}}\right)P_N\right], \label{BMSb1}
\ea
for the settings:
${\vec{a}_0}=s\hat{x}+c\hat{y}$,
${\vec{a}_1}=c\hat{x}-s\hat{y}$,
${\vec{b}_0}=\frac{1}{\sqrt{2}}(\hat{x}+\hat{y})$ and
${\vec{b}_1}=\frac{1}{\sqrt{2}}(\hat{x}-\hat{y})$, where
$q=\frac{s\left||c+s|-|c-s|\right|}{\sqrt{2}}$ and $g=|s(s-c)|$.
This box gives,
\be
\mathcal{G}=2\sqrt{2}s|s-c|>0 \quad \text{except when} \quad \theta \ne0, \frac{\pi}{8}, \nonumber
\ee
\ba
\mathcal{Q}&=&s\sqrt{2}\Big||c+s|-|c-s|\Big|>0 \quad \text{except when} \quad \theta \ne0, \frac{\pi}{4}\nonumber\\
&=&\left\{\begin{array}{lr}
2\sqrt{2}s^2 \quad \text{when} \quad c>s\\
2\sqrt{2}cs \quad \text{when} \quad s>c\\
\end{array}
\right.\nonumber
\ea
and
\ba
\mathcal{T}&=&\left\{
\begin{array}{lr}
2\sqrt{2}s^2 \quad \text{when} \quad s>c\\
2\sqrt{2}cs \quad \text{when} \quad c>s\\
\end{array}
\right.\nonumber\\
&=&\mathcal{G}+\mathcal{Q},
\ea
which implies that the box has nonclassical correlations purely as the box does not have classical correlation component;
a fraction of the ensemble quantified by $g$ exhibits nonlocality wholly, a fraction of the ensemble quantified by $q$ exhibits contextuality and
the remaining fraction behaves as white noise.

(ii) For the settings: ${\vec{a}_0}=c\hat{x}+s\hat{z}$,
${\vec{a}_1}=s\hat{x}-c\hat{z}$,
${\vec{b}_0}=\frac{1}{\sqrt{2}}(\hat{x}+\hat{z})$ and
${\vec{b}_1}=\frac{1}{\sqrt{2}}(-\hat{x}+\hat{z})$, the correlations have the same amount of $\mathcal{G}$ and $\mathcal{Q}$ as for the correlations in Eq. (\ref{BMSb1}),
however, they have a different amount of $\mathcal{T}$ which is given as follows,
\ba
\mathcal{T}&=&\left\{
\begin{array}{lr}
\sqrt{2}s^2(1+s) \quad \text{when} \quad s>c\\
\sqrt{2}cs(1+s) \quad \text{when} \quad c>s.\\
\end{array}
\right.\nonumber
\ea
Thus, the correlations arising from the latter settings (ii) have the decomposition that has the same amount of PR-box and Mermin box components as for the former settings (i) except that white noise in Eq. (\ref{BMSb1})
is replaced by the classically correlated box.
The classical correlations are quantified by,
\ba
\mathcal{C}&=&\mathcal{G}+\mathcal{Q}-\mathcal{T}\nonumber\\
&=&\left\{
\begin{array}{lr}
\sqrt{2}s^2(1-s) \quad \text{when} \quad s>c\\
\sqrt{2}cs(1-s) \quad \text{when} \quad c>s.\\
\end{array}
\right.\nonumber
\ea
\subsection{Werner states} Consider the correlations arising from the Werner states \cite{Werner},
\be
\rho_W=p\ketbra{\psi^+}{\psi^+}+(1-p)\frac{\openone}{4}.
\ee
The Werner states are entangled if $p>\frac{1}{3}$ and have nonzero quantum discord if $p>0$ \cite{OZ}.
Since the Werner states have the component of an irreducible entangled state if $p>0$, they can give rise to nonclassical correlations if $p>0$.
As the Werner states can only give rise to maximally mixed marginals correlations,
the nonclassical correlations arising from the Werner states cannot have the component of classical correlation.
\subsubsection{Bell-Werner box}
The correlations have the following decomposition,
\be
P=p\left[\frac{1}{\sqrt{2}} P^{000}_{PR}+\left(1-\frac{1}{\sqrt{2}}\right)P_N\right]+(1-p)P_N.
\ee
for the settings that correspond to the correlation in Eq. (\ref{BSb1}). These correlations have,
\be
\mathcal{T}=\mathcal{G}=2\sqrt{2}p.
\ee
\subsubsection{Mermin-Werner box}
The Werner states give rise to the noisy Mermin box,
\be
P=(1-p) P_N+ p\left(\frac{P^{000}_{PR}+P^{111}_{PR}}{2}\right),
\ee
for the settings corresponding to the correlation in Eq. (\ref{MSb1}).
These correlations have,
\be
\mathcal{T}=\mathcal{Q}=2p.
\ee
\subsubsection{Bell-Mermin-Werner box}
Th correlations admit the following decomposition:
\be
P=(1-q-r) P_N+ \frac{q}{2}\left(P^{000}_{PR}+P^{11\gamma}_{PR}\right)+|r|P^{000}_{PR}, \label{BMWb1}
\ee
for the settings:
${\vec{a}_0}=\sqrt{p}\hat{x}+\sqrt{1-p}\hat{y}$,
${\vec{a}_1}=\sqrt{1-p}\hat{x}-\sqrt{p}\hat{y}$,
${\vec{b}_0}=\frac{1}{\sqrt{2}}(\hat{x}+\hat{y})$ and
${\vec{b}_1}=\frac{1}{\sqrt{2}}(\hat{x}-\hat{y})$, where
$q=p\sqrt{2(1-p)}$ and $r=\frac{1}{\sqrt{2}}p\left(\sqrt{p}-\sqrt{1-p}\right)$. The box gives
\ba
\mathcal{G}&=&2\sqrt{2}p|\sqrt{p}-\sqrt{1-p}|>0 \quad \text{except when} \quad p\ne0, \frac{1}{2} \nonumber,\\
\mathcal{Q}&=&\sqrt{2}p\Big|\sqrt{p}+\sqrt{1-p}-\left|\sqrt{p}-\sqrt{1-p}\right|\Big|\nonumber\\
&&>0 \quad \text{except when} \quad p\ne0, 1\nonumber\\
&=&\left\{\begin{array}{lr}
2p\sqrt{2p}\quad \text{when} \quad 0\le p\le\frac{1}{2}\nonumber\\
2p\sqrt{2(1-p)} \quad \text{when} \quad\frac{1}{2}\le p\le1\nonumber
\end{array}
\right.
\ea
and
\ba
\mathcal{T}&=&\left\{\begin{array}{lr}
2p\sqrt{2(1-p)}\quad \text{when} \quad 0\le p\le\frac{1}{2}\nonumber\\
2p\sqrt{2p}\quad \text{when} \quad\frac{1}{2}\le p\le1\nonumber
\end{array}
\right.\\
&=&\mathcal{G}+\mathcal{Q}.
\ea
\subsection{Mixture of maximally entangled state with colored noise}
Consider the correlations arising from the mixture of the Bell state and the classically correlated state,
\be
\rho= p \ketbra{\psi^+}{\psi^+}+(1-p) \rho_{CC}, \label{rCC1}
\ee
where $\rho_{CC}=\frac{1}{2}(\ketbra{00}{00}+\ketbra{11}{11})$.
Only when suitable incompatible measurements that lie in the \plane{x}{z} are performed on these states, correlations arising from these states have different nonclassical behavior than the Werner states.
\subsubsection{Bell discordant box}
For the settings that give rise to the noisy PR-box in Eq. (\ref{BSb1}),
\be
\mathcal{T}=\mathcal{G}=2\sqrt{2}p.
\ee

The measurement settings,
${\vec{a}_0}=\hat{z}$, ${\vec{a}_1}=\hat{x}$,
${\vec{b}_0}=\cos t\hat{z}+\sin t\hat{x}$ and ${\vec{b}_1}=\cos t\hat{z}-\sin t\hat{x}$, where $\cos t=\frac{1}{\sqrt{1+p^2}}$, gives rise to the
violation of the Bell inequality, $\mathcal{B}_{00}=2\sqrt{1+p^2}>2$, if $p>0$. Since the box lies at the face of the Bell polytope when $p=0$, any tiny amount of
entanglement gives rise to the violation Bell-CHSH inequality.
The correlations have $\mathcal{G}=\frac{4p^2}{\sqrt{1+p^2}}$ and $\mathcal{T}=2\sqrt{1+p^2}$ which implies that the classical correlations is quantified as follows,
\be
\mathcal{C}=\mathcal{T}-\mathcal{G}=\frac{2(1-p^2)}{\sqrt{1+p^2}}.
\ee
\subsubsection{Mermin discordant box}
The measurement settings,
${\vec{a}_0}=\frac{1}{\sqrt{2}}(\hat{z}+\hat{x})$, ${\vec{a}_1}=\frac{1}{\sqrt{2}}(\hat{z}-\hat{x})$,
${\vec{b}_0}=\cos t\hat{z}+\sin t\hat{x}$ and ${\vec{b}_1}=\cos t\hat{z}-\sin t\hat{x}$, where $\cos t=\frac{1}{\sqrt{1+p^2}}$, gives rise to the
violation of the EPR-steering inequality, $\mathcal{M}_{00}=\sqrt{2}\sqrt{1+p^2}>\sqrt{2}$, if $p>0$. The correlations have
$\mathcal{Q}=\frac{2\sqrt{2}p^2}{\sqrt{1+p^2}}$ and $\mathcal{T}=\sqrt{2}\sqrt{1+p^2}$ which implies that the amount of classical correlations in the JPD is
quantified as follows,
\be
\mathcal{C}=\mathcal{T}-\mathcal{Q}=\frac{\sqrt{2}(1-p^2)}{\sqrt{1+p^2}}.
\ee
\subsubsection{Bell-Mermin discordant box}
For the measurement settings
${\vec{a}_0}=\sqrt{p}\hat{z}+\sqrt{1-p}\hat{x}$,
${\vec{a}_1}=\sqrt{1-p}\hat{z}-\sqrt{p}\hat{x}$,
${\vec{b}_0}=\frac{1}{\sqrt{2}}(\hat{z}+\hat{x})$ and
${\vec{b}_1}=\frac{1}{\sqrt{2}}(\hat{z}-\hat{x})$,
the correlations have the same amount of Bell discord and Mermin discord as for the correlations in Eq. (\ref{BMWb1}), however, the box has different amount of
total correlations,
\ba
\mathcal{T}&=&\left\{\begin{array}{lr}
(1+p)\sqrt{2(1-p)}\quad \text{when} \quad 0\le p\le\frac{1}{2}\nonumber\\
(1+p)\sqrt{2p}\quad \text{when} \quad\frac{1}{2}\le p\le1\nonumber
\end{array}
\right.\\
&>&\mathcal{G}+\mathcal{Q},
\ea
because of the classically correlated noise. The amount of classical correlations is given by,
\ba
\mathcal{C}&=&\mathcal{T}-\mathcal{G}-\mathcal{Q}\nonumber\\
&=&\left\{
\begin{array}{lr}
(1-p)\sqrt{2(1-p)}\quad \text{when} \quad 0\le p\le\frac{1}{2}\nonumber\\
(1-p)\sqrt{2p}\quad \text{when} \quad\frac{1}{2}\le p\le1.\nonumber
\end{array}
\right.\nonumber
\ea

\section{Conclusion}
We have interpreted Bell discord and Mermin discord as distance measures for nonlocality and contextuality which led us to construct the distance measure, $\mathcal{T}$, which is zero iff the box is a product. We have discussed the problem of separating the total correlations in the quantum boxes into nonlocality, contextuality and classical correlations using these
three measures. We have studied the additivity relation for quantum correlations in two-qubit systems.
The distance measure interpretation has allowed us to understand why some entangled states cannot lead to the violation of a Bell-CHSH inequality.

\chapter{Isolating genuine nonclassicality in tripartite quantum correlations}
\label{Ch5}

\section*{Abstract}
We introduce the measures, Svetlichny and Mermin discord, to characterize the presence of genuine nonclassicality 
in tripartite quantum correlations. We show that
any correlation in the Svetlichny-box polytope which is a subpolytope of full nonsignaling polytope admits a three-way decomposition 
using these measures of nonclassicality. This decomposition allows us to isolate the origin of nonclassicality into 
three disjoint sources: a Svetlichny box, a maximally two-way nonlocal box, and a classical correlation. 
Svetlichny and Mermin discord quantify three-way nonlocality and three-way contextuality of quantum correlations with respect to 
the three-way decomposition in that they reveal the presence of incompatible measurements. 
A third measure is introduced to separate the total correlations in a quantum joint probability distribution into a purely nonclassical and a classical part.

\section{Introduction}
Correlations between outcomes of local measurements on entangled states are in general incompatible 
with local hidden variable (LHV) theories \cite{bell64}. In the multipartite scenario, distinct types of LHV theories exist \cite{BNL}. 
In the tripartite case, Svetlichny showed that quantum correlations can have genuine nonlocality which cannot be explained by 
hybrid local-nonlocal hidden variable (HLHV) theory \cite{SI}. Just like bipartite quantum correlations cannot violate a Bell-CHSH inequality
more than the Tsirelson bound \cite{BNL}, multipartite quantum correlations cannot violate a Svetlichny inequality more than a certain bound \cite{multi}. 
Quantum theory is only a subclass of a multipartite generalized nonsignaling theory 
that predicts extremal genuine nonlocality \cite{MAG06}. Generalized nonsignaling theories have been under investigation to find out what physical principles 
exactly captures quantum correlations in addition to nonsignaling principle and nonlocality  \cite{PR,BNL}. 
In Ref. \cite{LOmulti}, it was shown that a complete characterization of quantum correlations requires genuine multipartite principles.  
Genuine multipartite nonlocality is a resource for multipartite quantum information tasks \cite{DQKD}. 
Thus, characterizing and quantifying multipartite correlations using genuine multipartite concepts is of interest to both foundations and quantum information.

Georgi \etal \cite{GTC} introduced a notion of genuine discord to quantify tripartite nonclassicality in quantum states 
that cannot be reduced to the correlations in subsystems.
In this work, we introduce two notions of genuine discord for tripartite NS boxes.
We characterize genuine nonclassicality of tripartite quantum correlations by using two binary inputs and two binary outputs nonsignaling (NS) polytope \cite{Pironioetal}. 
We define Svetlichny and Mermin discord using Svetlichny and Mermin operators which put an upper bound on the correlations under the constraints 
of the HLHV model \cite{SI} and fully LHV model \cite{mermin}. 
Analogous to genuine quantum discord \cite{GTC}, these measures detect the presence of genuine nonclassicality in Svetlichny-local correlations as well. 
We obtain a $3$-decomposition that any correlation in the Svetlichny-box polytope which is a subpolytope of full NS polytope 
can be written as a convex combination of a Svetlichny-box, a maximally three-way contextual box, 
and a box which does not have Svetlichny and Mermin discord.
Svetlichny and Mermin discord quantify the components of Svetlichny-box and three-way contextual box respectively in the $3$-decomposition.
Thus, Svetlichny and Mermin discord quantify genuine nonclassicality of Svetlichny-local quantum correlations originating from 
Svetlichny nonlocality and three-way contextuality respectively.
We identify the set of genuinely nonclassical biseparable and separable three-qubit states using Svetlichny and Mermin discord.

This chapter is organized as follows. In Sec. \ref{prl}, we review the tripartite nonsignaling polytope with two-inputs and two-outputs. 
In Sec. \ref{Spoly}, we define Svetlichny-box polytope and the two measures, Svetlichny and Mermin
discord. In this section, we find the canonical decomposition of any correlation in the Svetlichny-box polytope.
In Sec. \ref{QC}, we characterize quantum correlations arising from $2\times 2 \times 2$ states.
We present conclusions in Sec. \ref{conc}.
\section{Preliminaries}\label{prl}
Consider the Bell scenario in which three spatially separated parties, Alice, Bob and Charlie, share a tripartite box which has two binary inputs and two binary outputs per party. 
The correlation between the outputs is captured by the set of joint probability distributions (JPDs), $P(a_m,b_n,c_o|A_i,B_j,C_k)$, here $m, n, o, i, j, k \in \{0,1\}$. 
In addition to positivity and normalization, the JPDs characterizing a given box satisfy nonsignaling constraints:
\be
\sum_m P(a_m,b_n,c_o|A_i,B_j,C_k)=P(b_n,c_o|B_j,C_k) \quad \forall n,o,i,j,k, 
\ee
and the permutations. The set of such NS boxes forms a convex polytope, 
$\mathcal{N}$, in a $26$ dimensional space \cite{Barrett}. Any box that 
belongs to this polytope can be uniquely described by $6$ single-party, $12$ two-party and $8$ three-party expectations as follows,
\ba
&&P(a_m,b_n,c_o|A_i,B_j,C_k)\nonumber \\
&&=\frac{1}{8}[1+(-1)^m\braket{A_i}+(-1)^n\braket{B_j}+(-1)^o\braket{C_k}+(-1)^{m\oplus n}\braket{A_iB_j}\nonumber\\
&&+(-1)^{m\oplus o}\braket{A_iC_k}+(-1)^{n\oplus o}\braket{B_jC_k}+(-1)^{m\oplus n\oplus o}\braket{A_iB_jC_k}].
\ea
Pironio \etal \cite{Pironioetal} found that $\mathcal{N}$ has $53856$ extremal boxes (vertices) which belong to $46$ classes. The vertices in each class are equivalent 
in that they can be converted into each other through local reversible operations (LRO), which include local relabeling of the inputs
and outputs \cite{Barrett}. These $46$ classes of vertices 
can be classified into local, two-way nonlocal and $44$ classes of three-way nonlocal vertices. 

Two-way local polytope, $\mathcal{L}_2$, is a convex subpolytope of $\mathcal{N}$ whose vertices are the $64$ local vertices and the $48$ two-way nonlocal vertices.
The local vertices are fully deterministic boxes given as follows,
\be
P^{\alpha\beta\gamma\epsilon\zeta\eta}_D(a_m,b_n,c_o|A_i,B_j,C_k)=\left\{
\begin{array}{lr}
1, & m=\alpha i\oplus \beta\\
   & n=\gamma j\oplus \epsilon \\
   & o=\zeta k \oplus \eta\\
0 , & \text{otherwise}.\\
\end{array}
\right.  \label{DB} \ee 
Here $\alpha,\beta,\gamma,\epsilon, \zeta, \eta\in \{0,1\}$  and $\oplus$ denotes addition modulo $2$. 
The two-way nonlocal vertices are the bipartite PR-boxes: there are $16$ vertices in which PR-box is shared between $A$ and $B$,
\begin{align}
P^{\alpha\beta\gamma\epsilon}_{12}(a_m,b_n,c_o|A_i,B_j,C_k)
=\left\{
\begin{array}{lr}
\frac{1}{2}, & m\oplus n=i\cdot j \oplus \alpha i\oplus \beta j \oplus \gamma \quad \& \quad o=\epsilon k\\ 
0 , & \text{otherwise},\\
\end{array}
\right. \label{PR}
\end{align}
and the other $32$ two-way nonlocal vertices, $P^{\alpha\beta\gamma\epsilon}_{13}$ and $P^{\alpha\beta\gamma\epsilon}_{23}$, 
in which PR-box is shared by $AC$ and $BC$ are similarly defined.
$\mathcal{L}_2$ can be divided into a two-way nonlocal region and Bell-local polytope, $\mathcal{L}$, whose vertices are the deterministic
boxes given in Eq. (\ref{DB}).
All correlations in $\mathcal{L}$ can be explained by the LHV theory, i.e., the correlations can be decomposed as follows,
\ba
P(a_m,b_n,c_o|A_i,B_j,C_k)=\sum_\lambda p_\lambda P_\lambda(a_m|A_i)P_\lambda(b_n|B_j)P_\lambda(c_k|C_k), \label{LHV}
\ea
whereas all correlations in the two-way nonlocal region can be decomposed into the
hybrid local-nonlocal form in which arbitrary nonlocality consistent with nonsignaling principle is allowed between
two parties in the different bipartitions,
\begin{align}
P(a_m,b_n,c_o|A_i,B_j,C_k)=p_1\sum_\lambda p_\lambda P_\lambda^{AB|C}+p_2\sum_\lambda q_\lambda P_\lambda^{AC|B}+p_3\sum_\lambda r_\lambda P_\lambda^{A|BC}, \label{HLNL}
\end{align}
where $P_\lambda^{AB|C}=P_\lambda(a_m,b_n|A_i,B_j)P_\lambda(c_o|C_k)$, and, where $P_\lambda^{AC|B}$ and
$P_\lambda^{A|BC}$  are  similarly  defined. 
  
Bell-nonlocal correlations that do not admit the decomposition in Eq. (\ref{HLNL}) exhibit genuine three-way nonlocality. 
Three-way nonlocal correlations violate a facet inequality corresponding to $\mathcal{L}_2$. Bancal \etal \cite{Banceletal} 
found that $\mathcal{L}_2$ has $185$ classes of facet inequalities. 
In this work, we consider two classes of $3$-way nonlocal vertices that belong to the classes $8$ and $46$ given in Pironio \etal \cite{Pironioetal}. The extremal boxes 
that belong to the class $8$ violate a class $99$ facet inequality to its algebraic maximum. A representative of the class $99$ facet inequality
is given by,
\be
\mathcal{L}^{99}_2=\braket{A_0B_0}+\braket{A_0C_0}+\braket{B_1C_0}+\braket{A_1B_0C_1}-\braket{A_1B_1C_1}\le3. \label{NSFI}
\ee
The representative of class $8$ extremal box given in the table of Ref. \cite{Pironioetal} has 
$\braket{A_0B_0}=\braket{A_0B_1}=\braket{A_0C_0}=\braket{B_0C_0}=\braket{B_1C_0}=\braket{A_1B_0C_1}=-\braket{A_1B_1C_1}=1$ and the rest of the expectations are zero 
which imply $\mathcal{L}^{99}_2=5$. 
The extremal boxes that belong to 
the class $46$ are $16$ Svetlichny-boxes,
\begin{align}
P^{\alpha\beta\gamma\epsilon}_{\rm Sv}(a_m,b_n,c_o|A_i,B_j,C_k)\!=\!\left\{
\begin{array}{lr}
\frac{1}{4}, & \!m\!\oplus \!n\!\oplus \!o\!=\!i\cdot j \!\oplus \!i\cdot k\! \oplus \!j\cdot k \!\oplus \!\alpha i\!\oplus\! \beta j\! \oplus\! \gamma k \!\oplus\! \epsilon\\ 
0 , & \text{otherwise},\\
\end{array}
\right. \label{NLV} 
\end{align}
which violate one of the class $185$ facet inequalities,
\be
\mathcal{S}_{\alpha\beta\gamma\epsilon}=\sum_{ijk}(-1)^{i\cdot j \oplus i\cdot k \oplus j\cdot k \oplus \alpha i\oplus \beta j \oplus \gamma k \oplus \epsilon}\braket{A_iB_jC_k}\le4, \label{SI}
\ee
to its algebraic maximum of $8$. A class $185$ facet inequality is known as Svetlichny inequality \cite{SI}. 
We will refer to the correlations which do not violate a Svetlichny inequality as Svetlichny-local.

In this work, we consider quantum correlations arising from Svetlichny scenario \cite{SI} in which the parties generate the JPDs by making spin projective measurements
$A_i=\hat{a}_i\cdot \vec{\sigma}$, $B_j=\hat{b}_j\cdot \vec{\sigma}$ and 
$C_k=\hat{c}_k\cdot \vec{\sigma}$ on an ensemble of three-qubit system described by the 
density matrix $\rho$ in the Hilbert space $\mathcal{H}^A_2\otimes\mathcal{H}^B_2\otimes\mathcal{H}^C_2$.
The correlation predicted by quantum theory is defined as follows,
\be
P(a_m,b_n,c_o|A_i,B_j,C_k)=\mathrm{Tr} \left(\rho \Pi^{a_m}_{A_i}\otimes \Pi^{b_n}_{B_j}\otimes \Pi^{c_o}_{C_k}\right), \label{QCD}
\ee
where 
\begin{align}
\Pi^{a_{m}}_{A_i}={1/2}\left[\openone +a_{m}\hat{a}_i \cdot \vec{\sigma}\right], \Pi^{b_{n}}_{B_j}={1/2}\left[\openone+b_{n}\hat{b}_j \cdot \vec{\sigma}\right]
 \&  \Pi^{c_{0}}_{C_k}={1/2}\left[\openone+c_{o}\hat{c}_k \cdot \vec{\sigma}\right] \nonumber
\end{align}
are the projectors generating binary outcomes $a_{m},b_{n}, c_o \in \{-1,1\}$.  
Any such tripartite quantum correlation can be written as a convex mixture of the extremal boxes of the tripartite NS polytope.

\section{Svetlichny-box polytope and two notions of genuine nonclassicality for Svetlichny-local boxes} \label{Spoly}
Svetlichny-box polytope, $\mathcal{R}$, is a restricted NS polytope in which we discard in total $53856-128=53728$ extremal boxes.
The $128$ extremal boxes of $\mathcal{R}$ are the Svetlichny-boxes, the bipartite PR-boxes and the deterministic boxes. Svetlichny-box polytope is convex, 
i.e., if $P\in\mathcal{R}$,
\begin{equation}
P=\sum^{15}_{i=0}p_iP^i_{Sv}+\sum^{15}_{i=0}q_iP^i_{12}+\sum^{15}_{i=0}r_iP^i_{13}+\sum^{15}_{i=0}s_iP^i_{23}+\sum^{63}_{j=0}t_jP^j_{D},
\label{eq:gendecomp}
\end{equation}
with $\sum_ip_i+\sum_iq_i+\sum_ir_i+\sum_is_i+\sum_jt_j=1$, $i=\alpha\beta\gamma\epsilon$ and $j=\alpha\beta\gamma\epsilon\zeta\eta$.
Svetlichny-box polytope can be divided into a three-way nonlocal region and the two-way local polytope ($\mathcal{L}_2$).

\begin{figure}[h!]
\centering
\includegraphics[scale=0.38]{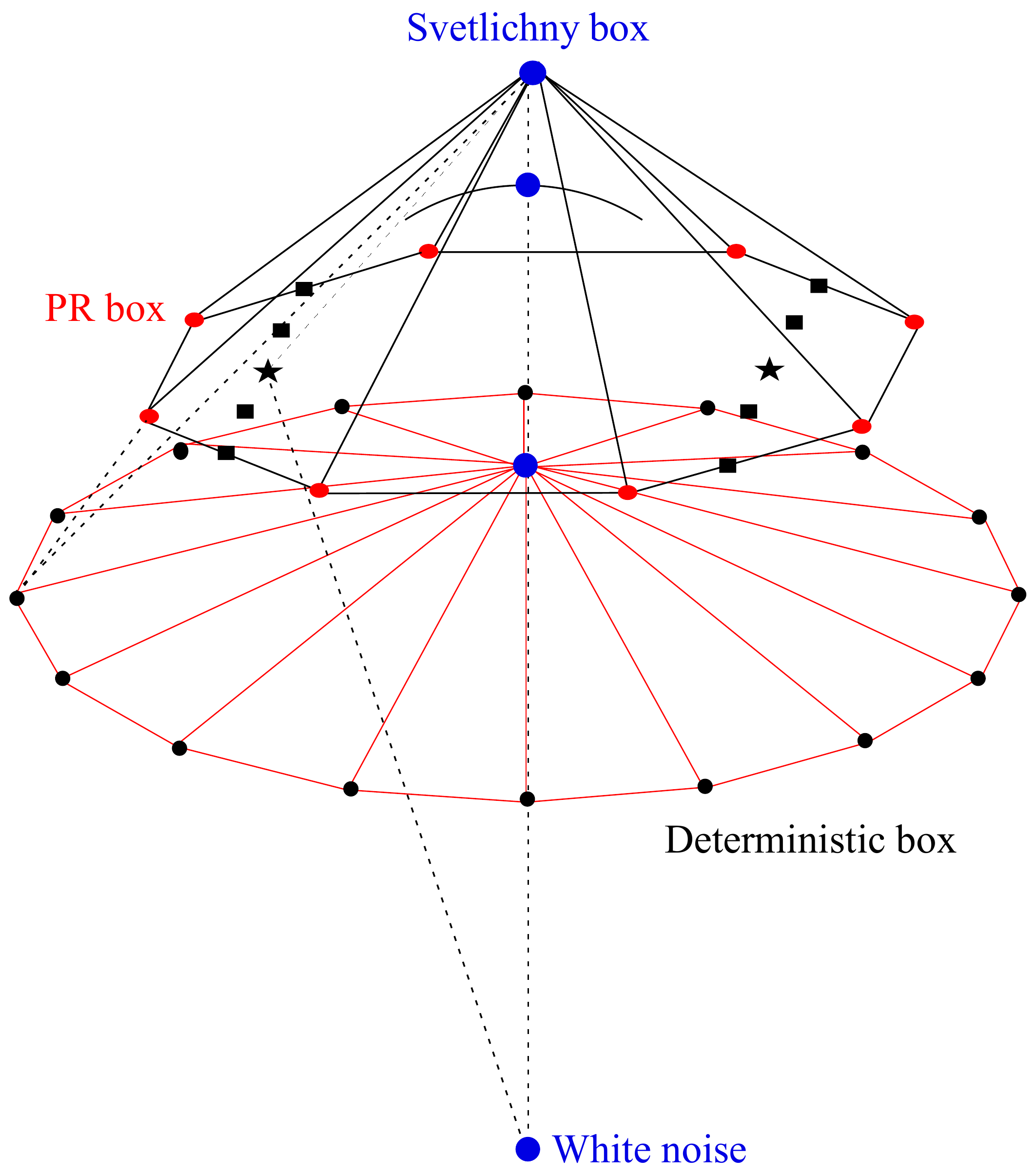} 
\caption[Svetlichny-box polytope]{A three-dimensional representation of the Svetlichny-box polytope is shown here.  
The fully deterministic boxes are represented by the circular points on the hexadecagon. The bipartite PR-boxes are represented by the
circular points on the octagon. The circular point on the top represents the Svetlichny-box.
The region that lies above the hexadecagon and below the octagon represents the two-way nonlocal region.
The region below the curved surface contains quantum correlations and the point on this curved surface represents the quantum box that achieves 
maximal Svetlichny nonlocality.
The star and square points represent quantum and nonquantum Mermin boxes respectively. 
The triangular region (shown by dotted lines) which is a convex hull
of the Svetlichny-box, the Mermin box and white noise represents the $3$-decomposition fact that  
any point that lies inside the triangle can be decomposed into Svetlichny-box,
the Mermin-box and white noise. The circular point at the center of the hexadecagon is the isotropic Svetlichny-box with $p_{Sv}=\frac{1}{2}$ 
which can be decomposed as an equal mixture of
the $16$ deterministic boxes or an equal mixture of two quantum Mermin boxes.}\label{NS3dfig1}
\end{figure} 

$\mathcal{L}_2$ is a convex hull of the $48$ two-way nonlocal vertices and the $64$ deterministic boxes, i.e., if $P \in \mathcal{L}_2$,
\begin{eqnarray}
P&=&\sum^{15}_{i=0}q_iP^i_{12}+\sum^{15}_{i=0}r_iP^i_{13}+\sum^{15}_{i=0}s_iP^i_{23}+\sum^{63}_{j=0}t_jP^j_{D};\nonumber\\
&&\sum_iq_i+\sum_ir_i+\sum_is_i+\sum_jt_j=1.
\label{eq:gendecompSl}
\end{eqnarray}
The set of correlations in $\mathcal{L}_2$ is only a subset of the Svetlichny-local correlations as there are three-way nonlocal correlations that satisfy the Svetlichny inequalities \cite{Banceletal}.
\begin{proposition}\label{SIv}
The complete set of Svetlichny inequalities is a necessary and sufficient condition for the correlations in $\mathcal{R}$ to belong to the two-way local polytope.
\end{proposition}
\begin{proof}Svetlichny inequality can be interpreted as bipartite Bell-CHSH inequality between any two combined system and the third system which can be readily seen by rewriting Svetlichny operator as bipartite Bell-CHSH operator, for instance,
\be
\mathcal{S}_{0000}\!=\!\braket{(A_0B_1+A_1B_0)(C_0+C_1)-(A_0B_0-A_1B_1)(C_0-C_1)}.\nonumber
\ee
Here we have considered the combined system AB as a single subsystem. In the bipartite scenario, the complete set of Bell-CHSH inequalities,
\ba
\mathcal{B}_{\alpha\beta\gamma}&:= &(-1)^\gamma\braket{A_0B_0}+(-1)^{\beta \oplus \gamma}\braket{A_0B_1}\nonumber\\
&+&(-1)^{\alpha \oplus \gamma}\braket{A_1B_0}+(-1)^{\alpha \oplus \beta \oplus \gamma \oplus 1} \braket{A_1B_1}\le2, \label{BCHSH}
\ea
serve as the necessary and sufficient condition for the correlations to belong to the Bell polytope and is invariant under LRO \cite{Fine,WernerWolf}.
Just as the complete set of Bell-CHSH inequalities,
the set of Svetlichny inequalities in Eq. (\ref{SI}) is invariant under LRO and the permutations of the parties and, therefore, they form a complete set of inequalities \cite{WernerWolfmulti}.
As any genuinely nonlocal correlation in $\mathcal{R}$ can be written as a convex combination of an irreducible Svetlichny-box and a Svetlichny-local box (see fig. \ref{NS3dfig1}),
it violates a Svetlichny inequality.
If genuine nonlocality of a correlation is due to some other extremal
three-way nonlocal box, it may not violate a Svetlichny inequality; for instance, the class $8$ three-way nonlocal box which violates a class $99$
facet inequality does not violate a Svetlichny inequality.
\end{proof}

The Bell-local polytope ($\mathcal{L}$), which is a subpolytope of the two-way local polytope,
is a convex hull of the $64$ deterministic boxes, i.e., if $P \in \mathcal{L}$,
\be
P=\sum^{63}_{j=0}t_jP^j_{D}; \quad \sum_j t_j=1. \label{LD}
\ee
\begin{proposition}\label{MIBIv}
The necessary and sufficient condition for a correlation to admit the local deterministic hidden variable model inEq. (\ref{LD}) is that the correlation and its three bipartite marginals satisfy all the Mermin inequalities and all the Bell-CHSH inequalities.
\end{proposition}
\begin{proof}
The decomposition in Eq. (\ref{LD}) implies that all three bipartite marginal distributions can be written as a convex combination of the $16$ deterministic boxes that are the vertices of the bipartiteBell polytope, however, the converse is not true as there 
are nonlocal correlations whose bipartite marginal correlations admit a
local deterministic model.
Therefore, the three complete set of Bell-CHSH inequalities corresponding to the three bipartite marginals is only a sufficient condition for the correlations to belong to the tripartite Bell-local polytope. Notice that the nonlocal correlations that satisfy all the Bell-CHSH inequalities violate a Mermin inequality in Eq. (\ref{MI}), for instance, a tripartite Mermin box in Eq. (\ref{Mbdec})
whose marginal correlations
are white noise violate a Mermin inequality.
The set of Mermin inequalities in Eq. (\ref{MI}) is invariant under LRO and thus it forms a complete set of inequalities \cite{WernerWolfmulti}.
Consider the following Mermin inequality,
\be
\braket{(A_0B_0-A_1B_1)C_0-(A_0B_1-A_1B_0)C_1}\le2.
\ee
This inequality becomes bipartite Bell-CHSH inequality between $A$ and $B$ iff $C$ is deterministic i.e., $\braket{C_i}=\pm1$.
Therefore, there are nonlocal correlations that do not violate a Mermin inequality; however, they violate a Bell-CHSH inequality
since nonlocality is due to one of the bipartite marginals.
\end{proof}

\subsection{Svetlichny discord}
Consider isotropic Svetlichny-box which is a convex mixture of the Svetlichny-box and white noise,
\be
P=p_{Sv}P^{0000}_{Sv}+(1-p_{Sv})P_N.\label{nSv}
\ee
The isotropic Svetlichny-box violates the Svetlichny inequality i.e., $\mathcal{S}_{0000}=8p_{Sv}>4$ if $p_{Sv}>\frac{1}{2}$.
Notice that even if the isotropic Svetlichny-box is local when $p_{Sv}\le\frac{1}{2}$, it admits a decomposition that has
the single Svetlichny-box component. We call such a single Svetlichny-box in the decomposition of any correlation (three-way nonlocal, or not)
irreducible Svetlichny-box.

The isotropic Svetlichny-box which is quantum realizable if $p_{Sv}\le\frac{1}{\sqrt{2}}$ illustrates the following observation.
\begin{observation}\label{m1}
When a Svetlichny-local correlation arising from a given genuinely entangled state has an irreducible Svetlichny-box component,
the correlation arises from incompatible measurements which are noncommuting on each side: $[A_0,A_1]\ne0$, $[B_0,B_1]\ne0$ and $[C_0,C_1]\ne0$.
\end{observation}
\begin{proof}
For the incompatible measurements $A_0=\sigma_x$, $A_1=\sigma_y$,    
$B_0=\sigma_x$, $B_1=\sigma_y$ and $C_k=\frac{1}{\sqrt{2}}\left(\sigma_x-(-1)^k\sigma_y\right)$, the GHZ state, 
\be
\ket{\psi_{GHZ}}=\frac{1}{\sqrt{2}}\left(\ket{000}+\ket{111}\right), \label{GHZ}
\ee
violates the Svetlichny inequality, $\mathcal{S}_{0000}\le4$, to its quantum bound of 
$4\sqrt{2}$. For this choice of measurements,
the GGHZ states,
\be
\ket{\psi_{GGHZ}}=\cos\theta\ket{000}+\sin\theta\ket{111}; \quad 0\le\theta\le\frac{\pi}{4},\label{GGHZ}
\ee
give rise to the isotropic Svetlichny-box in Eq. (\ref{nSv}) with $p_{Sv}=\frac{\sin2\theta}{\sqrt{2}}$. 
Thus, the nonzero irreducible Svetlichny-box component
implies the presence of incompatible measurements and genuine entanglement even if the correlation is local. 
\end{proof}
The observation that Svetlichny-local quantum correlations that have an irreducible Svetlichny-box component can arise from incompatible measurements performed
on the genuinely entangled states
motivates to define a notion of genuine nonclassicality which we call Svetlichny discord.
\begin{definition}\label{df1}
A quantum correlation arising from incompatible measurements performed on a given three-qubit state is said to have \textit{Svetlichny discord} iff the
correlation admits a decomposition with an irreducible Svetlichny-box component. 
\end{definition}
Svetlichny discord is of course not equivalent to Svetlichny nonlocality since there are Svetlichny-local correlations that have an irreducible Svetlichny-box
component; for instance, the isotropic Svetlichny-box in Eq. (\ref{nSv}) has Svetlichny discord if $p_{Sv}>0$ and exhibits Svetlichny nonlocality
if $p_{Sv}>\frac{1}{2}$.

We now define a measure of Svetlichny discord to detect
irreducible Svetlichny-box component in any correlation by using the modulus of the Svetlichny functions in Eq. (\ref{SI}), 
\be
\mathcal{S}_{\alpha\beta\gamma}=\left|\sum_{ijk}(-1)^{i\cdot j \oplus i\cdot k \oplus j\cdot  k \oplus \alpha i\oplus \beta j \oplus \gamma k }\braket{A_iB_jC_k}\right|. \label{mSf}
\ee
\begin{definition}
Svetlichny discord, $\mathcal{G}$, is defined as,
\be
\mathcal{G}=\min\{\mathcal{G}_1,...,\mathcal{G}_9\}, \label{GBD}
\ee
where
\ba
\mathcal{G}_1&=&|\Big||\mathcal{S}_{000}-\mathcal{S}_{001}|-|\mathcal{S}_{010}-\mathcal{S}_{011}|\Big| \nonumber\\
&&-\Big||\mathcal{S}_{100}-\mathcal{S}_{101}|-|\mathcal{S}_{110}-\mathcal{S}_{111}|\Big||\nonumber,
\ea
and the other eight $\mathcal{G}_i$ are obtained by permuting $\mathcal{S}_{\alpha\beta\gamma}$ in $\mathcal{G}_1$. Here $0\le\mathcal{G}\le8$.
\end{definition}
Svetlichny discord is constructed such that it satisfies the following properties: (i) positivity, i.e., $\mathcal{G}\ge0$,
(ii) the bipartite PR-boxes and the deterministic boxes have $\mathcal{G}=0$,
(iii) the algebraic maximum of Svetlichny discord is achieved by the Svetlichny boxes, i.e., $\mathcal{G}=8$ for any Svetlichny-box.
Svetlichny discord is clearly invariant under LRO since the set $\{\mathcal{G}_i\}$ is invariant under LRO.
Svetlichny discord divides the correlations in the two-way local polytope into two disjoint sets: $\mathcal{G}>0$ boxes and $\mathcal{G}=0$ boxes.
Before characterizing the $\mathcal{G}>0$ boxes, we make the following two observations.
\begin{observation}
The set of $\mathcal{G}=0$ boxes forms a subpolytope of the two-way local polytope and is nonconvex. 
\end{observation}
\begin{proof}
The set of $\mathcal{G}=0$ boxes is nonconvex since certain convex mixture 
of the $\mathcal{G}=0$ boxes can have $\mathcal{G}>0$; for instance, the isotropic Svetlichny-box in Eq. (\ref{nSv}) can be written as the convex mixture of the
deterministic boxes if $p_{Sv}\le\frac{1}{2}$, however, it has Svetlichny discord $\mathcal{G}=8p$ if $p_{Sv}>0$. Thus, the set of $\mathcal{G}=0$ boxes forms 
a nonconvex subpolytope of the two-way local polytope as the deterministic boxes and the bipartite PR-boxes have $\mathcal{G}=0$.
\end{proof}

\begin{observation}\label{umSv}
An unequal mixture of any two Svetlichny-boxes: $pP^i_{Sv}+qP^j_{Sv}$, here $p>q$, can be written as the convex sum of an irreducible Svetlichny-box and a Svetlichny-local box.
\end{observation}
\begin{proof}
$pP^i_{Sv}+qP^j_{Sv}=(p-q)P^i_{Sv}+2qP^{ij}_{SvL}$. Here $P^{ij}_{SvL}=\frac{1}{2}(P^i_{Sv}+P^j_{Sv})$ is a Svetlichny-local box since uniform mixture of any two Svetlichny-boxes
belongs to the two-way local polytope.  
Notice that the second Svetlichny-box, $P^j_{Sv}$,  
in the unequal mixture is not irreducible as its presence vanishes with the first Svetlichny-box in the other possible decomposition by the uniform mixture. 
\end{proof}

We obtain the following canonical decomposition of the correlations in $\mathcal{R}$.
\begin{lem}
Any correlation that belongs to the Svetlichny-box polytope can be written as a convex mixture of an irreducible Svetlichny-box and a Svetlichny-local box 
with $\mathcal{G}=0$,
\be
P=\mathcal{G}'P^{\alpha\beta\gamma\epsilon}_{Sv}+(1-\mathcal{G}')P^{\mathcal{G}=0}_{SvL}. \label{canoG>}
\ee
\end{lem}
\begin{proof}
Any correlation given by the decomposition in Eq. (\ref{eq:gendecomp}) can be written as the convex combination 
of the $16$ Svetlichny-boxes and a Svetlichny-local box that does not have the Svetlichny-box components,
\be
P=\sum^{15}_{i=0} g_i P^i_{Sv}+\left(1-\sum^{15}_{i=0}g_i\right)P_{SvL}, \label{st1}
\ee
here $P_{SvL}\ne \sum^{15}_{i=0}p'_iP^i_{Sv}+\sum^{15}_{i=0}q'_iP^i_{12}+\sum^{15}_{i=0}r'_iP^i_{13}+\sum^{15}_{i=0}s'_iP^i_{23}+\sum^{63}_{j=0}t'_jP^j_{D}$ i.e., 
$P_{SvL}$ cannot have nonzero $p'_i$. Thus this decomposition is obtained by maximizing the Svetlichny-box components $p_i$ in Eq. (\ref{eq:gendecomp})
overall possible decompositions.
It follows from the observation \ref{umSv} that the mixture of the  Svetlichny-boxes in the first term
in Eq. (\ref{st1}) can be written as a mixture of a single Svetlichny-box and the $15$ Svetlichny-local boxes, $P^i_{SvL}$,
which are the uniform mixture of two Svetlichny-boxes. 
The largest component of the Svetlichny-box which is unequal to any other Svetlichny-box components in Eq. (\ref{st1}) gives rise to irreducible 
Svetlichny-box component, $\mathcal{G}'$: 
\be
\sum_i g_i P^i_{Sv}=\mathcal{G}'P^{\alpha\beta\gamma\epsilon}_{Sv}+\sum^{15}_{i=1} p_iP^i_{SvL}. \label{st2}
\ee
Here $\mathcal{G}'$ is obtained by minimizing the single Svetlichny-box excess overall possible decompositions 
i.e., $\mathcal{G}'>0$ iif  $\sum_i g_i P^i_{Sv}\ne\sum^{15}_{i=1} q_iP^i_{SvL}$ to ensure that this component is irreducible. 
Substituting Eq. (\ref{st2}) in Eq. (\ref{st1}), we get the canonical decomposition for any correlation in $\mathcal{R}$,
\be
P=\mathcal{G}'P^{\alpha\beta\gamma\epsilon}_{Sv}+(1-\mathcal{G}')P^{\mathcal{G}=0}_{SvL}, \label{canonical}
\ee
where $P^{\mathcal{G}=0}_{SvL}=\frac{1}{1-\mathcal{G'}}\left\{\sum_i p_iP^i_{SvL}+\left(1-\sum_i g_i\right)P_{SvL}\right\}$. 
The fact that the Svetlichny-local box, $P^{\mathcal{G}=0}_{SvL}$, in this decomposition has $\mathcal{G}=0$ 
follows from the geometry of the convex polytope that any point in the polytope lies along a line joining the two points of the polytope: 
Notice that $\mathcal{G}$ divides the two-way local polytope into a $\mathcal{G}>0$ region and $\mathcal{G}=0$ polytope. 
Since the box in the first term in the decomposition given in Eq. (\ref{canonical}) is from the $\mathcal{G}>0$ region and 
the decomposition is for any correlation, 
the box in the second term must be from the $\mathcal{G}=0$ polytope. 
\end{proof}

It follows from the canonical decomposition in Eq. (\ref{canoG>})
that a Svetlichny-local correlation has nonzero Svetlichny discord iff it admits a decomposition with an irreducible Svetlichny-box component.
\begin{cor}
Svetlichny discord of the correlation given by the decomposition in Eq. (\ref{canoG>}) is given by $\mathcal{G}=8\mathcal{G}'$.
\end{cor}
\begin{proof}
The nonextremal correlations in the two-way local polytope can have
the following three types of linear combination due to the convexity of $\mathcal{R}$: 
(i) a convex mixture of two $\mathcal{G}=0$ boxes, 
(ii) a convex mixture of two $\mathcal{G}>0$ boxes and
(iii) a convex mixture of a $\mathcal{G}>0$ box and a $\mathcal{G}=0$ box. 
Since certain convex mixture of the $\mathcal{G}=0$ boxes ($\mathcal{G}>0$ boxes) can have $\mathcal{G}>0$ ($\mathcal{G}=0$),
$\mathcal{G}$ is, in general, not linear for the two decompositions (i) and (ii). However, $\mathcal{G}$ is linear for the decomposition (iii) which
implies that Svetlichny discord for the correlation given by the decomposition in Eq. (\ref{canoG>}) can be evaluated as follows,  
$\mathcal{G}(P)=\mathcal{G}'\mathcal{G}\left(P^{\alpha\beta\gamma\epsilon}_{Sv}\right)+(1-\mathcal{G}')\mathcal{G}\left(P^{\mathcal{G}=0}_{SvL}\right)=8\mathcal{G'}>0$ 
if $\mathcal{G'}>0$. 
\end{proof}
Thus, we say that the decomposition of the correlations given in Eq. (\ref{canoG>}) is canonical 
in that it classifies any box in $\mathcal{R}$ according to whether it has Svetlichny discord or not.  

\begin{cor}
Irreducible Svetlichny-box component, $\mathcal{G}'$, in the canonical decomposition given in Eq. (\ref{canoG>}) is invariant under LRO and permutations of the parties. 
\end{cor}
\begin{proof}
Since $\mathcal{G}$ is invariant under LRO and
permutations of the parties, the irreducible Svetlichny-box
component, $\mathcal{G}'$, in Eq. (\ref{canoG>}) is invariant under LRO.
\end{proof}
\subsection{Mermin-boxes}
For the following choice of incompatible measurements: $A_0=\sigma_x$, $A_1=\sigma_y$,    
$B_0=\sigma_x$, $B_1=\sigma_y$, $C_0=\sigma_x$, and $C_1=\sigma_y$, the correlation arising from the GHZ state can be written as 
an equal mixture of the four bipartite PR-boxes as follows, 
\be
P_{M}(a_m,b_n,c_o|A_i,B_j,C_k)=\frac{1}{4}\sum^4_{\lambda=1} P_\lambda(a_m|A_i)P_\lambda(b_n,c_o|B_j,C_k), \label{Mbdec}
\ee
where $P_1(a_m|A_i)\!=\!\delta^i_{m\oplus i}$, $P_2(a_m|A_i)=\delta^i_{m\oplus i\oplus1}$, $P_3(a_m|A_i)=\delta^i_{m\oplus 1}$, $P_4(a_m|A_i)=\delta^i_m$,
$P_1(b_n,c_o|B_j,C_k)\!=\!P_{PR}^{110}$, $P_2(b_n,c_o|B_j,C_k)=P_{PR}^{111}$, $P_3(b_n,c_o|B_j,C_k)=P_{PR}^{001}$ and 
$P_4(b_n,c_o|B_j,C_k)=P_{PR}^{000}$. Thus, this correlation cannot give rise to the violation of a Svetlichny inequality, however, the correlation
is genuinely nonclassical since it exhibits the GHZ paradox \cite{GHZ}.
Mermin illustrated that the measurements associated with the GHZ paradox 
exhibits KS paradox that illustrates contextuality as well as Bell nonlocality \cite{UNLH}. 
For the measurements that give rise to the correlation in Eq. (\ref{Mbdec}), the outcomes satisfy the following relation:
\be
A_0B_0C_0=-A_0B_1C_1=-A_1B_0C_1=-A_1B_1C_0=1. \label{GHZp}
\ee
It can be inferred from this relation that the correlation  exhibits logical
contradiction with a local(noncontextual)-realistic value assignment to the observables.
We call a maximally two-way nonlocal box that exhibits the logical contradiction with noncontextual-realism Mermin-box; for instance,
the correlation in Eq. (\ref{Mbdec}) represents a Mermin-box as it violates a Mermin inequality \cite{mermin} maximally and exhibits the GHZ paradox.

We say that a Mermin-box exhibits three-way contextuality in analogy with Svetlichny-box which exhibits three-way nonlocality.
Just as there are $16$ Svetlichny-boxes maximally violating only one of the Svetlichny inequalities, 
there are $16$ tripartite Mermin-boxes arising from the GHZ states which maximally violate only one of the Mermin inequalities \cite{WernerWolfmulti},
\be
\mathcal{M}_{\alpha\beta\gamma\epsilon}=(\alpha\oplus\beta\oplus\gamma\oplus1)\mathcal{M}^+_{\alpha\beta\gamma\epsilon}+(\alpha\oplus\beta\oplus\gamma)\mathcal{M}^-_{\alpha\beta\gamma\epsilon}\le2, \label{MI}
\ee
where 
\ba
\mathcal{M}^+_{\alpha\beta\gamma\epsilon}&:=&(-1)^{\gamma\oplus\epsilon}\braket{A_0B_0C_1}
+(-1)^{\beta\oplus\epsilon}\braket{A_0B_1C_0}\nonumber\\&&+(-1)^{\alpha\oplus\epsilon}\braket{A_1B_0C_0}+(-1)^{\alpha\oplus\beta\oplus\gamma\oplus\epsilon\oplus1}\braket{A_1B_1C_1}\nonumber\\
\mathcal{M}^-_{\alpha\beta\gamma\epsilon}&:=&(-1)^{\alpha\oplus\beta\oplus\epsilon\oplus 1}
\braket{A_1B_1C_0}+(-1)^{\alpha\oplus\gamma\oplus\epsilon\oplus 1}\braket{A_1B_0C_1}\nonumber\\&&
+(-1)^{\beta\oplus\gamma\oplus\epsilon\oplus 1}\braket{A_0B_1C_1}+(-1)^{\epsilon}\braket{A_0B_0C_0}.\nonumber
\ea
The Mermin inequalities serve as the criterion for the tripartite EPR-steering under the
constraint that the measurements chosen by each party is noncommuting \cite{UFNL}. In the seminal paper, Mermin inequality was derived by using
anti-commuting observable on each side to show that the correlations arising from the genuinely multipartite entangled states are incompatible with 
the fully LHV model \cite{mermin}, furthermore, this Mermin inequality is equivalent to a noncontextual inequality \cite{Canasetal}.
 
There are two types of two-way nonlocal correlations which can be distinguished according to whether nonlocality is due to tripartite correlations
or bipartite correlations. 
\begin{definition}
We say that a correlation in the two-way nonlocal region exhibits three-way contextuality 
iff the observed nonlocality is due to the tripartite correlation. 
\end{definition}
Just as genuine three-way nonlocal correlations exhibit monogamy of Svetlichny inequality 
violation (see Appendix \ref{msivio}),
three-way contextual correlations exhibit monogamy of Mermin
inequality violation, i.e., a three-way contextual box can violate only one of the Mermin inequalities in Eq. (\ref{MI}).
As the Svetlichny-boxes and the bipartite PR-boxes maximally violate two Mermin inequalities, they do not exhibit monogamy of Mermin inequality violation.
Thus, monogamy of Mermin inequality violation distinguishes three-way contextual correlations from other nonlocal correlations. Mermin-boxes are the extremal
correlations of the set of three-way contextual correlation as they violate a Mermin inequality maximally.

Notice that the Mermin-boxes associated with the GHZ paradox can be decomposed into the uniform mixture of two Svetlichny-boxes; for instance, the Mermin-box
in Eq. (\ref{Mbdec}) can be written as follows,
\be
P_M=\frac{1}{2}(P^{0000}_{Sv}+P^{1110}_{Sv}). \label{cMb}
\ee
Thus, the nonlocality of these maximally two-way nonlocal boxes is not due to the bipartite correlations as they have maximally mixed bipartite marginals.
Not all uniform mixture of two Svetlichny-boxes can give rise to three-way contextuality; 
for instance, white noise can be decomposed into the uniform mixture of the two Svetlichny-boxes. The uniform mixture of two Svetlichny-boxes in a Mermin-box
destroys three-way nonlocality; however, the perfect correlations left in it for the four joint measurements, $A_iB_jC_k$, leads to genuine three-way contextuality \cite{UNLH}.
The decomposition of the Mermin-box given in Eq. (\ref{Mbdec}) 
implies that the set of two-way nonlocal correlations 
which do not possess three-way contextuality is nonconvex in that certain convex mixture of the bipartite PR-boxes gives rise to 
a genuinely three-way contextual correlation. Notice that if we permute the party's indices in the decomposition in Eq. (\ref{Mbdec}), it will also give rise to the Mermin-box.
Thus, three-way contextuality of the correlations are symmetric under the permutations of the parties. 

Two-way local polytope admits two types of Mermin-boxes which can be distinguished by their marginals. 
\begin{observation}
The nonmaximally mixed
bipartite marginals Mermin-boxes are not quantum realizable, whereas the maximally mixed bipartite marginals Mermin-boxes are
quantum realizable.
\end{observation}
\begin{proof}
Consider the following uniform mixture of two bipartite PR-boxes,
\be
P=\frac{1}{2}\sum^2_{\lambda=1} P_\lambda(a_m|A_i)P_\lambda(b_n,c_o|B_j,C_k) \label{Mbnmm}
\ee
where 
$P_1(a_m|A_i)\!=\!\delta^i_{m\oplus i}$, $P_2(a_m|A_i)=\delta^i_{m\oplus 1}$, 
$P_1(b_n,c_o|B_j,C_k)\!=\!P_{PR}^{110}$, and \\$P_2(b_n,c_o|B_j,C_k)
=P_{PR}^{001}$. Notice that 
this correlation that has nonmaximally mixed marginals and the Mermin box in 
Eq. (\ref{Mbdec}) which has maximally mixed marginals are equivalent with respect to the joint expectations $\braket{A_iB_jC_k}$. Thus, the correlation in Eq. (\ref{Mbnmm}) 
also exhibits the logical contradiction with local-realism and violate only one of the Mermin inequalities. 
Notice that the marginal distribution $P(a_m|A_i)$ of the Mermin box
in Eq. (\ref{Mbnmm}) has the deterministic outcome for the input $A_1$ and fully random outcomes for the input $A_0$. 
Since there does not exist a quantum state that can give rise to 
the deterministic outcome and random outcomes simultaneously, the Mermin boxes with nonmaximally mixed marginals are nonquantum boxes.  
\end{proof}

\subsection{Mermin discord and 3-decomposition}
Consider isotropic Mermin-box which is a convex mixture of the Mermin-box in Eq. (\ref{Mbdec}) and white noise,
\be
P=p_MP_M+(1-p_M)P_N, \label{nM}
\ee
The isotropic Mermin-box violates the Mermin inequality i.e., $\mathcal{M}_{0010}=4p_M>2$ if $p_M>\frac{1}{2}$.
Notice that even if the isotropic Mermin-box is local when $p_M\le\frac{1}{2}$, it admits a decomposition that has the single Mermin-box component.
We call such a single Mermin-box in any correlation (nonlocal, or not) irreducible Mermin-box.

The following observation can be illustrated by the isotropic Mermin-box.
\begin{observation}\label{m2}
When a local quantum correlation arising from a given genuinely entangled state has an irreducible Mermin-box component,
the correlation arises from incompatible measurements that give rise to three-way contextuality.
\end{observation}
\begin{proof}
For the incompatible measurements that give rise to the GHZ paradox in Eq. (\ref{GHZp}),
the GGHZ states in Eq. (\ref{GGHZ}) give rise to the isotropic Mermin-box in Eq. (\ref{nM}) with $p_M=\sin2\theta$. 
Thus, the nonzero irreducible Mermin-box component
implies the presence of incompatible measurements and genuine entanglement even if the correlation is local. 
\end{proof}
The observation that local quantum correlations that have an irreducible Mermin-box component can arise from incompatible measurements performed
on the genuinely entangled states
motivates to define a notion of genuine nonclassicality which we call Mermin discord.
\begin{definition}\label{df2}
A quantum correlation arising from incompatible measurements performed on a given three-qubit state is said to have \textit{Mermin discord} iff the
correlation admits a decomposition with an irreducible Mermin-box component. 
\end{definition}
Mermin discord is not equivalent to three-way contextuality since the correlations that do not violate a Mermin inequality can also have
an irreducible Mermin-box component; for instance, the isotropic Mermin-box in Eq. (\ref{nM}) has Mermin discord if $p_M>0$ and 
exhibits three-way contextuality if $p_M>\frac{1}{2}$.

\begin{observation}
For any Mermin-box, only one of the Mermin functions,
$\mathcal{M}_{\alpha\beta\gamma}:=|\mathcal{M}_{\alpha\beta\gamma\epsilon}|$,
attains the maximum and the rest of them take zero, where $\mathcal{M}_{\alpha\beta\gamma\epsilon}$ are the Mermin operators given in Eq. (\ref{MI}).
\end{observation}
The above observation motivates us to define a measure of Mermin discord using the Mermin functions similar to the measure of Svetlichny discord. 
\begin{definition}\label{MDdef}
Mermin discord, $\mathcal{Q}$, is defined as,
\be
\mathcal{Q}=\min\{\mathcal{Q}_1,...,\mathcal{Q}_9\}, \label{GMD}
\ee
where
\ba
\mathcal{Q}_1&=&|\Big||\mathcal{M}_{000}-\mathcal{M}_{001}|-|\mathcal{M}_{010}-\mathcal{M}_{011}|\Big| \nonumber\\
&&-\Big||\mathcal{M}_{100}-\mathcal{M}_{101}|-|\mathcal{M}_{110}-\mathcal{M}_{111}|\Big||\nonumber,
\ea
and the other eight $\mathcal{Q}_i$ are obtained by permuting $\mathcal{M}_{\alpha\beta\gamma}$ in $\mathcal{Q}_1$. Here $0\le\mathcal{Q}\le4$.
\end{definition}
Mermin discord is constructed such that it satisfies the following properties: (i) $\mathcal{Q}=0$ for the Svetlichny-boxes, bipartite PR-boxes and deterministic boxes
(ii) the algebraic maximum of $\mathcal{Q}$ is achieved by the Mermin boxes, i.e., $\mathcal{Q}=4$ for any Mermin-box and
(iii) $\mathcal{Q}$ is invariant under LRO since the set $\{\mathcal{Q}_i\}$ is invariant under LRO.

We obtain the following observations from the Mermin discord defined in Eq. (\ref{GMD}).
\begin{observation}
The set of $\mathcal{Q}=0$ boxes in $\mathcal{R}$ forms a nonconvex subpolytope of the full Svetlichny-box polytope. 
\end{observation}
\begin{proof}
Since the extremal boxes of the Svetlichny-box polytope have $\mathcal{Q}=0$, and certain convex mixture of the $\mathcal{Q}=0$ boxes can have $\mathcal{Q}>0$,
the set of $\mathcal{Q}=0$ boxes 
forms a nonconvex subpolytope of the full Svetlichny-box polytope.
\end{proof}
\begin{observation}\label{GQG=0}
$\mathcal{Q}$ divides the $\mathcal{G}=0$ polytope into a $\mathcal{Q}>0$ region and $\mathcal{G}=\mathcal{Q}=0$ nonconvex polytope. 
\end{observation}
\begin{proof}
Since all the bipartite PR-boxes and deterministic boxes have $\mathcal{G}=\mathcal{Q}=0$ and certain convex mixture of these extremal boxes can have $\mathcal{Q}>0$,
the set of $\mathcal{G}=\mathcal{Q}=0$ boxes 
forms a nonconvex subpolytope of the $\mathcal{G}=0$ polytope.
\end{proof}
\begin{observation}\label{obsQ=4}
A $\mathcal{Q}=4$ box is, in general, 
a convex combination of a quantum Mermin-box and the four non-quantum Mermin-boxes which are equivalent with respect to $\braket{A_iB_jC_k}$,
\be
P_{\mathcal{Q}=4}=uP^{Q}_M+\sum^4_{i=1}v_iP^{nQ_i}_{M},
\ee
where $P^{Q}_M$ has maximally mixed bipartite marginals and $P^{nQ}_{M_i}$ have nonmaximally mixed bipartite marginals; all the Mermin-boxes in this decomposition
violate the same Mermin inequality as they are equivalent with respect to $\braket{A_iB_jC_k}$.  
\end{observation}
\begin{proof}
Notice that any convex mixture of the two Mermin boxes in Eqs. (\ref{Mbdec}) and (\ref{Mbnmm}) have $\mathcal{Q}=4$. There are four nonquantum Mermin boxes which 
are equivalent with respect to $\braket{A_iB_jC_k}$ corresponding to a given quantum Mermin box. Thus, any convex mixture of these five Mermin boxes have
$\mathcal{Q}=4$.
\end{proof}

We obtain the following $3$-decomposition fact of the Svetlichny-box polytope.
\begin{theorem}
Any correlation in $\mathcal{R}$ given by the decomposition in Eq. (\ref{eq:gendecomp}) can be written as a convex mixture of a Svetlichny-box,
a maximally two-way nonlocal box with $\mathcal{Q}=4$ and a box with $\mathcal{G}=\mathcal{Q}=0$,
\be
P=\mathcal{G}'P^{\alpha\beta\gamma\epsilon}_{Sv}+\mathcal{Q}'P_{\mathcal{Q}=4}+(1-\mathcal{G}'-\mathcal{Q}')P^{\mathcal{G}=0}_{\mathcal{Q}=0}. \label{3dfact}
\ee
\end{theorem}
\begin{proof}
Since all the Mermin-boxes have $\mathcal{G}=0$, they belong to the $\mathcal{G}=0$ polytope. 
Therefore, any $\mathcal{G}=0$ box can be written as a convex mixture of the Mermin-boxes and a Svetlichny-local box that does not have the Mermin-box components,
\be
P^{\mathcal{G}=0}_{SvL}=\sum^{15}_{i=0} u_iP^{Q_i}_{M}+\sum^{64}_{j=1}v_jP_{M}^{nQ_j}+\left(1-\sum^{15}_{i=0}u_i-\sum^{64}_{j=1}v_j\right)P_{SvL}, \label{g=0}
\ee
where $P^{Q_i}_{M}$ and $P_{M}^{nQ_j}$ are quantum and non-quantum Mermin-boxes.
It follows from the observation \ref{obsQ=4} that the mixture of the Mermin boxes in this decomposition can be written as 
the mixture of the $16$ maximally two-way nonlocal boxes that have $\mathcal{Q}=4$.
Notice that unequal mixture of any two $\mathcal{Q}=4$ boxes that violate the two different Mermin inequalities in Eq. (\ref{MI}): 
$pP^1_{\mathcal{Q}=4}+qP^2_{\mathcal{Q}=4}$, $p>q$,
can be written as a mixture of an irreducible $\mathcal{Q}=4$ box and a local box which is a uniform mixture of the two $\mathcal{Q}=4$ boxes: 
$(p-q)P^1_{\mathcal{Q}=4}+2qP_L$, here $P_L=\frac{1}{2}\left(P^1_{\mathcal{Q}=4}+P^2_{\mathcal{Q}=4}\right)$ is a Bell-local box which has $\mathcal{Q}=0$.
Therefore, the first term in the decomposition given
in Eq. (\ref{g=0}) can be written as a mixture of an irreducible $\mathcal{Q}=4$ box and a Bell-local box, 
\be
\sum^{15}_{i=0} u_iP^{Q_i}_{M}+\sum_{j}v_jP_{M}^{nQ_j}=\mathcal{Q}''P_{\mathcal{Q}=4}+\sum^{15}_{i=1}l_iP^{i}_L, \label{q''}
\ee
where $P^{i}_L$ are the Bell-local
boxes which are the uniform mixture of two $\mathcal{Q}=4$ boxes. Here $\mathcal{Q}''$ is obtained by minimizing the single $\mathcal{Q}=4$ box
excess overall possible decompositions i.e., $\mathcal{Q}''>0$ iff $\sum^{15}_{i=0} u_iP^{Q_i}_{M}+\sum_{j}v_jP_{M}^{nQ_j}\ne\sum^{15}_{i=1}l'_iP^{i}_L$. 
Substituting Eq. (\ref{q''}) in Eq. (\ref{g=0}), we obtain the canonical decomposition of the $\mathcal{G}=0$ correlations,  
\be
P^{\mathcal{G}=0}_{SvL}=\mathcal{Q}''P_{\mathcal{Q}=4}+(1-\mathcal{Q}'')P^{\mathcal{G}=0}_{\mathcal{Q}=0}, \label{g=0cano}
\ee
where $P^{\mathcal{G}=0}_{\mathcal{Q}=0}=\frac{1}{1-\mathcal{Q}''}\left\{\sum^{15}_{i=1}l_iP^{i}_L+\left(1-\sum^{15}_{i=0}u_i-\sum_{j}v_j\right)P_{SvL}\right\}$.
The fact that the box in the second term in this decomposition has $\mathcal{G}=\mathcal{Q}=0$ 
follows from the geometry of the $\mathcal{G}=0$ polytope: The observation \ref{GQG=0} implies that any correlation
in the $\mathcal{G}=0$ polytope lies on a line segment joining a $\mathcal{Q}>0$ box and a $\mathcal{G}=\mathcal{Q}=0$ box.
Therefore, the box in the second term in the decomposition given in Eq. (\ref{g=0cano}) must have $\mathcal{G}=\mathcal{Q}=0$ 
as the box in the first term has $\mathcal{Q}>0$.  Thus,
decomposing the $\mathcal{G}=0$ box in Eq. (\ref{canoG>}) as given in Eq. (\ref{g=0cano}) gives the 
canonical decomposition given in Eq. (\ref{3dfact}) with $\mathcal{Q}'=\mathcal{Q}''(1-\mathcal{G}')$. 
\end{proof}

\begin{cor}
A correlation has nonzero Mermin discord iff it admits a decomposition with an irreducible Mermin box component
since Mermin discord $\mathcal{Q}=4\mathcal{Q}'$ for the correlation given by the canonical decomposition in Eq. (\ref{3dfact}). 
\end{cor}
\begin{proof}
Any correlation in $\mathcal{R}$ given by the $3$-decomposition in Eq. (\ref{3dfact}) can be written as a convex mixture of a maximally two-way nonlocal box with $\mathcal{Q}=4$
and a box with $\mathcal{Q}=0$,
\be
P=\mathcal{Q}'P_{\mathcal{Q}=4}+(1-\mathcal{Q}')P_{\mathcal{Q}=0},\label{Qcan}
\ee
where $P_{\mathcal{Q}=0}=\frac{1}{1-\mathcal{Q}'}\left((1-\mathcal{G}'-\mathcal{Q}')P^{\mathcal{G}=0}_{\mathcal{Q}=0}+\mathcal{G}'P^{\alpha\beta\gamma\epsilon}_{Sv}\right)$.
The nonconvexity property of the $\mathcal{Q}=0$ polytope implies that certain convex combination of the $\mathcal{Q}=0$ boxes can have $\mathcal{Q}>0$
and there are $\mathcal{Q}=0$ boxes  which 
can be written as a convex mixture of two $\mathcal{Q}>0$ boxes. Thus, $\mathcal{Q}$ is not linear for these two types of decomposition.
However, $\mathcal{Q}$ is linear for the decomposition given in Eq. (\ref{Qcan}) since the convex mixture of a $\mathcal{Q}>0$ box and a $\mathcal{Q}=0$ box 
is always a $\mathcal{Q}>0$ box. 
Therefore, Mermin discord of the correlation in Eq. (\ref{Qcan})
is given by $\mathcal{Q}(P)=\mathcal{Q}'\mathcal{Q}(P_{\mathcal{Q}=4})+(1-\mathcal{Q}')\mathcal{Q}(P^{\mathcal{Q}=0}_{SvL})=4\mathcal{Q}'>0$
if $\mathcal{Q}'>0$. As any correlation that has an 
irreducible Mermin-box component lies on a line segment joining a Mermin-box and a $\mathcal{Q}=0$ box, it has $\mathcal{Q}>0$.
\end{proof}
\subsection{Monogamy between the measures}
As the total amount of irreducible Svetlichny-box and
irreducible Mermin-box components of a correlation given by the decomposition in Eq. (\ref{3dfact}) is constrained i.e., $\mathcal{G}'+\mathcal{Q}'\le1$ 
which follows from the probability constraint in the $3$-decomposition, 
we obtain the following
trade-off relation.
\begin{cor}
Svetlichny discord and Mermin discord of any given correlation satisfy
the following monogamy relation,
\be
\mathcal{G}+2\mathcal{Q}\le8. \label{n3nl3c}
\ee 
\end{cor}
This tradeoff relation reveals monogamy between three-way contextual correlations and three-way nonlocal correlations and is 
analogous to the monogamy relations between locally contextual correlations and nonlocal correlations derived by 
Kurzy\'{n}ski \etal \cite{KCK}. The monogamy relations given by Kurzy\'{n}ski \etal 
implies that when measurements on qutrit system gives rise to contextuality in a qutrit-qubit entangled system, then these measurements
do not give rise to nonlocality for all measurements on qubit system. Similar monogamy character follows from the observations \ref{m1} and \ref{m2}:
For the measurements that gives rise to the GHZ paradox, the GHZ state gives rise to maximal Mermin discord and zero Svetlichny discord, i.e., $\mathcal{Q}=4$
and $\mathcal{G}=0$ which is consistent with Eq. (\ref{n3nl3c}). Thus, for the measurements that give rise to the GHZ paradox, the GGHZ states give rise
to only Mermin discord i.e., $\mathcal{Q}=4\sin2\theta$ and $\mathcal{G}=0$.
Notice that for the measurements that gives rise to maximal three-way nonlocality, the GGHZ states give rise to only Svetlichny discord, i.e., $\mathcal{G}=4\sqrt{2}\sin2\theta$
and $\mathcal{Q}=0$.     
Thus, we see that the measurements that gives rise to extremal three-way contextuality do not give
rise to three-way nonlocality and vice versa. 

For general incompatible measurements, quantum correlations can have three-way contextuality and three-way nonlocality
simultaneously, however, the
tradeoff exists between three-way nonlocality and three-way contextuality as given by Eq. (\ref{n3nl3c}).
For instance, the correlations arising from the GHZ state for the measurements 
$A_0=\sigma_x$, $A_1=\sigma_y$,
$B_0=\sqrt{p}\sigma_x-\sqrt{1-p}\sigma_y$, $B_1=\sqrt{1-p}\sigma_x+\sqrt{p}\sigma_y$,
$C_0=\sigma_x$ and $C_1=\sigma_y$
can be decomposed into the Svetlichny-box, the Mermin-box which is a uniform mixture of two Svetlichny-boxes, and white noise as follows,
\be
P=\mathcal{G}'P^{0000}_{Sv}+\mathcal{Q}'\left(\frac{P^{0000}_{Sv}+P^{1110}_{Sv}}{2}\right)+(1-\mathcal{G}'-\mathcal{Q}')P_N, \label{SMDghz}
\ee
where $\mathcal{G}'=\sqrt{1-p}$, $\mathcal{Q}'=\sqrt{p}-\sqrt{1-p}$ and $\frac{1}{2}\le p\le1$. These correlations have $\mathcal{G}+\mathcal{Q}=4\sqrt{p}\le4$. 
\section{Quantum correlations}\label{QC}
We will observe that 
any tripartite qubit correlation in the Svetlichny-box polytope can be decomposed into Svetlichny-box, a Mermin box with maximally mixed marginals
and a box with $\mathcal{G}=\mathcal{Q}=0$,
\be
P=\mathcal{G}'P^{\alpha\beta\gamma\epsilon}_{Sv}+\mathcal{Q}'P_M+(1-\mathcal{G}'-\mathcal{Q}')P^{\mathcal{G}=0}_{\mathcal{Q}=0}. \label{mc}
\ee
We will characterize genuine nonclassicality of quantum correlations arising from local projective measurements along the 
directions $\hat{a}_i$, $\hat{b}_j$ and $\hat{c}_k$ on the three-qubit systems using this three-way decomposition. 

We will apply Svetlichny discord and Mermin discord to quantify nonclassicality of correlations arising from two inequivalent 
classes of pure genuinely entangled states  \cite{Dur} and the Werner states. For these states, a nonzero Svetlichny discord originates from 
incompatible measurements that give rise to Svetlichny nonlocality. Similarly, a nonzero Mermin discord originates from
incompatible measurements that give rise to three-way contextuality.
For a given nonclassical quantum state, 
there are three different incompatible measurements corresponding to (i) Svetlichny discordant correlation which has $\mathcal{G}>0$ and $\mathcal{Q}=0$,
(ii) Mermin discordant correlation which has $\mathcal{G}=0$ and $\mathcal{Q}>0$ and
(iii) Svetlichny-Mermin discordant correlation which has $\mathcal{G}>0$ and $\mathcal{Q}>0$.
Three-way nonlocal quantum correlations in $\mathcal{R}$ are the subset of $\mathcal{G}>0$ correlations, 
whereas three-way contextual quantum correlations are the subset of $\mathcal{Q}>0$ correlations. 

Svetlichny (Mermin) discord for a given nonclassical state is maximized by minimizing the number of nonzero Svetlichny (Mermin) functions overall incompatible measurements
that give rise to $\mathcal{G}>0$ ($\mathcal{Q}>0$).
In the subsequent sections, we will choose the following four measurement settings:  
\be
\hat{a}_0\!=\!\hat{x}, \quad \hat{a}_1\!=\!\hat{y},\quad 
\hat{b}_j\!=\!\frac{1}{\sqrt{2}}\left(\hat{x}+(-1)^{j\oplus 1}\hat{y}\right),\quad \hat{c}_0\!=\!\hat{x},\quad \hat{c}_1\!=\!\hat{y} \label{SDxy}
\ee
\be
\hat{a}_0\!=\!\hat{z}, \quad \hat{a}_1\!=\!\hat{x},\quad 
\hat{b}_j\!=\!\frac{1}{\sqrt{2}}\left(\hat{z}+(-1)^j\hat{x}\right),\quad
\hat{c}_0\!=\!\hat{z},\quad\hat{c}_1\!=\!\hat{x} \label{SDxz}
\ee
\be
\hat{a}_0=\hat{x}, \quad \hat{a}_1=\hat{y},\quad 
\hat{b}_0=\hat{x}, \quad \hat{b}_1=\hat{y},\quad
\hat{c}_0=\hat{x},\quad\hat{c}_1=\hat{y} \label{MDxy}
\ee
\be
\hat{a}_0=\hat{z}, \quad \hat{a}_1=\hat{x},\quad 
\hat{b}_0=\hat{z}, \quad \hat{b}_1=\hat{x},\quad
\hat{c}_0=\hat{z},\quad\hat{c}_1=\hat{x} \label{MDxz}
\ee
for studying correlations arising from the genuinely nonclassical quantum states. The first two settings correspond to Svetlichny discordant correlations, whereas the last two settings correspond to Mermin discordant correlations. 
We will apply Svetlichny and Mermin discord to various states in order to illustrate the new insights that may be obtained regarding the origin of genuine nonclassicality. 
We will also apply the two bipartite measures, Bell and Mermin discord \cite{Jeba}, to the marginal correlations.
We denote the Bell and Mermin discord by $\mathcal{G}_{ij}$ and $\mathcal{Q}_{ij}$, here $ij$ indicates Bell/Mermin discord is between which two qubits.
\subsection{GHZ-class states}
The GHZ-class states which have bipartite entanglement between $A$ and $B$ are given as follows,
\be
|\psi_{gs}\rangle=\cos\theta|000\rangle+\sin\theta|11\rangle\Big\{\cos\theta_3|0\rangle+\sin\theta_3|1\rangle\Big\}\;.
\label{gs}
\ee
The genuine tripartite entanglement is quantified by the three tangle \cite{CKW}, $\tau_3=(\sin2\theta\sin\theta_3)^2$, 
and the bipartite entanglement is quantified by the concurrence \cite{WKW}, $C_{12}=\sin2\theta\cos\theta_3$.
\subsubsection{Svetlichny discordant box}
The settings in Eq. (\ref{SDxy}) maximizes Svetlichny discord for the GHZ-class states, 
since the correlations have only one of the Svetlichny functions nonzero i.e., $\mathcal{S}_{000}=4\sqrt{2\tau_3}$ and the rest of the 
Svetlichny functions are zero which implies that Svetlichny discord $\mathcal{G}=4\sqrt{2\tau_3}$. The correlations can be decomposed as follows,
\be
P=\frac{\sqrt{\tau_3}}{\sqrt{2}}P^{0000}_{Sv}+\left(1-\frac{\sqrt{\tau_3}}{\sqrt{2}}\right)P^{\mathcal{G}=0}_{SvL}, \label{optSDdec}
\ee
where the $\mathcal{G}=0$ box, $P^{\mathcal{G}=0}_{SvL}$, is given in Eq. (\ref{G=01}).  
These correlations are Svetlichny-local if $0\le\tau_3\le\frac{1}{2}$, however, they have genuine nonclassicality originating from
incompatible measurements that give rise to Svetlichny nonlocality  if $\mathcal{\tau}_3>0$. 
In addition to Svetlichny discord, the correlations have Bell discord between $A$ and $B$, $\mathcal{G}_{12}=2\sqrt{2}C_{12}$.

\begin{figure}[h!]
\centering
\includegraphics[scale=0.75]{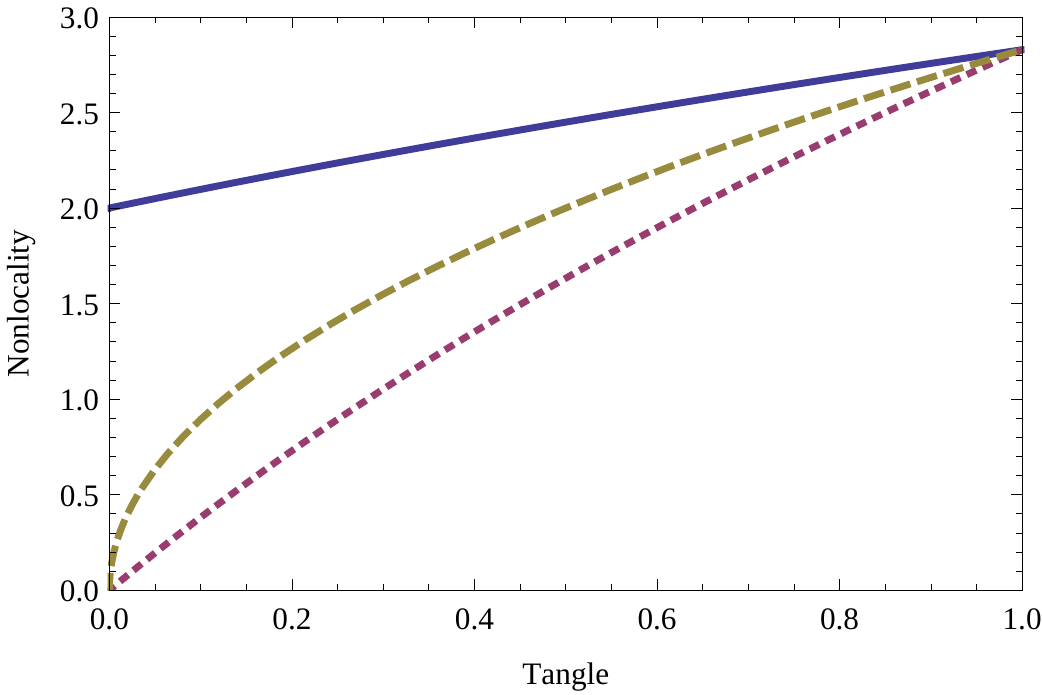}
\caption[Svetlichny nonlocality versus Svetlichny discord]{Dashed line shows the plots of the Svetlichny inequality violation
and Svetlichny discord for the JPD given in Eq. (\ref{optSDdec}) with $\theta=\frac{\pi}{4}$. Solid and dotted lines show the plots of the Svetlichny inequality violation and
Svetlichny discord respectively for the JPD given in Eq. (\ref{optSvdec}) with $\theta=\frac{\pi}{4}$.
We observe that the JPD in Eq. (\ref{optSvdec}) which gives optimal violation of the Svetlichny inequality does not give optimal Svetlichny discord for the
GHZ-class states, $\frac{1}{\sqrt{2}}\ket{000}+\frac{1}{\sqrt{2}}\ket{11}\{\cos\theta_3\ket{0}+\sin\theta_3\ket{1}\}$.}\label{plotine}
\end{figure}

Ghose \etal \cite{Ghoseetal} provided optimal measurement settings that give maximal violation of the Svetlichny inequality with respect
to the GHZ-class states; for instance,
the settings $\hat{a}_i=\frac{1}{\sqrt{2}}\left(\hat{x}+(-1)^{i}\hat{y}\right),
\hat{b}_j=\frac{1}{\sqrt{2}}\left(\hat{x}+(-1)^{j\oplus 1}\hat{y}\right),
\hat{c}_k=\frac{\sin\theta_3}{\sqrt{1+\sin^2\theta_3}}\hat{x}+(-1)^{k\oplus 1}\frac{\sin\theta_3}{\sqrt{1+\sin^2\theta_3}}\hat{y}+
\frac{\cos\theta_3}{\sqrt{1+\sin^2\theta_3}}\hat{z}$ gives rise to the violation of the Svetlichny inequality, $\mathcal{S}_{0000}=4\sqrt{C^2_{12}+2\tau_3}>4$, if $C^2_{12}+2\tau_3>1$.
For this optimal settings, the correlations admit the following decomposition,
\ba
P&=&\frac{\tau_3}{\sqrt{C^2_{12}+2\tau_3}}P^{0000}_{Sv}+\left(1-\frac{\tau_3}{\sqrt{C^2_{12}+2\tau_3}}\right)P^{\mathcal{G}=0}_{SvL}, \label{optSvdec}
\ea
where the $\mathcal{G}=0$ box, $P^{\mathcal{G}=0}_{SvL}$, is given in Eq. (\ref{G=02}). These correlations have Svetlichny discord $\mathcal{G}=\frac{8\tau_3}{\sqrt{C^2_{12}+2\tau_3}}$ which is nonzero if the state is genuinely entangled as the
correlations have the irreducible Svetlichny-box component. Thus, the Svetlichny-local correlations in Eq. (\ref{optSvdec}) have three-way nonclassicality
originating from Svetlichny nonlocality when $0<C^2_{12}+2\tau_3\le1$.

Notice that the correlations in Eq. (\ref{optSvdec}) have less irreducible Svetlichny-box component than the correlations in Eq. (\ref{optSDdec})
for a given amount of entanglement quantified by the three-tangle (see fig. \ref{plotine}).
Thus, for the pure states, the measurement settings which is optimal for Svetlichny discord does not, in general, maximize the violation of the Svetlichny inequality
and vice versa. For the GGHZ states, the correlations in Eqs. (\ref{optSDdec}) and (\ref{optSvdec}) become the isotropic Svetlichny-box,
\be
P=\frac{\sqrt{\tau_3}}{\sqrt{2}}P^{0000}_{Sv}+\left(1-\frac{\sqrt{\tau_3}}{\sqrt{2}}\right)P_N.
\ee

\subsubsection{Mermin discordant box}
The settings in Eq. (\ref{MDxy}) maximizes Mermin discord for the GHZ-class states, since only one of the Mermin functions is nonzero for this settings. 
The correlations can be written as a convex mixture of the Mermin-box and a Bell-local box:   
\be
P=\sqrt{\tau_3}\left(\frac{P^{0000}_{Sv}+P^{1110}_{Sv}}{2}\right)+\left(1-\sqrt{\tau_3}\right)P^{\mathcal{Q}=0}_{L}, \label{MDGHZop}
\ee
where the Bell-local box, $P^{\mathcal{Q}=0}_{L}$, which has $\mathcal{Q}=0$ is given in Eq. (\ref{Q=01}).
These correlations have Mermin discord $\mathcal{Q}=4\sqrt{\tau_{3}}$ and bipartite Mermin discord $\mathcal{Q}_{12}=2\sqrt{\tau_{12}}$.
Despite the correlations violate the Mermin inequality only if $\tau_3>\frac{1}{4}$,
they have genuine three-way nonclassicality originating 
from three-way contextuality if $\tau_3>0$.

Consider the following state dependent settings: $\hat{a}_0=\hat{x}$, $\hat{a}_0=\hat{y}$,
$\hat{b}_j=\frac{1}{\sqrt{2}}\left(\hat{x}+(-1)^{j\oplus 1}\hat{y}\right)$,
$\hat{c}_k=\frac{\sin\theta_3}{\sqrt{1+\sin^2\theta_3}}\hat{x}+(-1)^{k}\frac{\sin\theta_3}{\sqrt{1+\sin^2\theta_3}}\hat{y}+
\frac{\cos\theta_3}{\sqrt{1+\sin^2\theta_3}}\hat{z}$ which gives rise to optimal three-way contextuality.
For this settings, the GHZ-class states give rise to two nonzero Mermin functions $\mathcal{M}_{000}=\frac{2\sqrt{2}C^2_{12}}{\sqrt{(C^2_{12}+2\tau_3)}}$
and $\mathcal{M}_{110}=2\sqrt{2(C^2_{12}+2\tau_3)}$ which implies that there are GHZ-class states that give rise to the violation of two Mermin inequalities.
Notice that all the GHZ-class states with $\theta=\frac{\pi}{4}$ give rise to three-way contextuality since they exhibit monogamy of Mermin inequality violation.
The correlations admit the following decomposition,
\ba
P=\frac{\sqrt{2}\tau_3}{\sqrt{C^2_{12}+2\tau_3}}\left(\frac{P^{0000}_{Sv}+P^{1110}_{Sv}}{2}\right)
+\left(1-\frac{\sqrt{2}\tau_3}{\sqrt{C^2_{12}+2\tau_3}}\right)P^{\mathcal{Q}=0}_L, \label{MIGHZop}
\ea
where the Bell-local box, $P^{\mathcal{Q}=0}_L$, is given in Eq. (\ref{Q=02}). These correlations have tripartite Mermin discord
$\mathcal{Q}=\frac{4\sqrt{2}\tau_3}{\sqrt{C^2_{12}+2\tau_3}}$
and bipartite Bell discord $\mathcal{G}_{12}=2\sqrt{2}C_{12}$.
Notice that the correlations in Eq. (\ref{MIGHZop}) have less irreducible tripartite Mermin-box component than the correlations in Eq. (\ref{MDGHZop}) for a given amount of entanglement.
For the GGHZ states, both the correlations in Eqs. (\ref{MDGHZop}) and (\ref{MIGHZop}) become the isotropic Mermin-box,
\be
P=\sqrt{\tau_3}\left(\frac{P^{0000}_{Sv}+P^{1110}_{Sv}}{2}\right)+\left(1-\sqrt{\tau_3}\right)P_N.
\ee
\subsubsection{Svetlichny-Mermin discordant box}
For the following state dependent measurement settings:
$\hat{a}_0=\hat{x}$, $\hat{a}_1=\hat{y}$,
$\hat{b}_0=\sin2\theta\hat{x}-\cos2\theta\hat{y}$,
$\hat{b}_1=\cos2\theta\hat{x}+\sin2\theta\hat{y}$,
$\hat{c}_0=\hat{x}$ and $\hat{c}_1=\hat{y}$, the GGHZ state in Eq. (\ref{GGHZ}) gives rise to Svetlichny discord and Mermin discord simultaneously:
\ba
\mathcal{G}&=&\left\{\begin{array}{lr}
8\tau_3 \quad \text{when} \quad 0 \le \theta \le \frac{\pi}{8}\\ 
8\sqrt{\tau_3(1-\tau_3)}  \quad \text{when} \quad \frac{\pi}{8} \le \theta \le \frac{\pi}{4}\\ 
\end{array}
\right.\nonumber\\
&>&0 \quad \text{if} \quad \tau_3\ne0,1 \nonumber\\
\mathcal{Q}&=&4\left|\tau_3-\sqrt{\tau_3(1-\tau_3)}\right|\nonumber\\
&>&0 \quad \text{if} \quad \tau_3\ne0,\frac{1}{2}. \nonumber
\ea
The correlations have a $3$-decomposition as follows,
\be
P=\mathcal{G}'P^{0000}_{Sv}+\mathcal{Q}'\left(\frac{P^{0000}_{Sv}+P^{111\gamma}_{Sv}}{2}\right)+\left(1-\mathcal{G}'-\mathcal{Q}'\right)P_N, \label{cdQC}
\ee
where $\mathcal{G}'=\mathcal{G}/8$ and $\mathcal{Q}'=\mathcal{Q}/4$. 
Since the measurement settings corresponds to the GHZ paradox when $\theta=\pi/4$ and maximal three-way nonlocality  when $\theta=\pi/8$, 
the correlation has zero irreducible Svetlichny-box component when $\theta=\pi/4$ and zero irreducible Mermin-box component when $\theta=\pi/8$. 
\subsubsection{Svetlichny-box polytope vs three-way nonlocal quantum correlations}
Bancal \etal \cite{Banceletal} conjectured that all pure genuinely entangled states can give rise to three-way nonlocal correlations 
and it was noticed that there are three-way nonlocal quantum correlations arising from the pure states which do not violate a Svetlichny inequality.
In Ref. \cite{gnl99}, it has been shown that all
the GGHZ states can give rise to the violation of a class $99$ facet inequality whose representative is given in Eq. (\ref{NSFI}). 
For instance, the correlation arising from 
the GGHZ states in Eq. (\ref{GGHZ}) has $\mathcal{L}_2^{99}=1+2\sqrt{1+\sin^22\theta}>3$ if $\tau_3>0$ for the measurement settings $\hat{a}_0=\hat{z}$, $\hat{a}_1=\hat{x}$,
$\hat{b}_j=\cos t\hat{z}+(-1)^{j}\sin t\hat{x}$,
$\hat{c}_0=\hat{z}$ and $\hat{c}_1=\hat{x}$, where $\cos t=\frac{1}{\sqrt{1+\sin^22\theta}}$ .
For $\theta=\frac{\pi}{4}$, the correlation violates this inequality to its quantum bound of $1+2\sqrt{2}$ and 
can be decomposed in a convex mixture of the class $8$ extremal box given in the table of Ref. \cite{Pironioetal} and a local box,
\be
P=\frac{1}{\sqrt{2}} P_8+\left(1-\frac{1}{\sqrt{2}}\right)P_L. \label{GHZopGNL}
\ee
Here $P_L$ arises from the state $\rho=\rho_{AC}\otimes \frac{\openone}{2}$, where $\rho_{AC}=\frac{1}{2}\left(\ketbra{00}{00}+\ketbra{11}{11}\right)$.
As genuine nonlocality of the correlation is due to the class $8$ extremal box, the correlation does not violate a Svetlichny inequality
and hence it does not belong to the three-way nonlocal region of the Svetlichny-box polytope. Notice that the correlation in Eq. (\ref{GHZopGNL}) has $\mathcal{G}=\mathcal{Q}=0$.
\subsection{W-class states}
We now study the correlations arising from the W-class states,
\be
\ket{\psi_w}=\alpha\ket{100}+\beta\ket{010}+\gamma\ket{001}, \label{W}
\ee
We may consider the three nonvanishing bipartite concurrences 
$C_{12}=2\alpha\beta$, $C_{13}=2\alpha\gamma$ and $C_{23}=2\beta\gamma$ or the minimal concurrence of assistance \cite{Chietal}
$C^a_{min}=\min\{C_{12},C_{13},C_{23}\}$ as genuine tripartite entanglement measure for W-class states.
The optimal settings that maximizes Svetlichny/Mermin discord 
for the GHZ-class states do not maximize Svetlichny/Mermin discord for the W-class states. 
\subsubsection{Svetlichny discordant box}
Svetlichny discord for the W-class states is maximized by the settings in Eq. (\ref{SDxz}) which gives rise to,
\ba
\mathcal{G}&=&\min^3_{i=1}{\mathcal{G}_i}=4\sqrt{2}C^a_{min}>0 \quad \text{iff} \quad C_{12}C_{23}>0,  \nonumber
\ea
where
\ba
\mathcal{G}_1&=&\sqrt{2}|\Big||1+C_{12}+C_{13}+C_{23}| 
-|1+C_{12}-C_{13}-C_{23}|\Big|\nonumber\\
&&-\Big||1-C_{12}-C_{13}+C_{23}| 
-|1-C_{12}+C_{13}-C_{23}|\Big||, \nonumber
\ea
and $\mathcal{G}_2$ and $\mathcal{G}_3$ are obtained by permuting the four $\mathcal{S}_{\alpha\beta\gamma}$ in $\mathcal{G}_1$.
The correlations can be decomposed in a convex mixture of a Svetlichny-box and a Svetlichny-local box which has $\mathcal{G}=0$ as follows,
\be
P=\frac{C^a_{min}}{\sqrt{2}}P^{0100}_{Sv}+\left(1-\frac{C^a_{min}}{\sqrt{2}}\right)P^{\mathcal{G}=0}_{SvL}. \label{WclassSD}
\ee
The bipartite marginals of these correlations have $\mathcal{G}_{12}=2\sqrt{2C^2_{12}}$,
$\mathcal{Q}_{13}=2C_{13}$ and $\mathcal{G}_{23}=2\sqrt{2C^2_{23}}$.
The correlations do not violate a Svetlichny inequality when $C_{12}+C_{13}+C_{23}\le2\sqrt{2}-1$, however, 
Svetlichny discord is nonzero whenever the state is genuinely entangled. The Svetlichny-local box in Eq. (\ref{WclassSD}) must have  
a decomposition which has the class $8$ extremal box as the correlations also violate a class $99$ facet inequality of $\mathcal{L}_2$ when 
${C_{13}+\frac{1}{\sqrt{2}}\left(C_{12}+C_{23}\right)}>3-\sqrt{2}$. Therefore, the three-way nonlocal correlations arising from the W-class states 
lie outside the Svetlichny-box polytope.
\begin{observation}
When the W-class states give rise to Svetlichny discord, two bipartite marginals have Bell discord, and they satisfy monogamy of Bell discord,
\be
\mathcal{G}_{ij}+\mathcal{G}_{ik}\le4.
\ee
\end{observation}
This tradeoff relation originates from monogamy of Bell nonlocality \cite{Toner} (see Appendix \ref{bellmermindisc}).
\subsubsection{Mermin discordant box}
Mermin discord for the W-class states is maximized by settings in Eq. (\ref{MDxz}) which gives rises to,
\ba
\mathcal{Q}&=&\min^3_{i=1}{\mathcal{Q}_i}=4C^a_{min}>0 \quad \text{iff} \quad C_{12}C_{23}>0, \nonumber
\ea
where
\ba
\mathcal{Q}_1&=&|\Big||1+C_{12}+C_{13}+C_{23}|
-|1+C_{12}-C_{13}-C_{23}|\Big|\nonumber\\
&&-\Big||1-C_{12}-C_{13}+C_{23}|
-|1-C_{12}+C_{13}-C_{23}|\Big||, \nonumber
\ea
and $\mathcal{Q}_2$ and $\mathcal{Q}_3$ are obtained by permuting the four $\mathcal{M}_{\alpha\beta\gamma}$ in $\mathcal{Q}_1$.
The correlations can be decomposed into a convex mixture of a tripartite Mermin-box and a Bell-local box which has $\mathcal{G}=\mathcal{Q}=0$,
\be
P=C^a_{min}\left[\frac{P^{0001}_{Sv}+P_{Sv}^{1111}}{2}\right]+\left(1-C^a_{min}\right)P^{\mathcal{Q}=0}_{L}. \label{WclassMD}
\ee
The bipartite marginals of these correlations have $\mathcal{Q}_{12}=2C_{12}$, $\mathcal{Q}_{13}=2C_{13}$ and $\mathcal{Q}_{23}=2C_{23}$.
The correlations are genuinely two-way nonlocal if $C_{12}+C_{13}+C_{23}>1$, however, they have nonzero tripartite Mermin discord if the state is genuinely entangled.
Thus, nonzero Mermin discord of the local correlations in Eq. (\ref{WclassMD}) originates from three-way contextuality.
\begin{observation}
When the correlations arising from the W-class states have tripartite Mermin discord, at least two bipartite marginals have Mermin discord, and they satisfy
monogamy of Mermin discord,
\be
\mathcal{Q}_{ij}+\mathcal{Q}_{ik}\le2,
\ee
\end{observation}
As this tradeoff originates from the monogamy of Mermin-box in three-qubit systems (see Appendix \ref{bellmermindisc}),
it includes monogamy of EPR-steering \cite{Reid}.
\subsection{Mixture of GHZ state with white noise}
Here we study the correlations arising from the following Werner states \cite{Werner},
\be
\rho_W=p \ketbra{\psi_{GHZ}}{\psi_{GHZ}}+ (1-p) \frac{\openone}{4},
\ee
where $\ket{\psi_{GHZ}}=\frac{1}{\sqrt{2}}(\ket{000}+\ket{111})$.
The Werner states are separable iff $p\le0.2$, biseparable iff $0.2<p\le0.429$ and  genuinely entangled iff $p>0.429$ \cite{Guhne}. Notice that these Werner states have
the component of the irreducible GHZ state, $p$, even if the state is separable. We show that the Werner states can give rise to Svetlichny/Mermin discord if $p>0$. 
Thus, the separable and biseparable states that have an irreducible genuinely entangled state component are genuinely nonclassical states as they can give
rise to Svetlichny/Mermin discord.  
\subsubsection{Svetlichny discordant box}
For the settings in Eq. (\ref{SDxy}), the Werner states give rise to the isotropic Svetlichny-box,
\be
P=\frac{p}{\sqrt{2}}P^{0000}_{Sv}+\left(1-\frac{p}{\sqrt{2}}\right)P_N.
\ee
These correlations admit the local deterministic model if $p\le\frac{1}{\sqrt{2}}$ and have Svetlichny discord $\mathcal{G}=4p\sqrt{2}$. 
Due to the component of the irreducible GHZ state and the incompatible measurements, the local correlations arising from the Werner states have genuine nonclassicality
originating from Svetlichny nonlocality if $p>0$.   
\subsubsection{Mermin discordant box}
For the settings in Eq. (\ref{MDxy}) which gives maximal Mermin discord for the GHZ-class states, the Werner states give rise to the isotropic Mermin-box, 
\be
P=p\left(\frac{P^{0000}_{Sv}+P^{1110}_{Sv}}{2}\right)+(1-p)P_N.
\ee
These correlations have Mermin discord $\mathcal{Q}=4p>0$ whenever the state has the irreducible GHZ state component. The correlations do not violate a 
Mermin inequality if $p\le\frac{1}{2}$, however, they have genuine nonclassicality originating from three-way contextuality if $p>0$.
\subsection{Biseparable W class state}
Consider the following biseparable state,
\be
\rho=\frac{1}{3}\ket{\psi^{AB}_{bi}}\bra{\psi^{AB}_{bi}}+\frac{1}{3}\ket{\psi^{AC}_{bi}}\bra{\psi^{AC}_{bi}}+\frac{1}{3}\ket{\psi^{BC}_{bi}}\bra{\psi^{BC}_{bi}},
\ee
$\ket{\psi^{AB}_{bi}}=\frac{1}{\sqrt{2}}(\ket{100}+\ket{010})$, $\ket{\psi^{AC}_{bi}}=\frac{1}{\sqrt{2}}(\ket{100}+\ket{001})$ and 
$\ket{\psi^{BC}_{bi}}=\frac{1}{\sqrt{2}}(\ket{010}+\ket{001})$.
Svetlichny/Mermin discord for the above biseparable state can be achieved only for the suitable settings that lie in the $xz$-plane, for instance, 
the settings given in Eq. (\ref{SDxz}) gives rise to Svetlichny discord $\mathcal{G}=\frac{4\sqrt{2}}{3}$. The correlation can be decomposed as follows,
\ba
P&=&\frac{1}{3}\left[\frac{1}{\sqrt{2}}P^{011}_{PR}+\left(1-\frac{1}{\sqrt{2}}\right) P^{AB}_N \right] P_C+\frac{1}{3}\left(\frac{P^{001}_{PR}+P^{111}_{PR}}{2}\right)P_B\nonumber \\
&&+\frac{1}{3}P_A\left[\frac{1}{\sqrt{2}}P^{101}_{PR}+\left(1-\frac{1}{\sqrt{2}}\right) P_N \right], 
\ea
where $P_A=P(a_m|A_i)$, $P_B=P(b_n|B_j)$ and $P_C=P(c_o|C_k)$ are the distributions arising from the state 
$\ket{0}$. Notice that the correlation arising from this state does not have Svetlichny/Mermin discord for all the settings that lie in the $xy$-plane 
as the state belongs to biseparable W class i.e., 
the state can be written as a convex mixture of an irreducible genuinely entangled state that belongs to the W-class and a state which cannot give rise to
Svetlichny/Mermin discord. 
\subsection{Mixture of GHZ state and W state}
Consider the correlations arising from the following states, 
\be
\rho=p\ket{\psi_{GHZ}}\bra{\psi_{GHZ}}+q\ket{\psi_{W}}\bra{\psi_{W}}.
\ee
where $\ket{\psi_W}=\frac{1}{\sqrt{3}}(\ket{100}+\ket{010}+\ket{001})$.
Since the optimal settings that gives maximal Svetlichny/Mermin discord for the GHZ state does not give nonzero Svetlichny/Mermin discord for the W-state and vice versa, 
Svetlichny/Mermin discord for these states
arise from the component of the GHZ state or the W state for the four settings given in Eqs. (\ref{SDxy})-(\ref{MDxz}). 

For the settings in Eq. (\ref{SDxy}), the correlations have Svetlichny discord $\mathcal{G}=4\sqrt{2}p$ and admit the following decomposition,
\be
P=p\left[\frac{1}{\sqrt{2}}P^{0000}_{Sv}+\left(1-\frac{1}{\sqrt{2}}\right)P_N\right]+qP^{\mathcal{G}=0}_{SvL}, \label{GNLGHZSb}
\ee
where $P^{\mathcal{G}=0}_{SvL}$ is a Svetlichny-local box arising from the W state which has zero Svetlichny discord.

For the settings in Eq. (\ref{SDxz}), the correlations have Svetlichny discord $\mathcal{G}=\frac{8\sqrt{2}q}{3}$ and admit the following decomposition,
\be
P=pP^{\mathcal{G}=0}_{NL}(\psi_{GHZ})+qP_{NL}^{\mathcal{G}>0}(\psi_{W}), \label{GNLGHZW}
\ee
where $P^{\mathcal{G}=0}_{NL}(\psi_{GHZ})$ is the three-way nonlocal box arising from the GHZ state given in Eq. (\ref{GHZopGNL}) and 
$P_{NL}^{\mathcal{G}>0}(\psi_{W})$ is the three-way nonlocal box arising 
from the W state given in Eq. (\ref{WclassSD}) with $C^a_{min}=\frac{2}{3}$. 

As the correlations in Eq. (\ref{GNLGHZW}) violates the class $99$ facet inequality, they do not belong to the Svetlichny-box polytope. 
However, the correlations in Eq. (\ref{GNLGHZSb})
belong to the three-way nonlocal region of the Svetlichny-box polytope.

\subsection{Classical-quantum, quantum-classical and genuinely quantum-correlated states}
A mixed three-qubit state can give rise to Svetlichny discord or Mermin discord iff all the three qubits are nonclassically correlated.
The states that do not have Svetlichny discord and Mermin discord can be decomposed in the form of classical-quantum or quantum-classical states defined as follows.
\begin{definition}
The classical-quantum (CQ) states can be decomposed as,
\be
\rho^{1|23}_{CQ}=\sum_ip_i \rho^A_i \otimes \rho^{BC}_i, \label{BC}
\ee
whereas the quantum-classical (QC) states can be decomposed as,
\be
\rho^{12|3}_{QC}=\sum_ip_i \rho^{AB}_i \otimes \rho^{C}_i \label{AB}
\ee
or
\be
\rho^{13|2}_{QC}=\sum_ip_i \rho^{AC}_i \otimes \rho^{B}_i, \label{AC}
\ee
where $\rho^{AB}_i$, $\rho^{AC}_i$, and $\rho^{BC}_i$ are, in general, quantum-correlated states
which are neither classical-quantum nor quantum-classical states \cite{Dakicetal} and there is no restriction on $\rho^A_i$, $\rho^B_i$, and $\rho^C_i$. 
\end{definition}
 
\begin{theorem}
All CQ and QC states given in Eqs. (\ref{BC})-(\ref{AC}) have $\mathcal{G}=\mathcal{Q}=0$ for all measurements. 
\end{theorem}
\begin{proof}
Consider the QC states as given in Eq. (\ref{AB}). For these states, the expectation value factorizes as follows,
\be
\braket{A_iB_jC_k}=\sum_i p_i \braket{A_iB_j}_i\braket{C_k}_i, \label{exfact}
\ee
which implies that the Svetlichny operators in $\mathcal{G}_1$ factorize as follows,
\begin{align}
\mathcal{G}_1&\!\!=\!\!|\Big|\!\left|\sum_ip_i \!\left\{\mathcal{B}^i_{000}\braket{C_0}_i\!+\mathcal{B}^i_{111}\braket{C_1}_i\right\}\right|\!-\!
\left|\sum_ip_i\!\left\{\mathcal{B}^i_{000}\braket{C_0}_i\!-\mathcal{B}^i_{111}\braket{C_1}_i\right\}\right|\Big|\nonumber\\
&\!\!-\Big|\left|\sum_ip_i\left\{\mathcal{B}^i_{000}\braket{C_1}_i\!+\mathcal{B}^i_{111}\braket{C_0}_i\right\}\right|\!-\!
\left|\sum_ip_i\!\left\{\mathcal{B}^i_{000}\braket{C_1}_i\!-\mathcal{B}^i_{111}\braket{C_0}_i\right\}\right|\Big|\nonumber\\
&\!\!-|\Big|\left|\sum_ip_i\! \left\{\mathcal{B}^i_{010}\braket{C_0}_i\!+\mathcal{B}^i_{100}\braket{C_1}_i\right\}\right|\!-\!
\left|\sum_ip_i\!\left\{\mathcal{B}^i_{010}\braket{C_0}_i\!-\mathcal{B}^i_{100}\braket{C_1}_i\right\}\right|\Big|\nonumber\\
&\!\!-\Big|\left|\sum_ip_i\!\left\{\mathcal{B}^i_{010}\braket{C_0}_i\!+\mathcal{B}^i_{100}\braket{C_1}_i\right\}\right|\!-\!
\left|\sum_ip_i\!\left\{\mathcal{B}^i_{010}\braket{C_1}_i\!-\mathcal{B}^i_{100}\braket{C_0}_i\right\}\right|\!\Big||. \label{ExQC}
\end{align} 
Here $\mathcal{B}^i_{\alpha\beta\gamma}$ which are the Bell functions in the CHSH inequalities in Eq. (\ref{BCHSH})
and $\braket{C_k}_i$ are evaluated for $\rho^i_{AB}$ and $\rho^i_C$ 
given in Eq. (\ref{AB}). 
Let us now try to maximize $\mathcal{G}_1$ with respect to the quantum-classical states in which $\rho^i_{AB}$ are the quantum-correlated states. For an optimal settings 
that gives nonzero for only one of $\mathcal{B}^i_{\alpha\beta\gamma}$ in Eq. (\ref{ExQC}), 
$\mathcal{G}_1=0$. Similarly, we can prove that $\mathcal{Q}=0$ by exploiting the factorization property in Eq. (\ref{exfact}). 

Since $\mathcal{G}$ and $\mathcal{Q}$ are symmetric under the permutations of the parties, they are also zero for the states in Eqs. (\ref{BC}) 
and (\ref{AC}) for all measurements.   
\end{proof}

All the genuinely entangled states are only a subset of the set of nonclassical states with respect to $\mathcal{G}$ and $\mathcal{Q}$.
The nonclassical biseparable and separable states are the genuinely quantum-correlated states. 
\begin{definition}
A genuinely quantum-correlated state cannot be written in the classical-quantum or quantum-classical form given in Eqs. (\ref{BC}) -(\ref{AC}) and  
admits the following decomposition,
\be
\rho\!=\!p_1\sum_iq_i \rho^A_i \otimes \rho^{BC}_i+p_2\sum_jq_j \rho^{AC}_j \otimes \rho^{B}_j+p_3\sum_kq_k \rho^{AB}_k \otimes \rho^{C}_k,
\ee
with atleast two of the three coefficients $p_1$, $p_2$, and 
$p_3$ are nonzero.  
\end{definition}
\subsection{Total correlations}
In Ref. \cite{Jebat}, a measure has been introduced to study the total correlations in a bipartite quantum joint probability distribution.
The tripartite generalization of this measure is defined as follows:
\begin{definition}
Total genuine correlations, $\mathcal{T}$, is defined as,
\be
\mathcal{T}:=\min \{\mathcal{T}_{12|3},\mathcal{T}_{13|2},\mathcal{T}_{1|23}\},
\ee
where
\be
\mathcal{T}_{12|3}=\max_{\alpha\beta\gamma}|\mathcal{S}_{\alpha\beta\gamma}-\mathcal{S}^{12|3}_{\alpha\beta\gamma}|, \nonumber
\ee
here,
\be
\mathcal{S}^{12|3}_{\alpha\beta\gamma}
=|\sum_{ijk}(-1)^{i\cdot j \oplus i\cdot k \oplus j\cdot k \oplus \alpha i\oplus \beta j \oplus \gamma k }\braket{A_iB_j}\braket{C_k}|,\nonumber
\ee
and where $\mathcal{T}_{13|2}$ and $\mathcal{T}_{1|23}$ are similarly defined.
\end{definition}
$\mathcal{T}$ is defined such that it satisfies the following properties:
(i) $\mathcal{T}\ge0$,
(ii) $\mathcal{T}=0$ iff the JPD can be written in the product form $P=P(a_m|A_i)P(b_n,c_o|B_j,C_k)$ and the permutations, and
(iii) $\mathcal{T}$ is invariant under LRO and symmetric under permutations of the parties. $\mathcal{T}$ is analogous to the measure for
total genuine tripartite correlations
defined in \cite{GTC} as both the measures vanish for the product states that can be written as $\rho=\rho_A \otimes \rho_{BC}$ and the permutations.
\begin{observation}
As a consequence of these three properties,
$\mathcal{T}$ gives rise to the additivity relation (see Appendix \ref{addrelation}),
\be
\mathcal{T}=\mathcal{G}+\mathcal{Q}\pm\mathcal{C}
\ee
for quantum correlations in the Svetlichny-box polytope. Here $\mathcal{C}$ quantifies genuinely classical correlations
and the negative sign is observed for pure genuinely entangled states.
\end{observation}

\subsubsection{Total correlations in the 3-decomposition of the GHZ state}
EPR2 \cite{EPR2} showed that each pair in an ensemble of two-qubits in the singlet state exhibits nonlocality if the ensemble maximally violates a Bell-CHSH inequality.
Then, for nonmaximal violation by the nonmaximally entangled states, EPR2 showed that only certain fraction of the ensemble behaves nonlocally and the remaining fraction behaves locally. EPR2 conjecture for multi-qubit systems implies that when an ensemble of three qubits in the GHZ state gives rise to the maximal violation of a Svetlichny inequality, each trio in the ensemble behaves nonlocally. Consider the correlations arising from the GHZ state given in Eq. (\ref{SMDghz}).
The correlation violates the Svetlichny inequality if $p\ne1$ and gives maximal violation when $p=\frac{1}{2}$.
Since the violation of the Svetlichny inequality decreases if $p$ is increased from $\frac{1}{2}$ to $1$, the number of trios exhibiting nonlocality decreases and goes to zero when $p=1$. However, the correlation gives rise to the GHZ paradox when $p=1$ which  implies that each trio in the ensemble behaves contextually \cite{GHZ,UNLH,Canasetal}.
If $p$ is decreased from $1$ to $\frac{1}{2}$, the number of trios behaving contextually will decrease and the number of trios behaving nonlocally will increase as the violation of the Mermin inequality that detects the GHZ paradox decreases and the violation of the Svetlichny inequality increases.
The correlations in Eq. (\ref{SMDghz}) can be written as a mixture of the three-way nonlocal box that violates the Svetlichny inequality to its quantum bound,
the three-way contextual box which exhibits the GHZ paradox and white noise,
\ba
P\!&=&\!\sqrt{2}\mathcal{G}'\left[\frac{1}{\sqrt{2}}P^{0000}_{Sv}+\left(1-\frac{1}{\sqrt{2}}\right)P_N\right]\!+\!\mathcal{Q}'\left(\frac{P^{0000}_{Sv}+P^{1111}_{Sv}}{2}\right)\nonumber\\
\!&&+\!\left(1-\sqrt{2}\mathcal{G}'-\mathcal{Q}'\right)P_N.
\ea
Therefore, the fractions $\sqrt{2}\mathcal{G}'$ and $\mathcal{Q}'$ of the total ensemble exhibits nonlocality and contextuality (GHZ paradox) and the remaining
fraction behaves as white noise when $\frac{1}{2}< p<1$.
The total correlations in Eq. (\ref{SMDghz}) is given by,
\be
\mathcal{T}=4\left(\sqrt{p}+\sqrt{1-p}\right)=\mathcal{G}+\mathcal{Q}=\left\{\begin{array}{lr}
\mathcal{G} \quad \text{when} \quad p=\frac{1}{2}\\
\mathcal{Q} \quad \text{when} \quad p=1\\
\end{array}.
\right.
\ee
which is the sum of Svetlichny discord and Mermin discord. Thus, $\mathcal{G}$ and $\mathcal{Q}$ separates the total amount of nonclassical correlations in the JPDs
into nonlocality and contextuality.
\section{Conclusions}\label{conc}
We have introduced the measures, Svetlichny and Mermin discord, to characterize tripartite quantum correlations in the context of the Svetlichny-box polytope.
We have obtained the $3$-decomposition of any correlation in the Svetlichny-box polytope into Svetlichny-box, a maximally two-way nonlocal box that exhibits 
three-way contextuality and a box with Svetlichny and Mermin discord equal to zero.
We have defined the two types of Mermin boxes that are  three-way contextual and extremal with respect to the $3$-decomposition.  
We find that the Svetlichny-box polytope does not characterize all genuinely three-way nonlocal
quantum correlations.

Svetlichny discord and Mermin discord quantify three-way nonlocality and three-way contextuality of quantum correlations 
with respect to the $3$-decomposition even if the correlations do not violate a Svetlichny inequality or a Mermin inequality.
In the case of pure states, Svetlichny and Mermin discord can be nonzero iff the state is genuinely entangled.
Moving to the mixed states, Svetlichny/Mermin discord detects the component of the irreducible genuinely entangled state.
If a mixed state has an irreducible GHZ-class state and an irreducible W-class state components simultaneously, 
nonzero Svetlichny/Mermin discord originates from the GHZ-class state or the W-class state.   
We find that when GGHZ states and Werner states give rise optimal Svetlichny or Mermin discord,
irreducible GHZ state component in the Werner states plays a role analogous to entanglement in the GGHZ states.

\section{Appendix}

\subsection{An example to illustrate the notion of irreducible Svetlichny-box in unequal mixture of the Svetlichny-boxes}\label{irreducibleSv}
Notice that the subtraction done in $\mathcal{G}_i$ given in Eq. (\ref{GBD}) serves to calculate the amount of single Svetlichny-box excess in the unequal mixture of the Svetlichny-boxes.
Nonzero $\mathcal{G}_i$ does not necessarily imply that the correlation has an irreducible Svetlichny-box component which can be illustrated
by the following correlation,
\be
P=0.4P^{0000}_{Sv}+0.3P^{0010}_{Sv}+0.2P^{1000}_{Sv}+0.1P^{0110}_{Sv},
\ee
which has $\mathcal{G}_1=1.6$, however, other $\mathcal{G}_i$ are zero. Nonzero $\mathcal{G}_1$ for this correlation implies that it can be written as a convex mixture
of a single Svetlichny-box and a local box,
\be
P=\mathcal{G}'P^{0000}_{Sv}+(1-\mathcal{G}')P_L, \label{demosingle}
\ee
where $\mathcal{G}'=0.2$ and $P_L=\frac{1}{8}P^{0000}_{PR}+\frac{1}{2}P^{0100}_{PR}+\frac{1}{4}P^{1000}_{PR}+\frac{1}{8}P^{0110}_{PR}$.
The single Svetlichny-box component in this decomposition is not irreducible as $\mathcal{G}'$ vanishes for other possible decompositions.
Thus, minimizing the single Svetlichny-box component overall possible decompositions in Eq. (\ref{st2}) corresponds to the minimization in Eq. (\ref{GBD}) as $\mathcal{G}$
is intended to detect irreducible Svetlichny-box component.
\subsection{Svetlichny function monogamy}\label{msivio}
The fact that the violation of a Svetlichny inequality is monogamous, i.e.,
a Svetlichny nonlocal correlation cannot violate more than a Svetlichny inequality in Eq. (\ref{SI})
leads to the following Svetlichny function monogamy.
\begin{proposition}
For any given correlation $P(a_m,b_n,c_o|A_i,B_j,C_k)$, the Svetlichny functions,
\be
\mathcal{S}_{\alpha\beta\gamma}=\left|\sum_{ijk}(-1)^{i\cdot j \oplus i\cdot k \oplus j\cdot k \oplus \alpha i\oplus \beta j \oplus \gamma k }\braket{A_iB_jC_k}\right|, \label{mSf}
\ee
satisfy the monogamy relationship,
\be
\mathcal{S}_{i}+\mathcal{S}_j\le8 \quad \forall i,j, \label{SFm}
\ee
where $\mathcal{S}_{i}$ and $\mathcal{S}_j$ are any two of the Svetlichny functions defined in Eq. (\ref{mSf}).
\end{proposition}
\begin{proof}
Since the correlations in the two-way local polytope satisfy the complete set of Svetlichny inequalities, they satisfy the trade-off relations in Eq. (\ref{SFm}).
All the Svetlichny-boxes satisfy the trade-off relations in Eq. (\ref{SFm}), since only one of the Svetlichny functions attains
the algebraic maximum and the rest of them are zero for any Svetlichny box. Any nonextremal Svetlichny-nonlocal box in $\mathcal{R}$can be written as a convex mixture of a Svetlichny-box and a Svetlichny-local box that gives the local bound of $4$ for a Svetlichny inequality (see fig. \ref{NS3dfig1}),
\be
P=pP^{\alpha\beta\gamma\epsilon}_{Sv}+(1-p)P_{SvL}.
\ee
Now consider the Svetlichny-nonlocal correlations that maximize the left-hand side of Eq. (\ref{SFm});
for instance, any convex mixture of the Svetlichny-box and the deterministic box, $P=pP^{0000}_{Sv}+(1-p)P^{0000}_D$,
gives $\mathcal{S}_{000}+\mathcal{S}_j=8$ $\forall j$.
\end{proof}
\subsection{The $\mathcal{G}=0$ and $\mathcal{Q}=0$ correlations}
\be
P^{\mathcal{G}=0}_{SvL}=\frac{C_{12}}{\sqrt{2}-\sqrt{\tau_3}}P^{000}_{PR}P^C_N
+\left(1-\frac{C_{12}}{\sqrt{2}-\sqrt{\tau_3}}\right)P^{AB}_NP(\rho_C)\label{G=01}.
\ee
Here $P(\rho_C)$ arises from the state, $\rho_C=a_0\ketbra{x_+}{x_+}+a_1\ketbra{x_-}{x_-}$, where
$a_i=\frac{1}{2}+(-1)^{i}\frac{\sqrt{2}(\sin^2\theta\sin\theta_3\cos\theta_3)}{\sqrt{2}-\sqrt{\tau_3}-C_{12}}$.

\ba
P^{\mathcal{G}=0}_{SvL}&=&\frac{C_{12}}{1-\mathcal{G}'}\left(\frac{P^{010}_{PR}+P^{100}_{PR}}{2}\right)P(\rho^1_C)\nonumber\\
&&+\frac{1}{1-\mathcal{G}'}\left(1-\frac{\tau_3}{\sqrt{C^2_{12}+2\tau_3}}-C_{12}\right)P^{AB}_NP(\rho^2_C). \label{G=02}
\ea
Here $P(\rho^1_C)$ and $P(\rho^2_C)$ arise from the states,
$\rho^1_C=a_0\ketbra{0}{0}+a_1\ketbra{1}{1}$ and
$\rho^2_C=b_0\ketbra{0}{0}+
\frac{\sin^2\theta\sin\theta_3\cos\theta_3}{1-\mathcal{G'}-C_{12}}\left(\ketbra{0}{1}+\ketbra{1}{0}\right)
+b_1\ketbra{1}{1},$
where $a_i=\frac{1}{2}\left[1+(-1)^{i}\frac{\cos\theta_3}{\sqrt{1+\sin^2\theta_3}}\right]$, $b_i=\frac{1}{2}\left[1+(-1)^i\frac{\sqrt{1+\sin^2\theta_3}\left(\cos^2\theta+\sin^2\theta\cos2\theta_3\right)-C_{12}\cos\theta_3}{\sqrt{1+\sin^2\theta_3}
\left(1-C_{12}-\mathcal{G'}\right)}\right]$ and $\mathcal{G'}=\frac{\tau_3}{\sqrt{C^2_{12}+2\tau_3}}$.

\be
P^{\mathcal{Q}=0}_{L}=\frac{\sqrt{\tau_{12}}}{1-\sqrt{\tau_3}}\left(\frac{P^{000}_{PR}+P^{110}_{PR}}{2}\right)P^C_N+\left(1-\frac{\sqrt{\tau_{12}}}{1-\sqrt{\tau_3}}\right)P_NP_C.\label{Q=01}
\ee
Here $P(\rho_C)$ is a distribution arising from the state $\rho_C=a_0\ketbra{x_+}{x_+}+a_1\ketbra{x_-}{x_-}$ where
$a_i=\frac{1}{2}+(-1)^i\frac{\sin^2\theta\sin\theta_3\cos\theta_3}{1-\sqrt{\tau_3}-C_{12}}$.

\be
P^{\mathcal{Q}=0}_{L}=\frac{\mathcal{G}'_{12}}{1-\mathcal{Q}'}P^{000}_{PR}P(\rho^1_C)
+\frac{1}{1-\mathcal{Q}'}\left(1-\mathcal{G}'_{12}-\mathcal{Q}'\right)P_NP(\rho^2_C). \label{Q=02}
\ee
Here $P(\rho^1_C)$ and $P(\rho^2_C)$ arise from the states
$\rho^1_C=a_0\ketbra{0}{0}+a_1\ketbra{1}{1}$
and
$\rho^2_C=b_0\ketbra{0}{0}+
\frac{\sin^2\theta\sin\theta_3\cos\theta_3}{1-\mathcal{G}'_{12}-\mathcal{Q}'}\left(\ketbra{0}{1}+\ketbra{1}{0}\right)
+b_1\ketbra{1}{1}$, where $a_i=\frac{1}{2}\left[1+(-1)^i\frac{\cos\theta_3}{\sqrt{1+\sin^2\theta_3}}\right]$, $b_i=\frac{1}{2}\left[1+(-1)^i\frac{\sqrt{1+\sin^2\theta_3}\left(\cos^2\theta+\sin^2\theta\cos2\theta_3\right)-\mathcal{G}'_{12}\cos\theta_3}{\sqrt{1+\sin^2\theta_3}
\left(1-\mathcal{G}'_{12}-\mathcal{Q}'\right)}\right]$, $\mathcal{Q'}=\mathcal{Q}/4$ and $\mathcal{G}'_{12}=\mathcal{G}/4$.
\subsection{Proof for Monogamy of Bell discord and monogamy of Mermin discord}\label{bellmermindisc}
In the tripartite correlation scenario, Bell discord of subsystems $AB$ and $AC$ are constrained by the monogamy,
\be
\mathcal{G}_{12}+\mathcal{G}_{13}\le4. \label{mBD}
\ee
\begin{proof}
As nonzero Bell discord requires an irreducible PR-box component, $\mathcal{G}_{12}$ and $\mathcal{G}_{13}$ are simultaneously
nonzero if both the bipartite marginals have an irreducible PR-box component. Suppose parties $A$ and $B$ share a PR-box, then the third party is uncorrelated \cite{MAG06}.
The only possible way
for the joint parties, $AB$ and $AC$ share a PR-box simultaneously and maximize the left-hand side in Eq. (\ref{mBD}) is that they share the correlation
given by the convex mixture,
\be
P=pP^{AB}_{PR}P_C+qP^{AC}_{PR}P_B.
\ee
For this correlation, $\mathcal{G}_{12}+\mathcal{G}_{13}=4$.
\end{proof}

In a three-qubit system, Mermin discord arising from the bipartite systems $AB$ and $AC$ are constrained by the monogamy,
\be
\mathcal{Q}_{12}+\mathcal{Q}_{13}\le2. \label{mMD}
\ee
\begin{proof}
In a two-qubit system, a pure Mermin-box arises iff the parties share a maximally entangled state \cite{Jeba}. Suppose subsystem $AB$ of a three-qubit system gives rise to
a Mermin-box, a third party cannot share a Mermin-box due to the monogamy of entanglement \cite{CKW}. Thus, the only possible way
for the joint parties, $AB$ and $AC$ share a Mermin-box simultaneously and maximizes the left-hand side in Eq. (\ref{mMD}) is that the parties share
the correlation given by the convex mixture,
\be
P=pP^{AB}_{M}P_C+qP^{AC}_{M}P_B.
\ee
For this correlation, $\mathcal{Q}_{12}+\mathcal{Q}_{13}=2$.
\end{proof}
\subsection{Proof for the additivity relation}\label{addrelation}
The decomposition given in Eq. (\ref{cdQC}) implies that up to local unitary operations any quantum correlation arising from a three-qubit state
has the following $3$-decomposition,
\be
P=\mathcal{G}'P^{0000}_{Sv}+\mathcal{Q}'\left(\frac{P^{0000}_{Sv}+P^{111\gamma}_{Sv}}{2}\right)+(1-\mathcal{G}'-\mathcal{Q}')P^{\mathcal{G}=0}_{\mathcal{Q}=0},
\ee
where $\frac{1}{2}\left(P^{0000}_{Sv}+P^{111\gamma}_{Sv}\right)$ are the two Mermin-boxes canonical to the Svetlichny-box $P^{0000}_{Sv}$. Since this correlation maximizes
$\mathcal{S}_{000}$,
\ba
\mathcal{T}(P)&=&|\mathcal{S}_{000}(P)-\max\{\mathcal{S}^{12|3}_{000}(P),\mathcal{S}^{13|2}_{000}(P),\mathcal{S}^{1|23}_{000}(P)\}|\nonumber\\
&=&|8\mathcal{G}'+4\mathcal{Q}'+(1-\mathcal{G}'-\mathcal{Q}')[\mathcal{S}_{000}(P^{\mathcal{G}=0}_{\mathcal{Q}=0})\nonumber\\
&-&\max\{\mathcal{S}^{12|3}_{000}(P^{\mathcal{G}=0}_{\mathcal{Q}=0}),\mathcal{S}^{13|2}_{000}(P^{\mathcal{G}=0}_{\mathcal{Q}=0}),\mathcal{S}^{1|23}_{000}(P^{\mathcal{G}=0}_{\mathcal{Q}=0})\}]|\nonumber\\
&=&\mathcal{G}+\mathcal{Q}\pm\mathcal{C},
\ea
where
\ba
\mathcal{C}&=&(1-\mathcal{G}'-\mathcal{Q}')|\mathcal{S}_{000}(P^{\mathcal{G}=0}_{\mathcal{Q}=0})\nonumber\\
&-&\max\{\mathcal{S}^{12|3}_{000}(P^{\mathcal{G}=0}_{\mathcal{Q}=0}),\mathcal{S}^{13|2}_{000}(P^{\mathcal{G}=0}_{\mathcal{Q}=0}),\mathcal{S}^{1|23}_{000}(P^{\mathcal{G}=0}_{\mathcal{Q}=0})\}|.
\ea 


\chapter{Discussion} \label{chap:conclusion}
We have defined the two measures, Bell discord and Mermin discord, which are also nonzero for
boxes admitting local hidden variable model.
By using these measures, we have characterized nonclassicality of bipartite qubit correlations within the framework of generalized nonsignaling theories. For the bipartite nonsignaling boxes, we have obtained a canonical decomposition which is expressed as a convex combination of three boxes. In this decomposition, the presence of nonclassicality is manifested in three different ways: when only the fraction of PR box (which exhibits
nonlocality) is nonzero,
or only the fraction of Mermin box (which exhibits EPR steering) is nonzero, or both the PR box and Mermin box
fractions are nonzero. Bell and Mermin discords serve us to quantify the PR box fraction and Mermin box fraction, respectively, in the canonical decomposition. We have shown that in the case of boxes arising from two-qubit
states, both nonzero left and right quantum discords are necessary for nonzero Bell/Mermin discord.
In this case, nonzero Bell and Mermin discords originate from noncommuting measurements that give rise to Bell nonlocality and EPR steering (without Bell nonlocality), respectively.

We have generalized Bell and Mermin discords to the tripartite case to characterize
genuine nonclassicality of tripartite qubit correlations. We have obtained
a three-way decomposition for the tripartite nonsignaling boxes, which generalizes
the bipartite canonical decomposition. In this decomposition,
the presence of genuine nonclassicality is manifested in three different ways: when only the fraction of Svetlichny box (which exhibits genuine nonlocality) is nonzero,
or only the fraction of Mermin box (which exhibits three-way contextuality) is nonzero, or both the Svetlichny box and Mermin box fractions are nonzero.
The measures, Svetlichny and Mermin discords, serve us to quantify the Svetlichny-box and Mermin-box components,
respectively, in the three-way decomposition.
In the multipartite case, genuine quantum discord quantifies quantum correlation that is shared among
all the subsystems of the multipartite system.
We have demonstrated that if a box, having any of the tripartite Svetlichny/Mermin discord nonzero,
arises from a three-qubit state then presence of genuine tripartite quantum
discord is guaranteed, even when the box has a local hidden variable description.

In this thesis, we have restricted ourselves to the nonsignaling boxes with two binary
inputs and two binary outputs.
It would be interesting to generalize Bell and Mermin discords to the scenario in which
the black boxes have more than two outputs for a given input. This would be useful
to characterize nonclassicality of quantum correlations arising from two-qudit
states.

The canonical decomposition of bipartite nonsignaling boxes suggests that any bipartite
quantum state can be decomposed in a convex mixture of a pure entangled state
and a separable state which is neither a classical-quantum state nor a
quantum-classical state. This decomposition would be relevant to quantifying
quantum correlation that goes beyond entanglement.

The tripartite Svetlichny and Mermin discords can be defined for
n-partite nonsignaling boxes with more than three parties by using
n-partite Svetlichny and Mermin inequalities. These quantities
may be useful for characterizing multipartite quantum
states.

Bell and Mermin discords may have implications for
characterizing intrinsic randomness of quantum correlations.
It may be interesting to relate Bell/Mermin discord to various measures of intrinsic
randomness such as observed randomness, device-independent randomness
and semi-device-independent randomness. 


\cleardoublepage 
\phantomsection 
\addcontentsline{toc}{chapter}{Bibliography}
\bibliographystyle{alphaarxiv.bst}
\bibliography{th}

\end{document}